\documentclass{sig-alternate}
\usepackage[T1]{fontenc}
\usepackage{ae,aecompl}
\usepackage{times}
\usepackage{multirow}
\usepackage{subfig}
\usepackage{graphicx}

\newcounter{subcopyrightbox@save}
\usepackage{caption}
\usepackage{color, url}
\usepackage{xspace} 
\usepackage[ruled,linesnumbered]{algorithm2e}
\usepackage{mathrsfs}
\usepackage{amssymb}
\usepackage{float}
\usepackage{tabularx}

\newtheorem{theorem}{Theorem}

\newcommand{\san}{SAN\xspace}
\newcommand{\sanplural}{SANs\xspace}
\newcommand{\idg}{I\xspace}

\newcommand{\Zhel}{\emph{Zhel}\xspace}

\newcommand{\sanl}{Social-Attribute Network\xspace}
\newcommand{\eat}[1]{}
\newcommand{\ling}[1]{{\footnotesize\color{red}[Ling: #1]}}

\newcommand{\PI}{Phase \textrm{I}\xspace}
\newcommand{\PII}{Phase \textrm{II}\xspace}
\newcommand{\PIII}{Phase \textrm{III}\xspace}
\newcommand{\Section}{\S}

\newcommand{\mypara}[1]{\smallskip\noindent{\bf {#1}:}~}
\newcommand{\myparatight}[1]{\smallskip\noindent{\bf {#1}:}~}

\newcommand{\tightcaption}[1]{\vspace{-0.2cm}\caption{\em #1}\vspace{-0.2cm}}
\newcounter{packednmbr}

\newenvironment{packeditemize}{\begin{list}{$\bullet$}{\setlength{\itemsep}{0.2pt}\addtolength{\labelwidth}{-4pt}\setlength{\leftmargin}{\labelwidth}\setlength{\listparindent}{\parindent}\setlength{\parsep}{1pt}\setlength{\topsep}{0pt}}}{\end{list}}

\begin{document}

\conferenceinfo{IMC'12,} {November 14--16, 2012, Boston, Massachusetts, USA.} 
\CopyrightYear{2012} 
\crdata{978-1-4503-1705-4/12/11} 
\clubpenalty=10000 
\widowpenalty = 10000

\title{Evolution of Social-Attribute Networks: \\ Measurements, Modeling, and Implications using Google+}

\numberofauthors{6}

\author{
\alignauthor Neil Zhenqiang Gong\\
       \affaddr{EECS, UC Berkeley}\\
       \email{neilz.gong@berkeley.edu}
\alignauthor Wenchang Xu\\
       \affaddr{CS, Tsinghua University}
       \email{wencxu@gmail.com}
\alignauthor Ling Huang\\
       \affaddr{Intel Labs}
       \email{ling.huang@intel.com}
\and
\alignauthor  Prateek Mittal\\
       \affaddr{EECS, UC Berkeley}
       \email{pmittal@eecs.berkeley.edu}
\alignauthor  Emil Stefanov\\
       \affaddr{EECS, UC Berkeley}
       \email{emil@berkeley.edu}
\alignauthor Vyas Sekar\\
       \affaddr{CS, Stony Brook University}
       \email{vyas@cs.stonybrook.edu}
\and
\alignauthor Dawn Song\\
       \affaddr{EECS, UC Berkeley}
       \email{dawnsong@cs.berkeley.edu}
}

\maketitle

\begin{abstract}

Understanding social network structure and evolution has important implications
for many aspects of network and system design including provisioning, 
bootstrapping trust and reputation systems via social
networks, and defenses against Sybil attacks. Several recent results suggest
that augmenting the social network structure with user attributes (e.g., location,
employer, communities of interest) can provide a more fine-grained understanding
of social networks. However, there have been few studies to 
 provide  a systematic understanding of these effects at scale. 

We bridge this gap using a unique dataset collected as the  Google+ social
network grew over time since its release in late June 2011. We observe novel
phenomena with respect to both standard social network metrics and  new
attribute-related metrics (that we define). We also observe interesting
evolutionary patterns as Google+ went from a bootstrap phase to  a steady
invitation-only stage before a public release. 


Based on our empirical observations, we develop a new generative
model to jointly reproduce the social structure and the node attributes. 
Using theoretical analysis and empirical evaluations, we show that our
model can accurately reproduce the social and attribute structure of real
social networks. We also demonstrate that our model provides more accurate
predictions for practical application contexts.

\end{abstract}


\setcounter{section}{0}
\section{Introduction}
Online social networks (e.g., Facebook, Google+, Twitter) have become 
increasingly important platforms for interacting with people, 
processing information and diffusing social influence.  
Thus understanding social-network structure and evolution has important implications
for many aspects of network and system design including 
bootstrapping reputation via social networks (e.g.,~\cite{ostra}), 
defenses against Sybil attacks (e.g.,~\cite{sybil}),
leveraging social networks for  search~\cite{socialsearch}, and recommender systems with social regularization~\cite{Ma11}.

Traditional social network studies have largely focused on
understanding the topological structure of the social network, where each
user can be viewed as a node and a specific relationship (e.g., friendship,
co-authorship) is represented by a link between two nodes.  More
recently, there has been growing interest in augmenting this social network with \emph{user attributes}, which we call as \emph{Social-Attribute Network} (SAN).  User attributes could be
\emph{static} (e.g., school, major, employer and city derived from user
profiles), or \emph{dynamic} (e.g., online interest and community groups). Recent studies have demonstrated 
 the promise of social-attribute networks  
 in applications such as link 
prediction~\cite{Yang11, Gong11}, 
attribute inference~\cite{Gong11, Yang11}, and 
community detection~\cite{Zhou09}.

Despite the growing importance of such social-attribute networks in social network
analysis applications, there have been few efforts at systematically measuring
and modeling the evolution of social-attribute networks. Most prior work in the measurement and
modeling space focuses primarily on the social structure~\cite{Ahn07,
Backstrom12, Dong09, Kumar06, Kwak10, Leskovec05, Mislove07}.  Measuring social-attribute networks can simultaneously inform us the properties of social network structure, attribute structure, and how such attributes impact social network structure. 

In this paper, we present a detailed study of the evolution of social-attribute networks using a
unique large-scale dataset collected by crawling the Google+ social network struture and its user profiles.
This dataset offers a unique opportunity for us as
we were fortunate to observe the complete evolution of the social network and
its growth to around 30 million users within a span of three months.

We observe novel patterns in the growth of the Google+ social-attribute network.  First,  we observe
that the social reciprocity of Google+ is lower than many traditional social networks
and is closer to that of Twitter.   Second, in contrast to many
prior networks, the social degree distributions in Google+ are best modeled by a
lognormal distribution. Third, we observe that assortativity of Google+ social network is neutral while many other social networks own positive assortativities. Fourth, we also see that the distinct phases
(initial launch, invite only, public release) in the timeline of
Google+  naturally manifest themselves in the social and attribute structures.  Fifth, for the generalized attribute metrics (that we define), while some attribute metrics mirror their social counterparts (e.g., diameter), several show distributions and trends that are significantly different (e.g., clustering coefficient, attribute degree). Finally, via the social-attribute network framework, we study the impact of user attributes on the social structure and observe that
nodes sharing common attributes are likely to have higher social reciprocity and that
some attributes have much stronger influence than others (e.g., Employer vs.\
City).

Based on our observations, we develop a new generative model for
\sanplural.  Our model includes two new components, i.e., \emph{attribute-augmented preferential attachment} and \emph{attribute-augmented triangle-closing}, which extend the classical preferential attachment~\cite{Barabasi99, Kumar00} and triangle-closing~\cite{Leskovec08,
Sala10, Toivonen09, VAZQUEZ03}, respectively.  Using both theoretical analysis and empirical
evaluation, we show that our model can reproduce \sanplural that accurately
reflect the true ones with respect to various network metrics and real-world applications.
Such a generative model has a lot of applications~\cite{Leskovec10-JMLR} such as network extrapolation and sampling, network visualization and compression, and network anonymization~\cite{Sala11}.


To summarize, the key contributions of this work are: 

\begin{packeditemize}

\item We perform the first study of the evolution of social-attribute networks using Google+. We observe novel phenomena in standard social structure metrics and new
attribute-related metrics (that we define) and how attributes impact the social structure.

\item We develop a measurement-driven generative model for the social-attribute
network that models  the impact of user attributes into the network
evolution. 

\item Using both theoretical analysis and empirical evaluation, we validate
that our model can accurately reproduce real social-attribute networks. 


\end{packeditemize}

\eat
{
\mypara{Roadmap} In the rest of the paper, we describe the Google+ dataset and
 its \san in \Section\ref{sec:data}. We analyze its social 
 graph and attribute graph in \Section\ref{sec:structure} and \Section\ref{sec:attribute}
 respectively. We design a new evolution model for \san 
 in \Section\ref{sec:model}. We validate our model using graph metrics and 
 two  application contexts in \Section\ref{sec:evaluation}. We discuss a few outstanding 
 issues in \Section\ref{sec:discussion} and review related work in 
 \Section\ref{sec:related}  before concluding in \Section\ref{sec:conclude}.
}

\section{Preliminaries and Dataset}
\label{sec:data}
 In this section, we begin with some background on augmenting social network structure 
 with attributes. Then, we  describe how we collected the Google+ data 
 and how we augment the Google+ social network with user attribute 
information. We also present some basic measurements describing the evolution of 
the Google+.

\subsection{\sanl (\san)}
In this section, we review the definition of \emph{\sanl} (SAN)~\cite{Gong11} and introduce the
basic notations  used in the rest of this paper. 

Given a directed social network $G$, in which nodes are users and edges represent friend relationships between users, and $M$ distinct binary attributes, which could be static (e.g., name of
employer, name of school, major, etc.) or dynamic (e.g., interest groups), a \san is an
augmented network with $M$ additional nodes where each such  node corresponds to
a specific binary attribute.   For each node $u$ in $G$ with attribute $a$, we
create an undirected link between $u$ and $a$ in the \san.

\begin{figure}[t]
\centering
\includegraphics[width=0.3\textwidth]{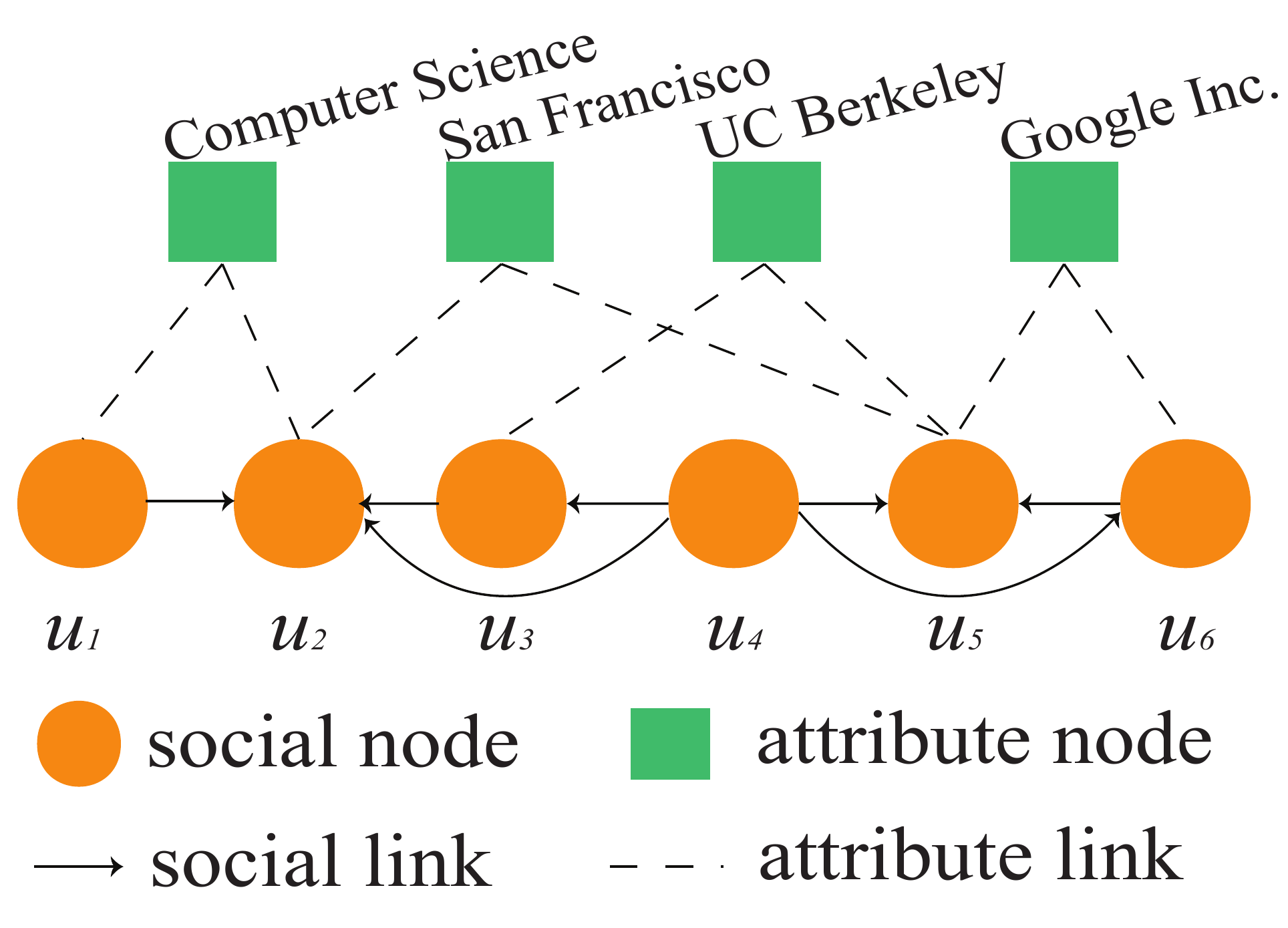}
\tightcaption{Illustration of a \san with six social nodes and four attribute nodes. Note that 
 the social links between users are directed whereas the attribute-user links are undirected.}
\label{fig:san}
\end{figure}

Nodes in a \san corresponding to nodes in $G$ are called \emph{social nodes}
and denoted as the set $V_s$, while nodes representing attributes are called
\emph{attribute nodes} and denoted as the set $V_a$. Figure~\ref{fig:san} shows an
example \san.  Links between social nodes are called
\emph{social links} and denoted as the set $E_s$, while links between social nodes and
attribute nodes are called \emph{attribute links} and denoted as the set $E_a$.
Thus a \sanl is denoted as $\san=(V_s, V_a, E_s, E_a)$.

For a given social or attribute node $u$ in a \san, we denote its
\emph{attribute neighbors} as  $\Gamma_a(u)=\{ v|v\in V_a, (u,v)\in E_a\}$,
\emph{social neighbors} as $\Gamma_{s}(u) = \{v|v\in V_s,(v, u) \in E_s\cup E_a
\ or \ (u, v) \in E_s\cup E_a\}$, \emph{social in neighbors} as
$\Gamma_{s,in}(u) = \{v|(v, u) \in E_s\}$ and  \emph{social out neighbors} as
$\Gamma_{s, out}(u) = \{v|(u, v) \in E_s\}$. Note  that an attribute node can only
 have  social neighbors.

\subsection{Google+ Data}
Google+ was launched with an invitation-only test phase on June 28, 2011,
 and opened to everyone 18 years of age or older on September 20, 2011. 
  We believed this was a tremendous opportunity to 
 observe the real-world evolution of a \emph{large-scale} social-attribute network.
 Thus, we began to crawl daily snapshots of public Google+ 
social network structure and user profiles; our crawls lasted from July 6 
to October 11, 2011. The first snapshot was crawled by breadth-first 
search (without early stopping). 
On subsequent days,  we expanded the social 
structure from the previous snapshot. For most snapshots, our crawl 
  finished within one day as Google did not limit the crawl 
rate during that time. 

We believe our crawl collected a large Weakly Connected Component (WCC) of
Google+. This may be  surprising as many past attempts on  Flickr, Facebook,
YouTube etc.,  were unable to do so~\cite{Mislove07}. The key difference is
that these were only able to access outgoing links.  In contrast,  each user
 in Google+ has both an outgoing list (i.e., ``in your
circles'') and an incoming list (i.e., ``have you in circles''). This allows us
to access both outgoing and incoming links making it feasible to crawl the
entire WCC.

We  have two points of reference that suggest our coverage is high  ($\geq$ 70\%):  1)
TechCrunch estimated the  number of Google+ users on July 12, 2011 is around 10
million~\cite{google1};  our crawled snapshot on the same day has  7 million
users. (2)  Google announced 40 million users had joined Google+ in middle October~\cite{google2};  our crawled snapshot on October 11 has around 30 million users.

We take each user $u$ in Google+ as a social node in \san, and connect it to
 her outgoing friends via outgoing links and incoming friends via incoming
links. We  use  four attribute types \emph{School},
\emph{Major}, \emph{Employer} and \emph{City} that were available and easy to extract. 
 Specifically, we find all
distinct schools, majors, employers and cities that appear in at least one user
profile and use them as attribute nodes. Recall that a social node $u$ is connected to
attribute node $a$ via an undirected link if $u$ has attribute $a$. 
In this way, we construct a \san from each crawled snapshot, resulting in 79
\sanplural during the period from July 6 to October 11, 2011.

\begin{figure}[t]
\vspace{-0.4cm}
\centering
\subfloat[\scriptsize{Social nodes}]{\includegraphics[width=0.25\textwidth, height=1.5in]{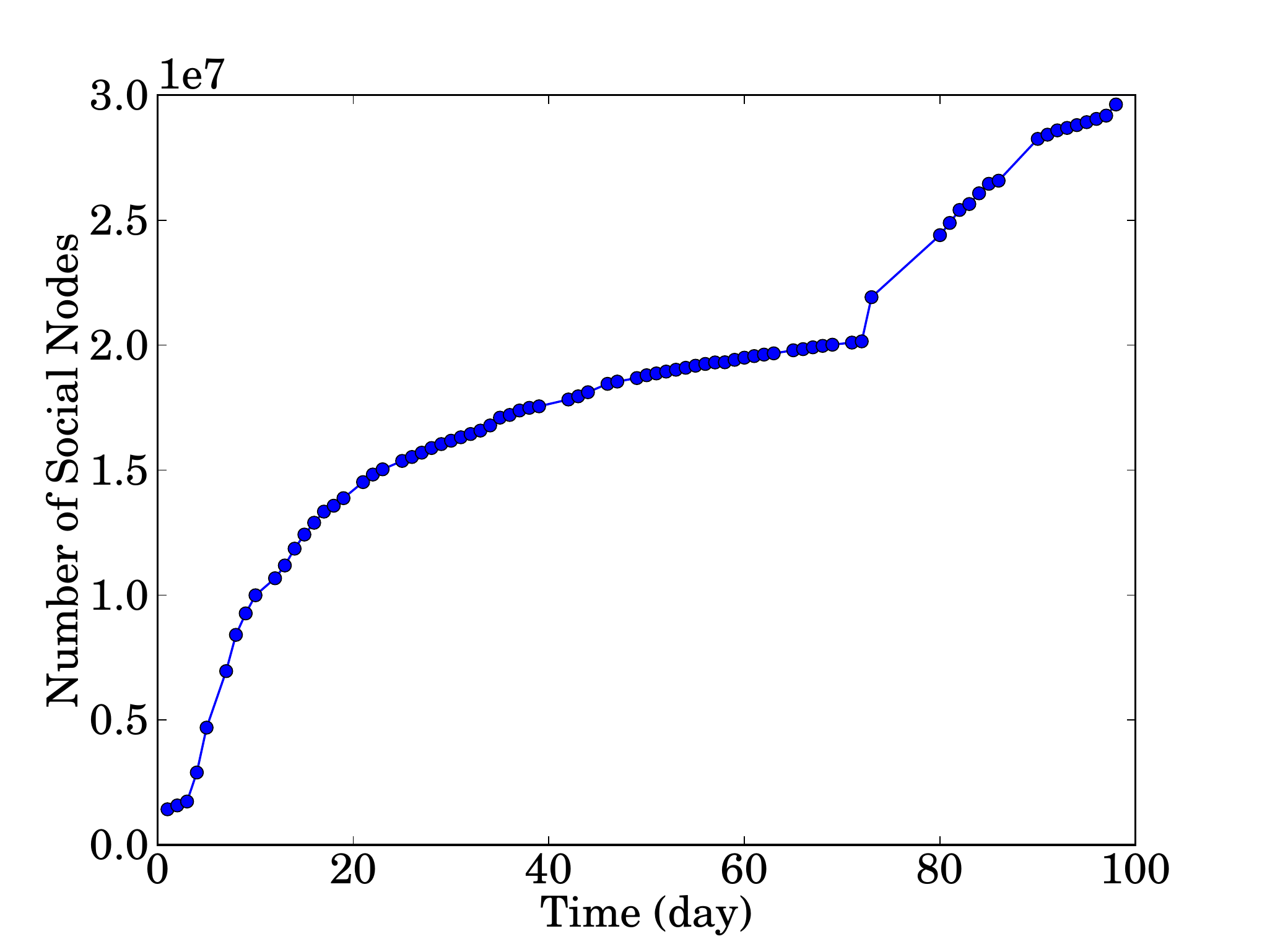}\label{soc-node}}
\subfloat[\scriptsize{Attribute nodes}]{\includegraphics[width=0.25\textwidth, height=1.5in]{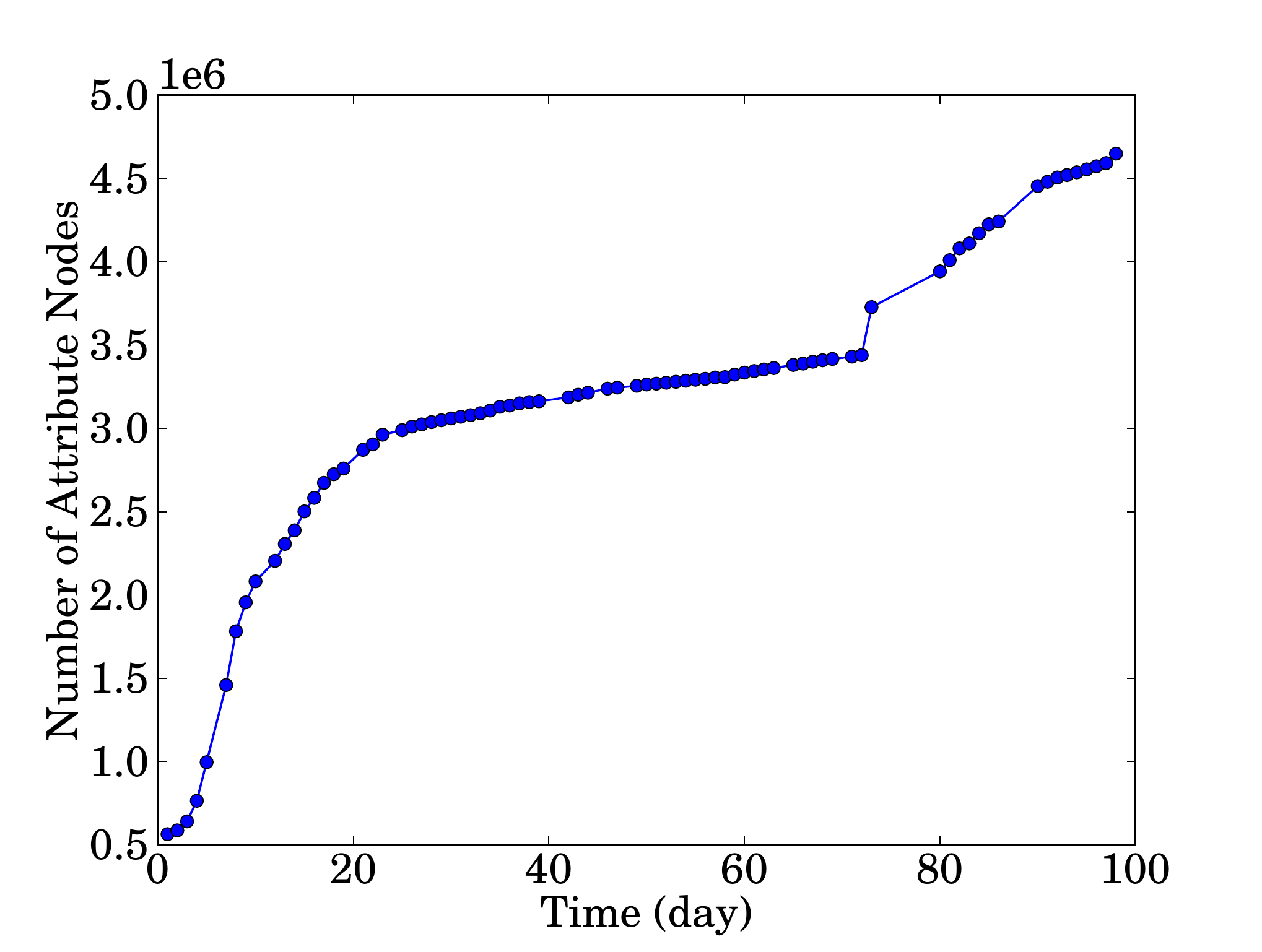}\label{attri-node}}
\tightcaption{Growth in the number of social and attribute nodes in the Google+ dataset.}
\label{fig:node}
\end{figure}

\begin{figure}[t]
\vspace{-0.4cm}
\centering
\subfloat[\scriptsize{Social links}]{\includegraphics[width=0.25\textwidth, height=1.5in]{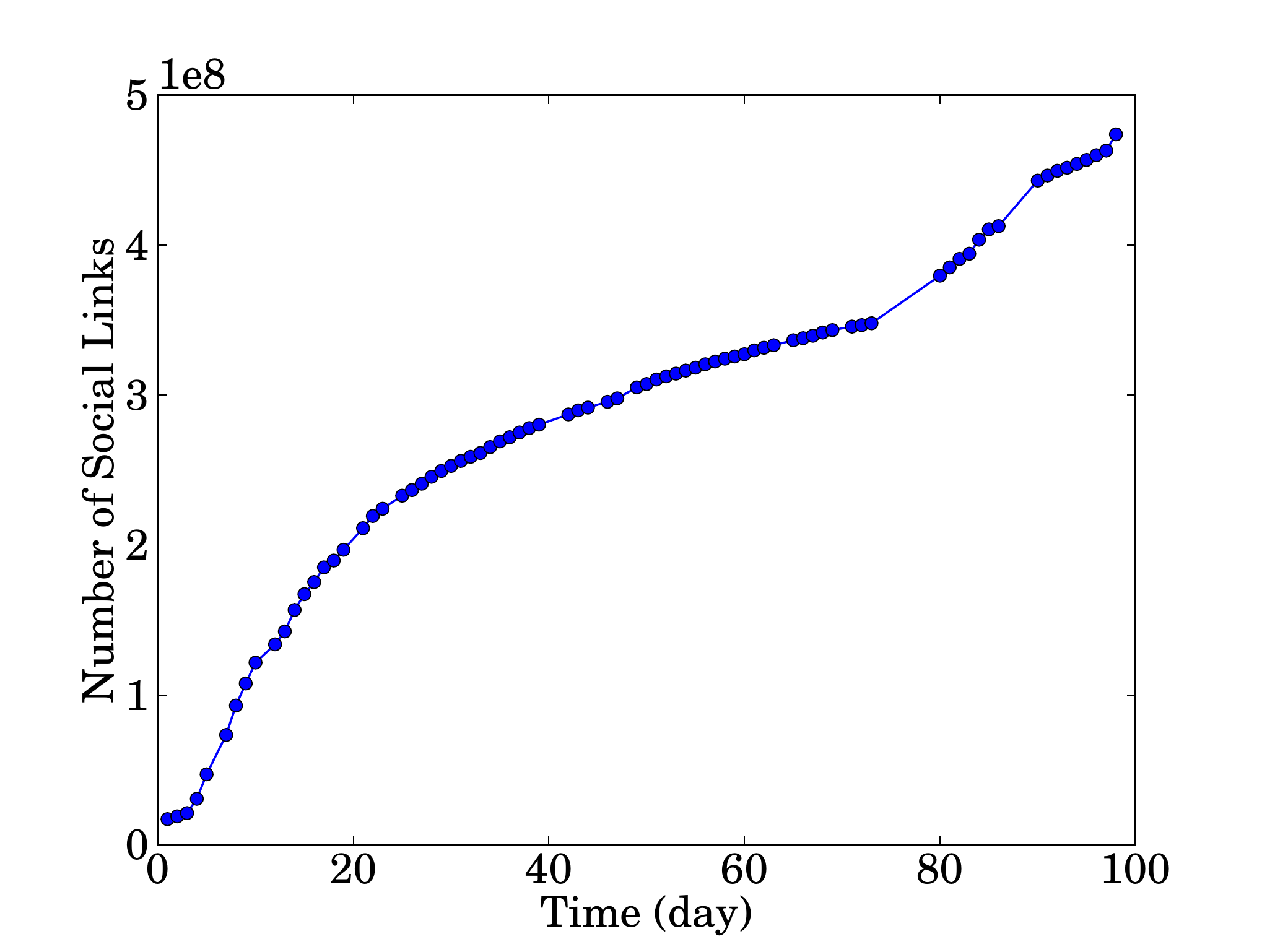}\label{soc-link}}
\subfloat[\scriptsize{Attribute links}]{\includegraphics[width=0.25\textwidth, height=1.5in]{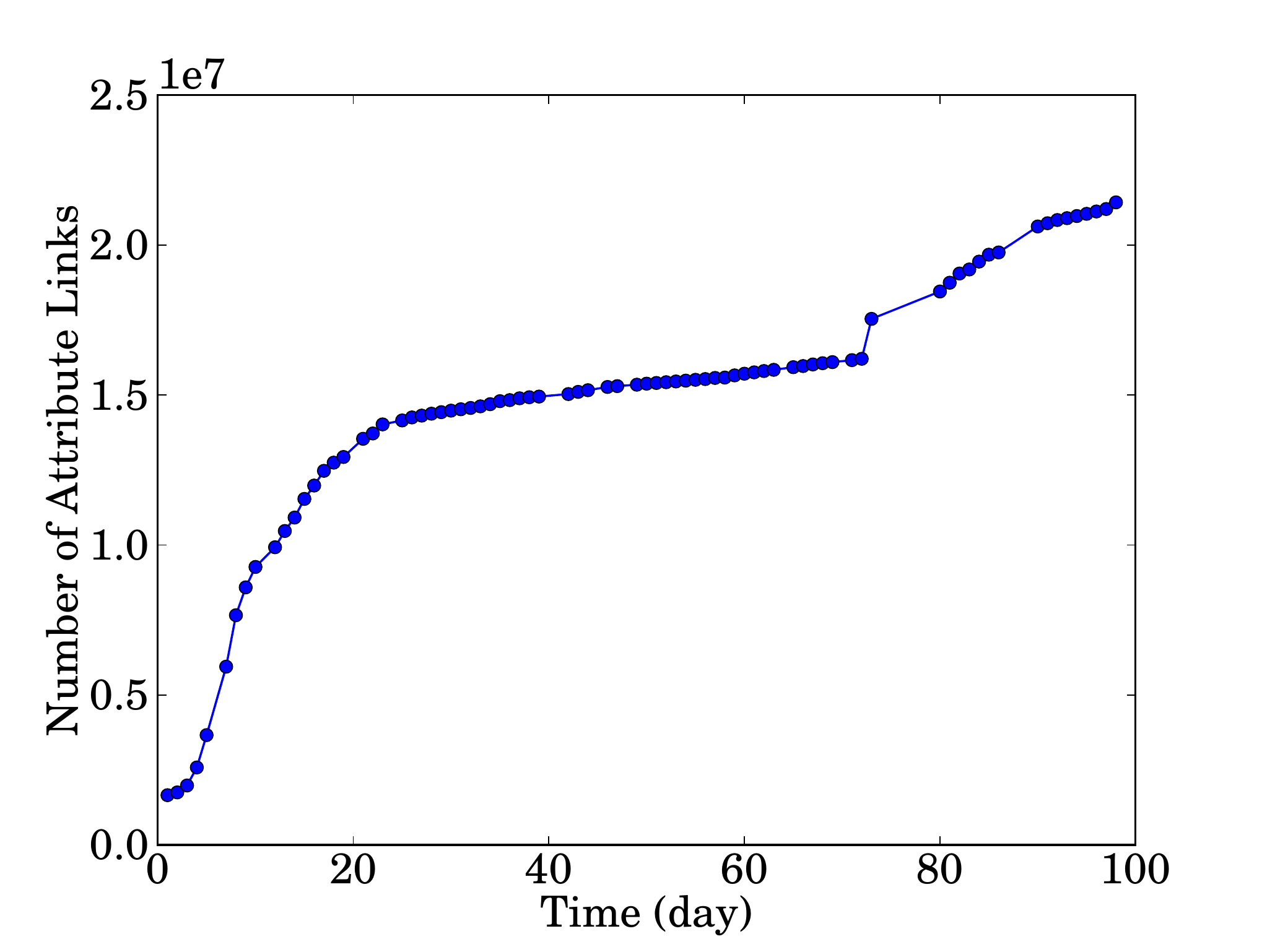}\label{attri-link}}
\tightcaption{Growth in the number of social and attribute links in the Google+ dataset}
\label{fig:link}
\end{figure}

Figures~\ref{fig:node} and~\ref{fig:link} show the temporal evolution of the number of 
nodes and links  in the Google+ \san. From the results we clearly see
three distinct phases in the evolution of Google+: Phase \textrm{I} from day one to
day 20, which corresponds to the early days of Google+ whose size increased 
dramatically; Phase \textrm{II} from day 21 to day 75, during which 
Google+ went into a stabilized increase phase; and Phase \textrm{III} from 
day 76 to day 98, when Google+ opened to public (i.e., without requiring an  invitation),
resulting in a dramatic growth again.
 We point this out because we observe  a similar three-phase evolution pattern for almost all  
network metrics that we analyze in the subsequent sections. 

In the following sections,  we use the last or largest snapshot, unless we are 
interested in the time-varying behavior. 

\mypara{Potential biases} We would like to acknowledge two possible biases.
First, users may keep some of their friends or circles private. In this case,
we can only see the publicly visible list. Thus we may not crawl the entire WCC
and underestimate the node degrees.
However, as discussed earlier, we obtain a very large connected
component that covers more than 70\% of known users which is sufficiently
representative. 
Second, users may choose not to declare their attributes, in
which case we may underestimate the impact of attributes on the social
structure. However, we find that roughly 22\% of users declare at least one
attribute which represents a statistically large sample from which to draw
conclusions. Furthermore, by validating the attribute-related results via 
further subsampling the attributes we have, we show that our attributes are representative of the entire attributes.

\section{Social Structure of the \\ Google+ \san}
\label{sec:structure}

In this section, we begin by presenting several canonical network metrics
commonly used for characterizing social networks such as the reciprocity, density,
clustering coefficient, and degree distribution~\cite{Mislove07,
Kossinets06, Kwak10, Newman03}.  These metrics are useful to expose the
inherent structure of a social network in terms of the friend relationships and
whether there are ``community'' structures beyond a one-hop friend
relationship. It is particularly useful to revisit these metrics in the
context of Google+ both because of its scale and because it enables  a somewhat
hybrid relationship model  compared to other networks such as Facebook,
Twitter, Flickr, and email networks. Furthermore, since we have a unique
opportunity to observe the network as it grew,  we also analyze how these
properties changed as the Google+ \san evolved.

\begin{figure*}
\vspace{-0.4cm}
\subfloat[\scriptsize{Reciprocity}]
{
\includegraphics[width=0.25\textwidth, height=1.5in]{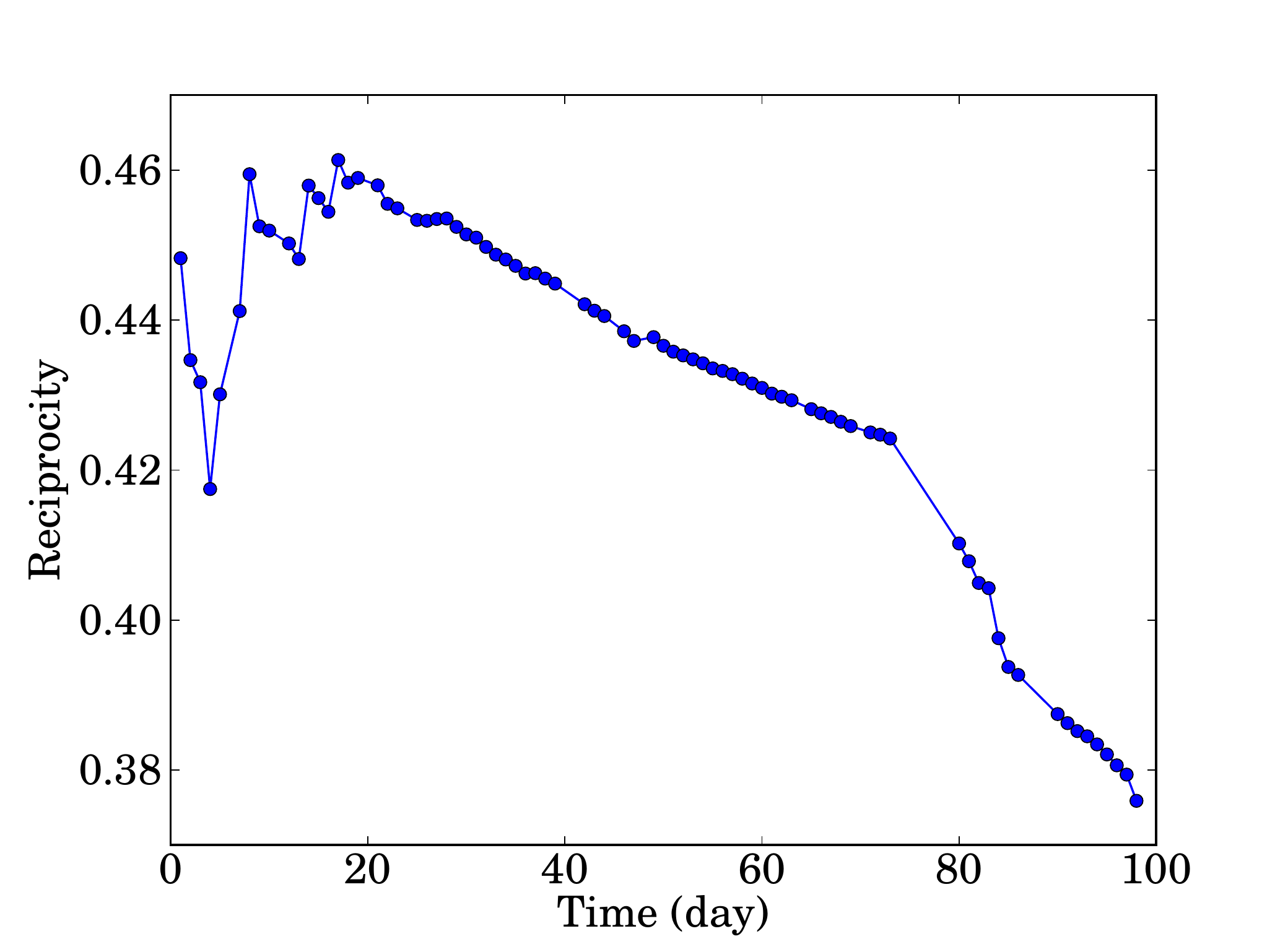}
\captionsetup{font=small,labelfont=bf}         
\label{fig:reciprocity}
}
\subfloat[\scriptsize{Social density}]
{
\includegraphics[width=0.25\textwidth, height=1.5in]{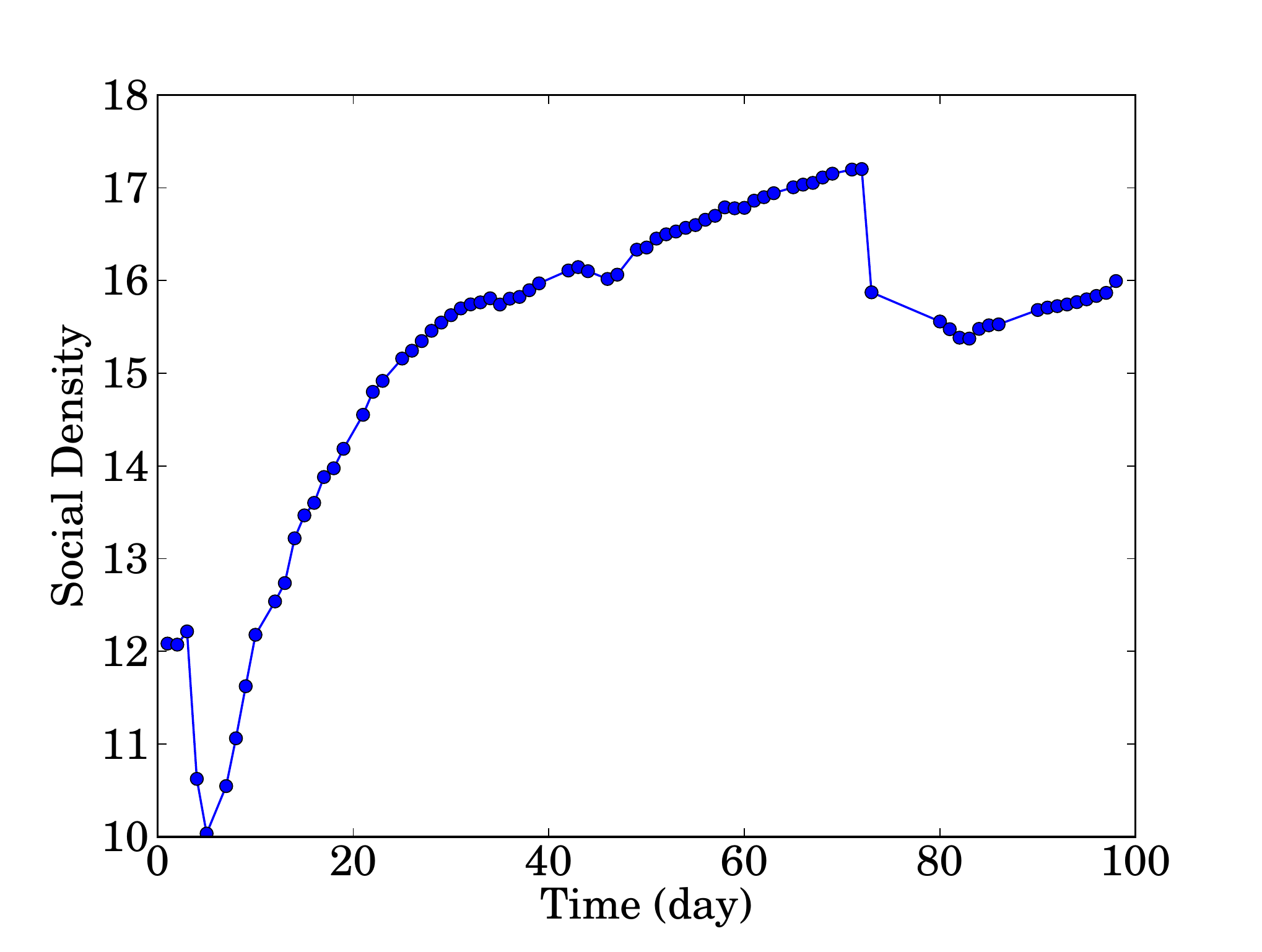}\label{soc-density}
}
\subfloat[\scriptsize{Diameter}]{
\includegraphics[width=0.25\textwidth, height=1.5in]{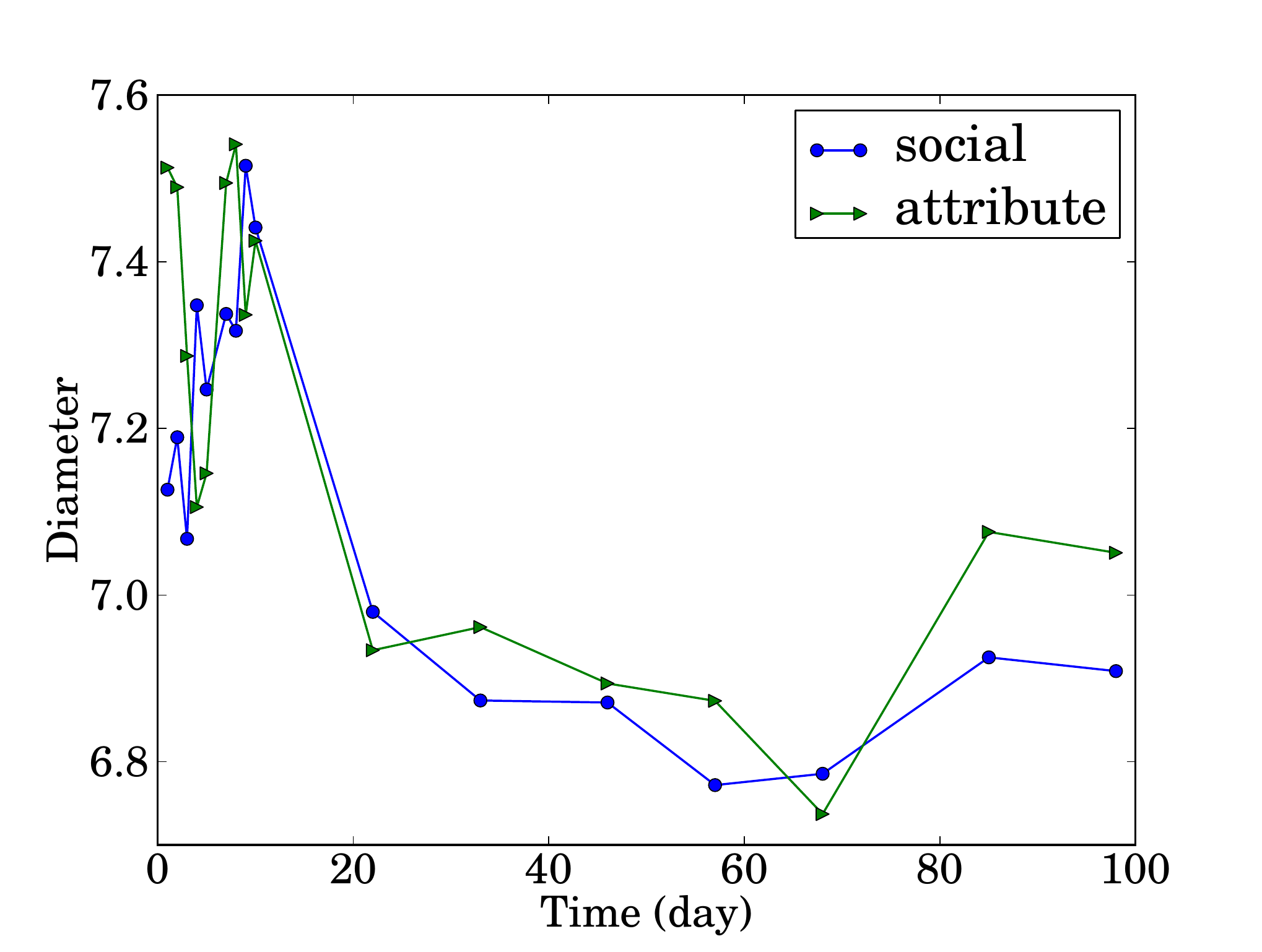}
\label{diameter-evo}
}
\subfloat[\scriptsize{Social clustering coefficient}]
{
\includegraphics[width=0.25\textwidth, height=1.5in]{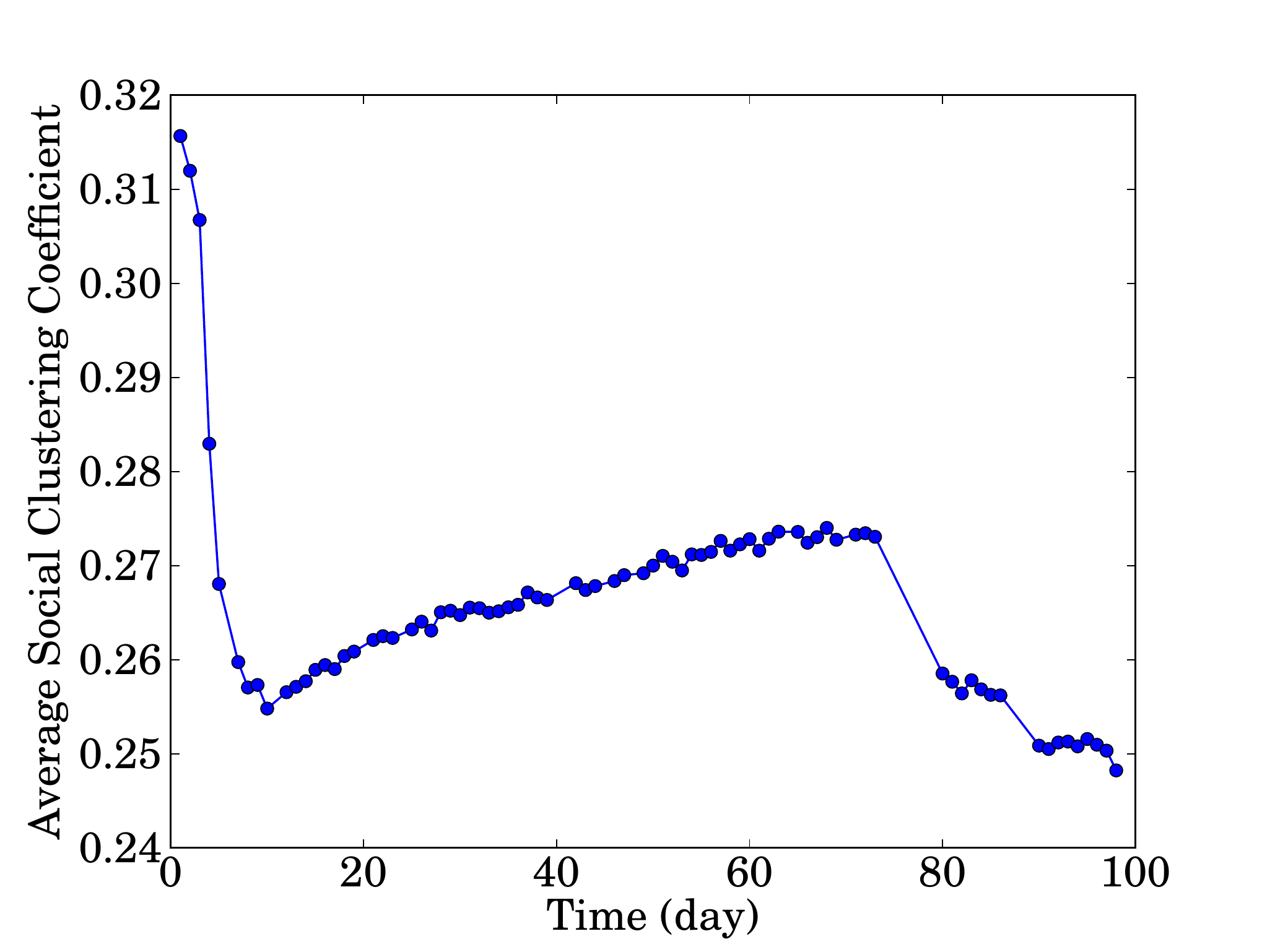}
\label{soc-clustering-coefficient}
}
\tightcaption{Evolution of four key metrics: reciprocity, density, diameter and clustering coefficient 
 on the Google+ \san. In each case, we observe distinct behaviors  in the three 
 phases corresponding to early initialization, time to public release, and time after 
 public release.}

\vspace{-2mm}
\end{figure*}

\subsection{Reciprocity}
  
The reciprocity metric  for directed social networks represents the fraction of
social links that are mutual; i.e., if there is a $A\rightarrow B$ edge what is
the likelihood of the reverse $B\rightarrow A$ edge. Previous work studied the
global reciprocities for \emph{specific snapshots} of social networks and
measured it to be $0.62$ on Flickr, $0.79$ on YouTube~\cite{Mislove07}, and
$0.22$ on Twitter~\cite{Kwak10}.  We focus  on the \emph{evolution} of global
reciprocity for Google+ in Figure~\ref{fig:reciprocity}.  The result shows  an
interesting behavior where the reciprocity fluctuates in Phase \textrm{I},
decreases in Phase \textrm{II} and decreases even faster in Phase \textrm{III}.
We speculate that this arises because of the hybrid nature of Google+.
Initially many people treat the network like a traditional social network
(e.g., Facebook) where the relationships are mutual. However, as time
progresses and people appear to become familiar with the  Twitter-like
publisher-subscriber model also offered by Google+, the reciprocity decreases.

\subsection{Density}
\label{sec:socialdensity}

The ratio of links-to-nodes, $\frac{|E_s|}{|V_s|}$,  captures the
\emph{density} 
\footnote{In graph theory, density is defined as the fraction of existing links with respect to all possible links. We follow the terminology in~\cite{Kumar06} in order to compare with previous results.}  of a social network. To put this in context, previous studies
show that the social density increases over time on citation and  affiliation
networks~\cite{Leskovec05}, on Facebook~\cite{Backstrom12}, and fluctuates in
an increase-decrease-increase fashion on Flickr~\cite{Kumar06}, and is
relatively constant on email communication networks~\cite{Kossinets06}.  

Figure~\ref{soc-density} shows the evolution of this social density metric  
in Google+.  We observe that social density in Google+ network has a sharp decrease followed by
an increase in \PI, a continued increase in \PII,  and a sudden drop in \PIII
 (when Google+ opened to the public) followed by a steady increase again.  This
three-phase pattern can be explained in conjunction with the trends in
Figures~\ref{soc-node} and~\ref{soc-link}.  In the early part of \PI, 
even though the rate of users joining Google+ is high, the rate of 
adding links is low, possibly because many 
 of a user's existing friends have not yet joined. This causes social density to decrease.
 As users acquire friends with a rate higher than the rate of new users in
later part of \PI and the same trend continuing in
\PII, the social density increases.  In \PIII, the number of users in  Google+
had a sudden jump due to the public release but the number friendship links
increases less dramatically,  which once again causes the social density to drop 
 significantly around t=70, but then starts slowly increasing again. Our findings have implications for network modeling. Specifically, many network models either assume constant density~\cite{Barabasi99, Kleinberg99} or power-law densification~\cite{Leskovec05}, which is not consistent with Google+.

\subsection{Diameter}
\label{sec:structure:distance}

In directed social networks, the distance between two user nodes $u$ and $v$,
$dist(u,v)$ is defined as the length of the shortest directed path whose 
head is $v$ and tail is $u$. Note that only social links $E_s$ are used in this definition. We find that the distribution of the distance between nodes 
has a dominant mode at a distance of six, with most nodes (90\%) having a distance 
of 5, 6, or 7 (not shown).

Based on the distance distribution, we can also define the \emph{effective diameter} as the 90-th percentile distance (possibly
with some interpolation) between every pair of
connected nodes~\cite{Leskovec05}. Unfortunately, 
computing the effective diameter is
infeasible for large networks, so we use the  HyperANF approximation 
algorithm~\cite{Boldi11}, which has been shown to be able to 
approximate diameter with  high accuracy. 

Previous work observed effective diameter shrinks in citation networks, autonomous networks and affiliation network~\cite{Leskovec05}, in Flickr and Yahoo! 360~\cite{Kumar06}, and in Cyworld~\cite{Ahn07}. 
However, we observe that the effective diameter
follows a three-phase evolution as seen in Figure~\ref{diameter-evo}, which again can be explained in conjunction with the trends in
Figures~\ref{soc-node} and~\ref{soc-link}.
In \PI, user joining rate outpaces link creation rate,  causing the diameter to increase; in
\PII, user joining rate is lower than link acquisition rate, resulting in
decreasing diameter; and in \PIII user joining rate is much higher, 
resulting in a diameter increasing phase again. Again, our observations have implications for network modeling. Existing network models either assume logarithmically growing diameter~\cite{Watts98, Barabasi99} or shrinking diameter~\cite{Leskovec10-JMLR, Leskovec05}.  


\subsection{Clustering Coefficient}
\label{sec:socialclustering}

Given a network $G$ and node $u$, $u$'s clustering coefficient is defined as $$c(u) =
\frac{L(u)}{|\Gamma_s(u)| (|\Gamma_s(u)| - 1)}, $$ where $L(u)$
is the number of links among $u$'s social neighbors $\Gamma_s(u)$ and the
average \emph{social clustering coefficient} is defined as $C_s=\frac{1}{|V_s|}\sum_{u\in
V_s}c(u)$~\cite{Watts98}.  Intuitively, this captures the community
structure among a user's friends.

Again, computing the average clustering coefficient is expensive.
Thus, we extend the constant-time approximate algorithm proposed by Schank et
al. for undirected networks~\cite{Schank05}, and develop an algorithm to
approximate the clustering coefficients for  a directed network.   With $\lceil
\frac{\mathrm{ln}2\nu}{2\epsilon^2} \rceil$ random samples, our constant time
algorithms can bound the error of average clustering coefficient within
$\epsilon$ with probability at least $1-\frac{1}{\nu}$.  In practice, 
we set the error to be $\epsilon=0.002$ and $\nu=100$. Algorithm details and theoretical analysis can be found in Appendix~\ref{sec:clustering}.

Kossinets et al.~\cite{Kossinets06} observed constant average social clustering coefficient over time in an email communication network. However, we find that the evolution of average social  clustering coefficient of Google+, which is shown in Figure~\ref{soc-clustering-coefficient},  again follows a three-phase evolution pattern where the 
 clustering  coefficient dramatically decreases in \PI, increases slowly in \PII and 
decreases again in \PIII. Our findings indicate that the community structure among users' friends is highly dynamic, which inspires us to do dynamic community detection.

\begin{figure}[thbf]
\vspace{-0.4cm}
\centering
\subfloat[\scriptsize{Outdegree}]{\includegraphics[width=0.25\textwidth, height=1.5in]{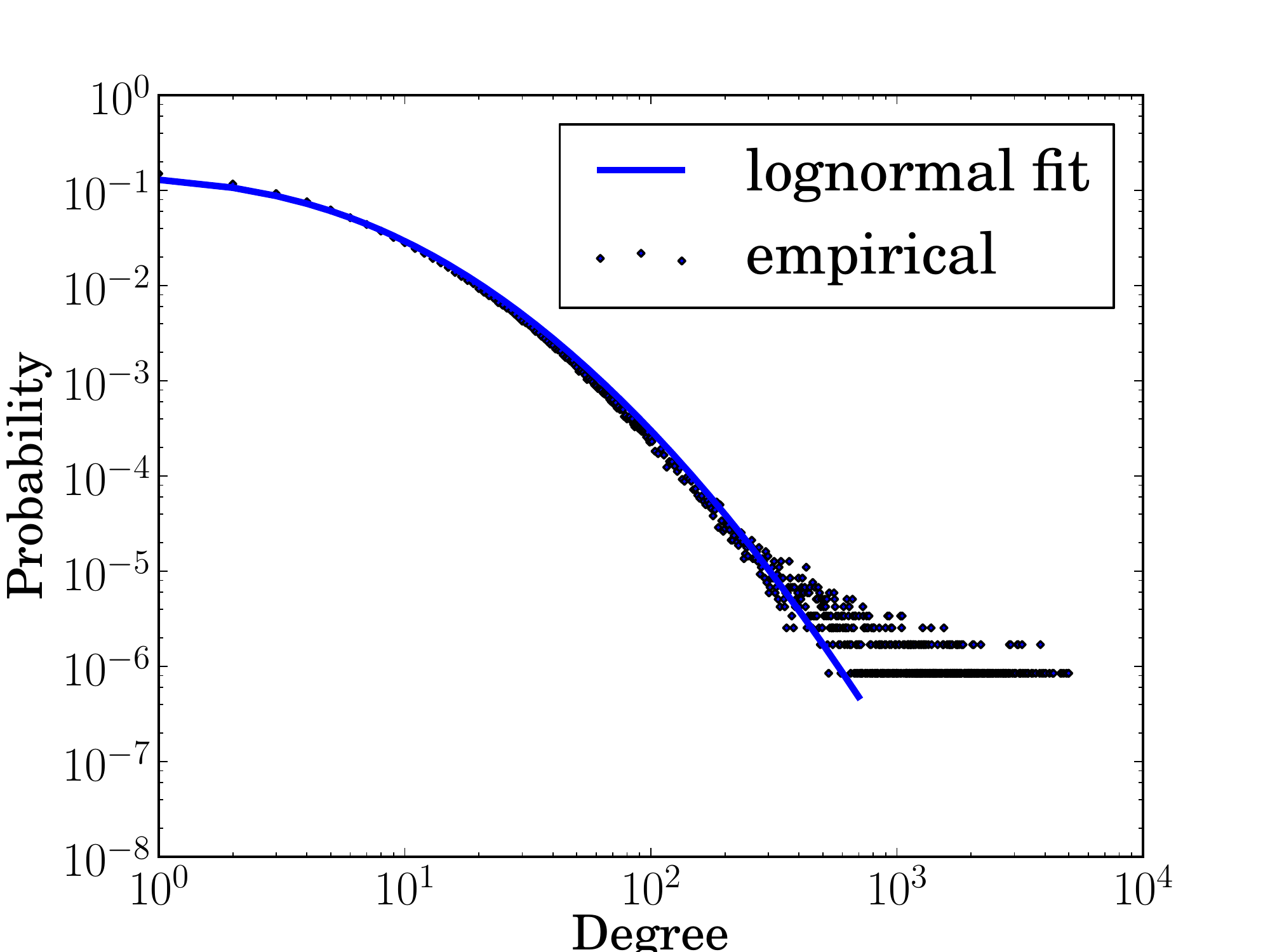} \label{fig:soc-soc-out}}
\subfloat[\scriptsize{Indegree}]{\includegraphics[width=0.25\textwidth, height=1.5in]{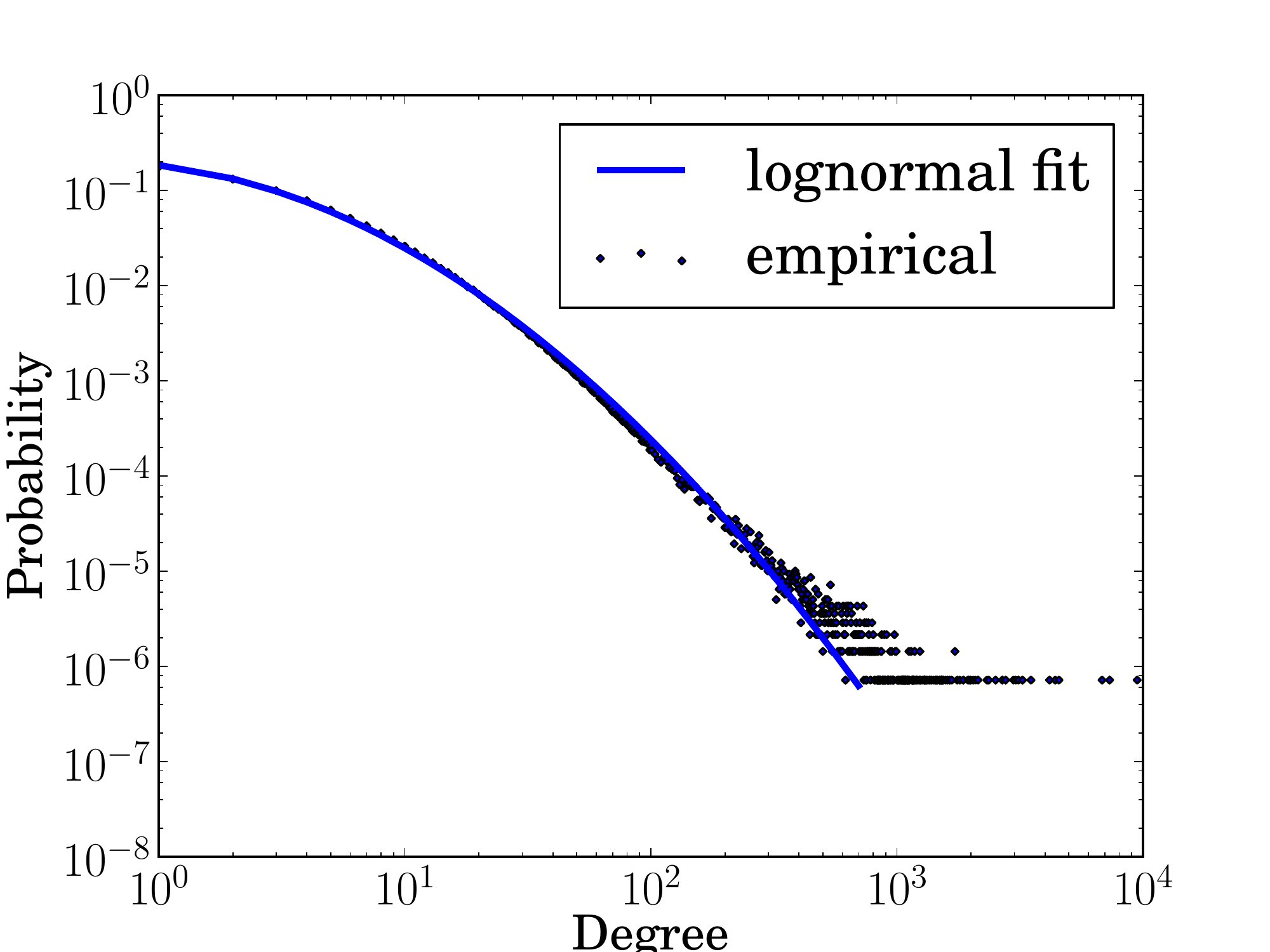} \label{fig:soc-soc-in}} 
\tightcaption{Indegree and outdegree distributions for the social nodes in the Google+ \san along with 
 their best-fit curves. 
We observe that both  are best modeled by a discrete lognormal distribution unlike many networks that suggest power-law distributions.
}
\label{fig:degree-distribution}
\vspace{-2mm}
\end{figure}

\begin{figure}[t]
\vspace{-0.4cm}
\centering
\subfloat[\scriptsize{Outdegree}]{\includegraphics[width=0.25\textwidth, height=1.5in]{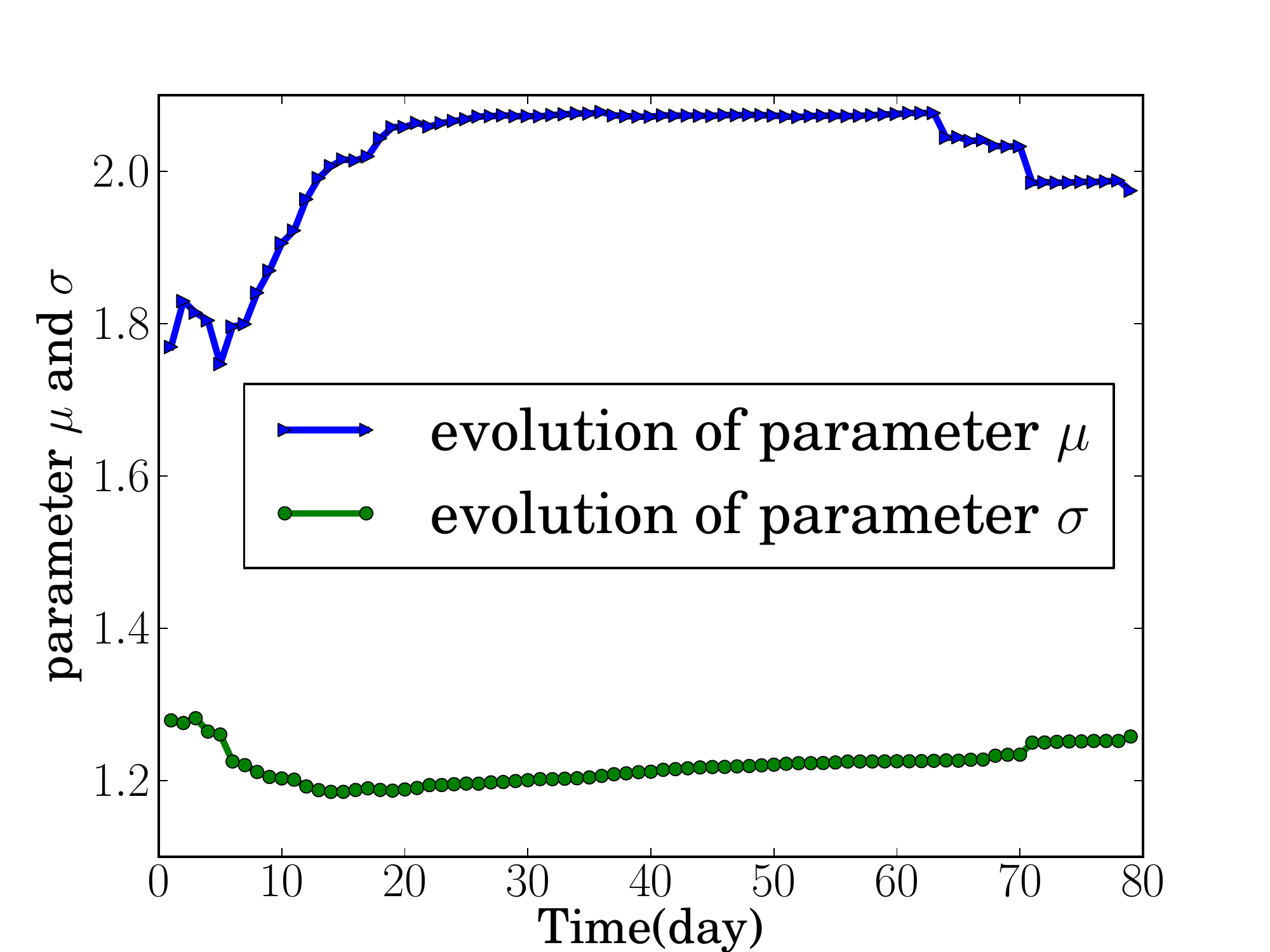} \label{fig:soc-soc-out-evo}}
\subfloat[\scriptsize{Indegree}]{\includegraphics[width=0.25\textwidth, height=1.5in]{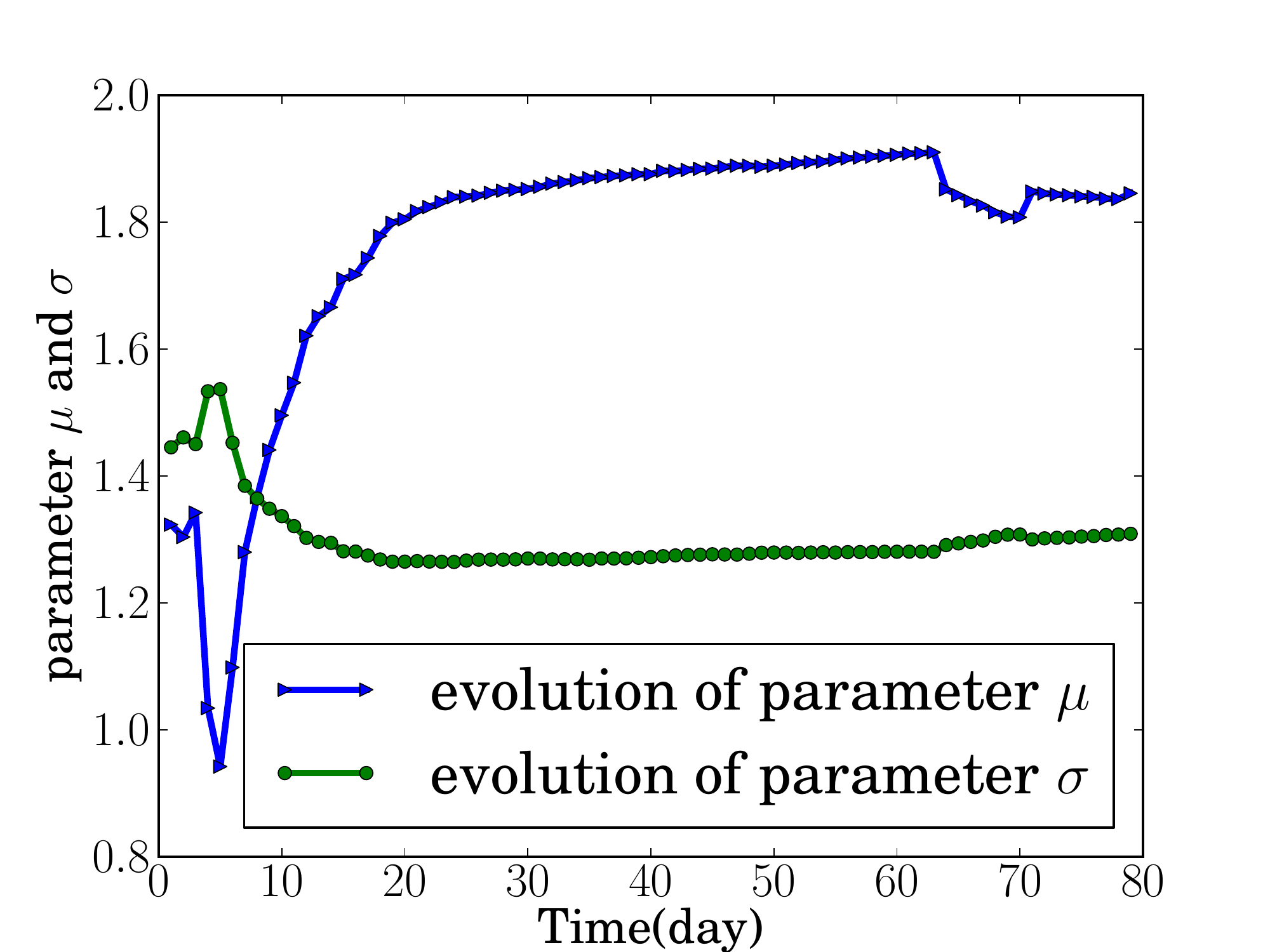} \label{fig:soc-soc-in-evo}}
\tightcaption{Evolution of the lognormal parameters for the indegree and outdegree distributions.}
\label{fig:degree-evolution}
\vspace{-2mm}
\end{figure}


\subsection{Degree Distributions}


Next, we focus on the social \emph{indegree}  and \emph{outdegree} of users in Google+. In each case, we are also interested in identifying an
empirical best-fit distribution using the tool~\cite{tool, Clauset09}, which compares fits of several widely used distributions (e.g., power-law,  lognormal, power-law with cutoff using) with respect to \emph{goodness-of-fit}. We find that unlike  many studies on social networks, in which
social degrees usually follow a power-law distribution~\cite{Dong09, Mislove07},  
 social degrees are best captured by a discrete lognormal distribution in Google+. Recall that
a random variable $x\in \mathbb{Z^+}$ follows a power-law distribution if
$p(x=k) \propto k^{-\alpha}$, where $\alpha$ is the exponent of the power-law
distribution. On the other hand,  a random variable $x\in \mathbb{Z^+}$ follows
a discrete lognormal distribution if $p(x=k) \propto
\frac{1}{k}\text{exp}(-\frac{(\text{ln}k - \mu)^2}{2\sigma^2})$~\cite{Bi01},
where $\mu$ and $\sigma$ are the mean and standard deviation respectively of
the lognormal distribution.

Figure~\ref{fig:degree-distribution} shows these degree distributions and their discrete lognormal fits, and Figure~\ref{fig:degree-evolution} shows the  evolutions of the parameters for the fitted discrete lognormal distributions. We see the evolution of the outdegree and indegree distributions follows a similar trend but with the fluctuation  differing 
in magnitude (Figures~\ref{fig:soc-soc-out-evo},~\ref{fig:soc-soc-in-evo}).

Lognormally distributed degree distributions imply that there are 
probabilistically more low degree social nodes in Google+ than those in 
power-law distributed networks.

\begin{figure}[htbf]
\vspace{-0.4cm}
\centering
\subfloat[\scriptsize{$k_{nn}$ metric}]{\includegraphics[width=0.25\textwidth, height=1.5in]{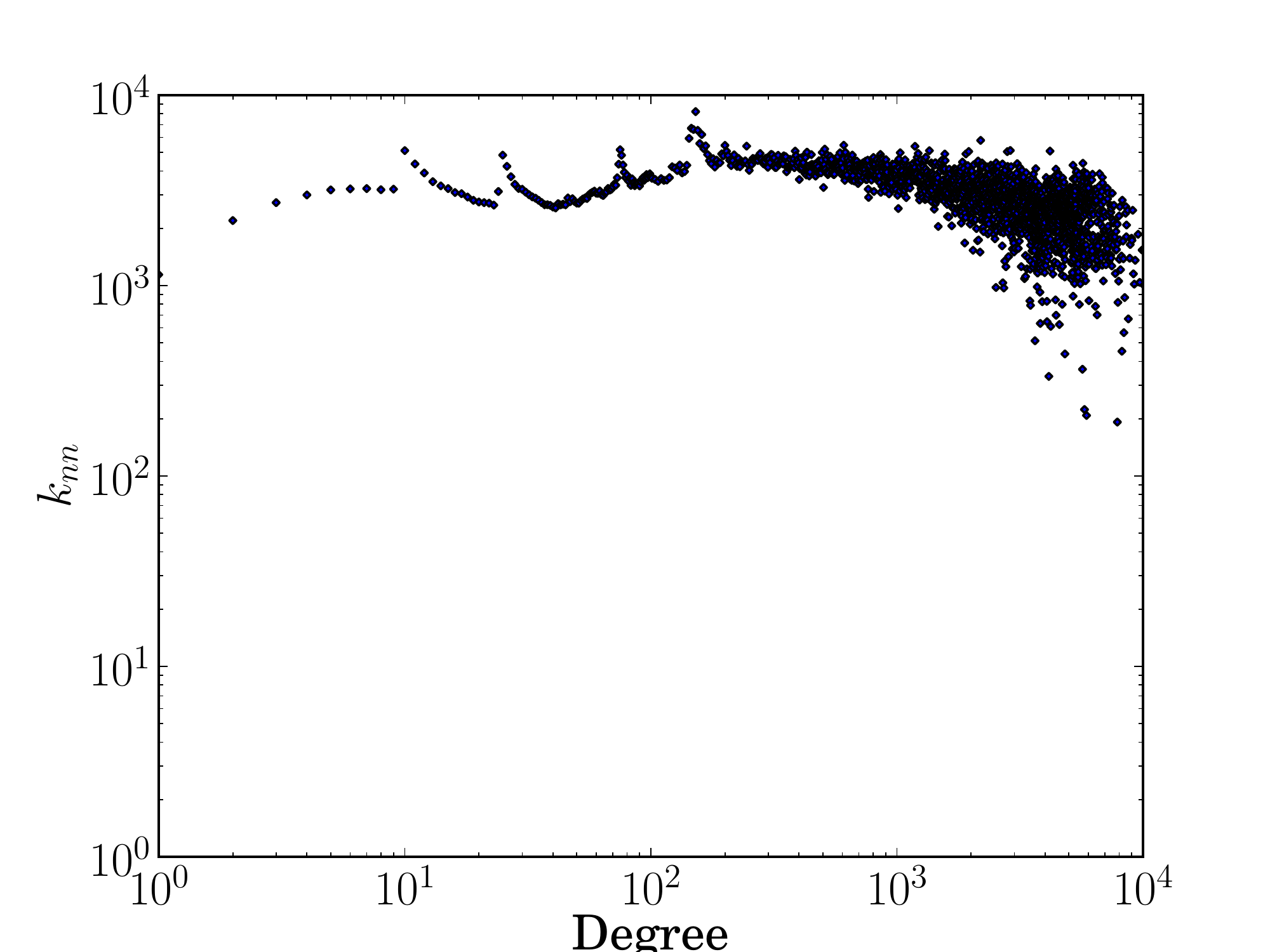}\label{fig:joint-degree-distribution}} 
\subfloat[\scriptsize{Assortativity}]{\includegraphics[width=0.25\textwidth, height=1.5in]{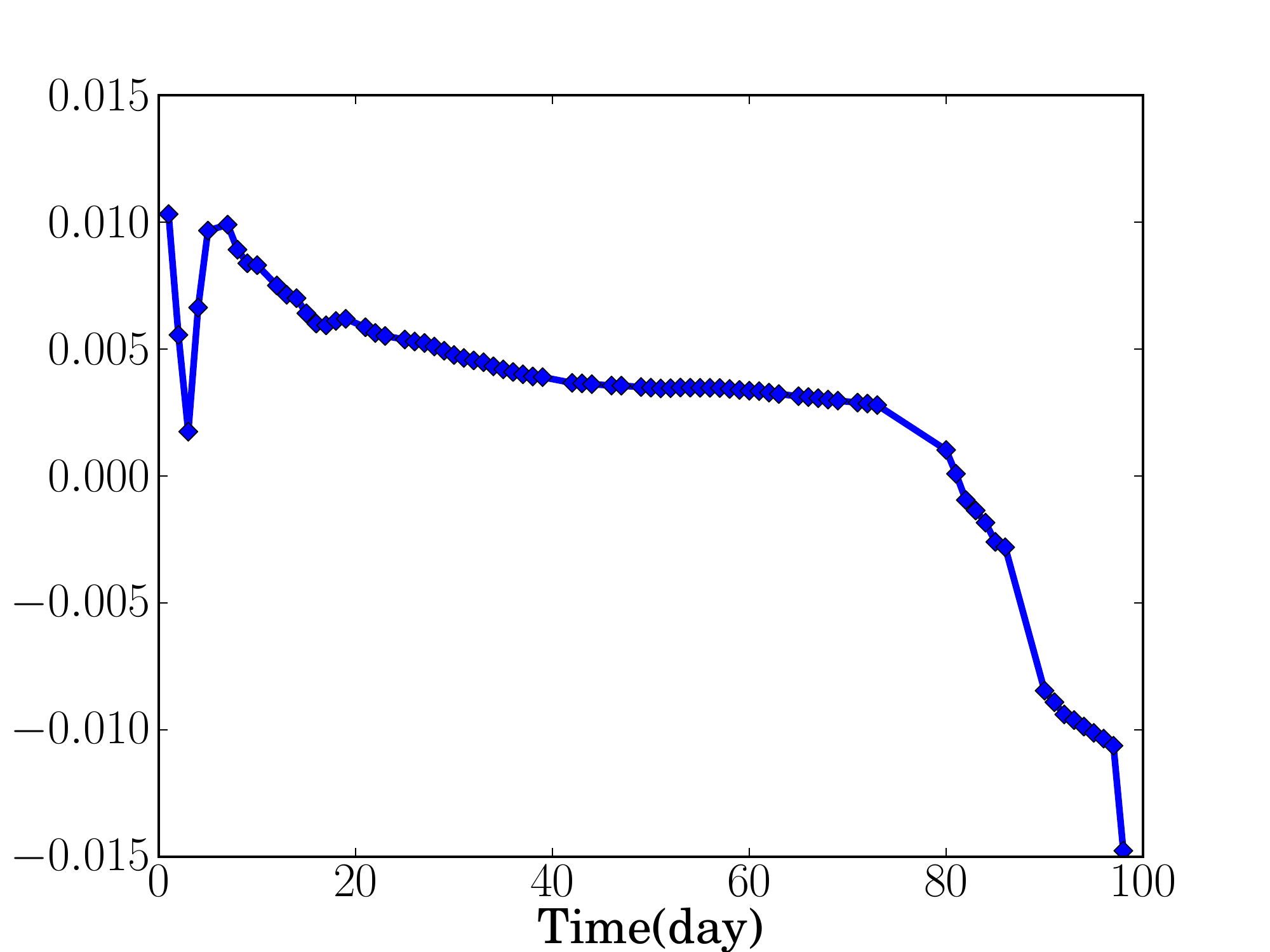}\label{fig:joint-degree-distribution-evo}} 
\tightcaption{ Two metrics for capturing the joint-degree distribution: (a) $k_{nn}$ shows a 
 log-log plot of the outdegree versus the average indegree 
of friends and (b) shows the evolution of the assortativity coefficient.} 
\vspace{-2mm}
\end{figure}

\subsection{Joint Degree Distribution} 
Last, we examine the joint degree distribution (JDD) of the Google+
social structure. JDD is useful for understanding the preference of a
node to attach itself to nodes that  are similar to itself. 
One way to approximate the JDD is using the degree correlation function $k_{nn}$, which maps outdegree to the average indegree of all nodes connected to nodes of that outdegree~\cite{Pastor-Satorras01, Mislove07}. 
An increasing $k_{nn}$ trend indicates high-degree nodes tend to connect to other high-degree nodes; a decreasing $k_{nn}$ represents the opposite trend. 
Figure~\ref{fig:joint-degree-distribution} shows the $k_{nn}$ function for Google+ social structure.

The JDD can further be quantified using the \emph{assortativity coefficient}
$r$ that can range from -1 to 1~\cite{Newman03}.  $r$ is positive if $k_{nn}$
is positively correlated to node degree $k$.
Figure~\ref{fig:joint-degree-distribution-evo} illustrates the evolution of the
assortativity coefficient. We observe that $r$ keeps decreasing in all three
phases but at different rates.
Furthermore, unlike many traditional social networks where the
assortativity coefficient is typically positive---0.202 for Flickr, 0.179 for
LiveJournal and 0.072 for Orkut~\cite{Mislove07, Newman03}---Google+ has 
 almost neutral assortativity close to 0.   The neutral assortativity can possibly be explained by the hypothesis that Google+ is a hybrid of two ingredients, i.e., a traditional social network and a publisher-subscriber network (e.g., Twitter). Traditional social networks usually have positive assortativity; publisher-subscriber networks often have negative assortativity because high-degree publisher nodes tend to be connected to low-degree subscriber nodes. Thus a hybrid of them results in a network with neutral assortativity.  The evolution pattern of Google+' assortativity coefficient (i.e., positive in \PI, around 0 in \PII, and negative in \PIII) manifests the competing process of the two ingredients of Google+. More specifically, the traditional network ingredient slightly wins in \PI, resulting in a slightly positive assortativity coefficient. A draw between them in \PII results in the neutral assortativity. In \PIII, the publisher-subscriber ingredient wins, resulting in a slightly negative assortativity coefficient. This implies that Google+ is more and more like a publisher-subscriber network.

\subsection{Summary of Key Observations and Implications}
 Analyzing the social structure of Google+ and its evolution over time, we 
 find that: 

\begin{packeditemize}

\item In contrast to many traditional networks, we find that Google+ has low reciprocity, the social degree
distribution is best modeled by a lognormal distribution rather than a 
power-law distribution, and the assortativity is neutral rather than positive.  

\item Google+ is somewhere between a traditional social network
(e.g., Flickr) and a  publisher-subscriber network (e.g., Twitter),
reflecting the hybrid interaction model that it offers. Moreover, it's more and more closer to a publisher-subscriber network.

\item The evolutionary patterns of various network metrics in Google+ are different from those in many traditional networks or assumptions of various network models. These findings imply that existing models cannot explain the underlying growing mechanism of Google+, and we need to design new models for reproducing social networks similar to Google+.

\end{packeditemize}

\section{Attribute Structure of the \\ Google+ \san}
\label{sec:attribute}

In the previous section we looked at well-known social network metrics.  In
this section, we focus on analyzing the attribute structure of the Google+
\san. To this end, we extend the metrics from the previous section to the
attributes as well.    Finally, we show the importance of using attributes
in understanding the social structure by studying their impact on metrics we
analyzed earlier (e.g., reciprocity, clustering coefficient, and
degree distribution). These attribute-related studies will characterize the attribute structure, give us insights about the underlying growing mechanism of Google+, and eventually guide us design a new generative model for Google+ \san.

\begin{figure}[htbf]
\hspace{-0.5cm}
\subfloat[\scriptsize{Attribute density}]
{
\includegraphics[width=0.25\textwidth, height=1.5in]{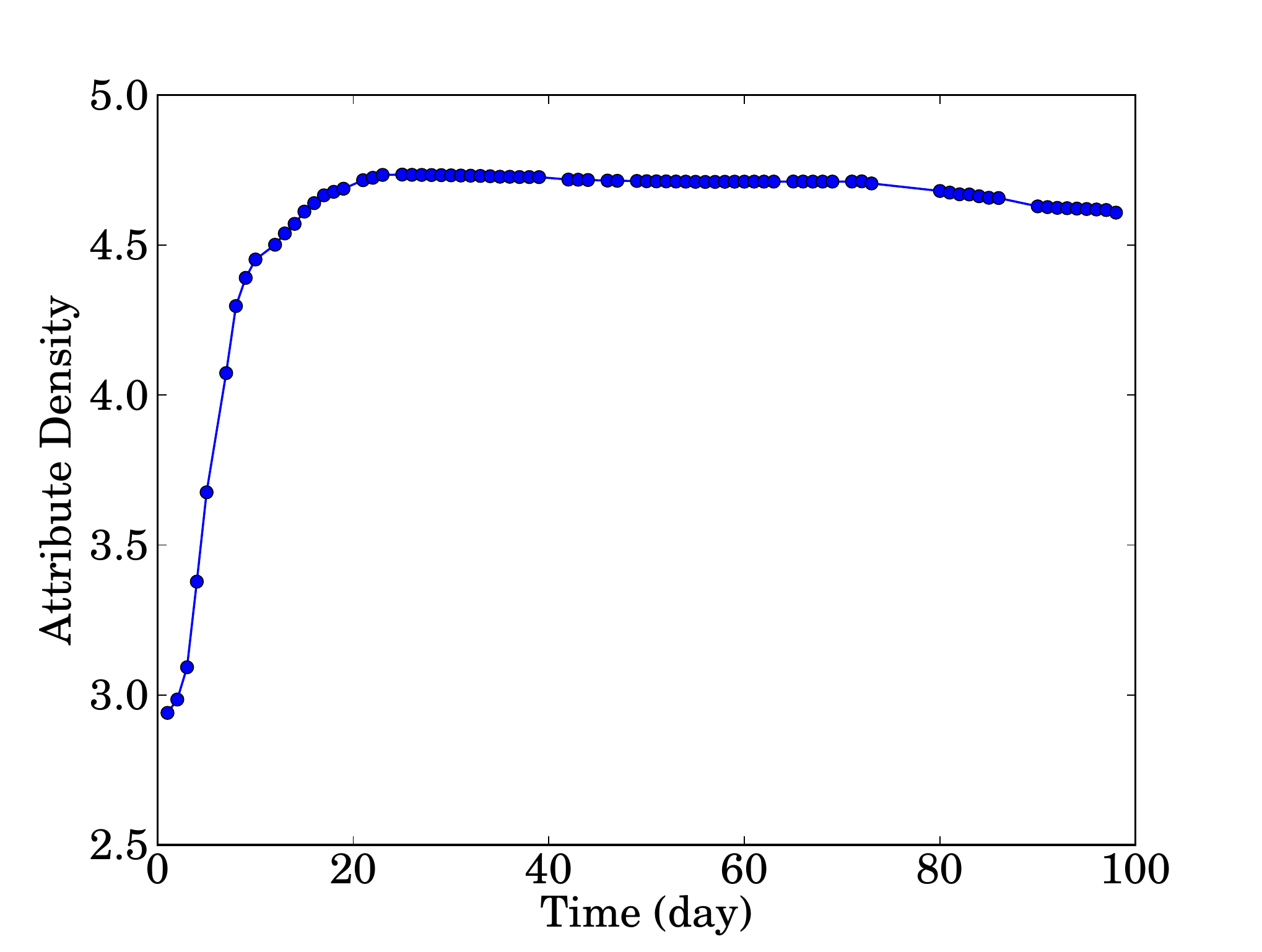}\label{attri-density}
}
\hspace{-0.5cm}
\subfloat[\scriptsize{Clustering coefficient}]
{
\includegraphics[width=0.25\textwidth, height=1.5in]{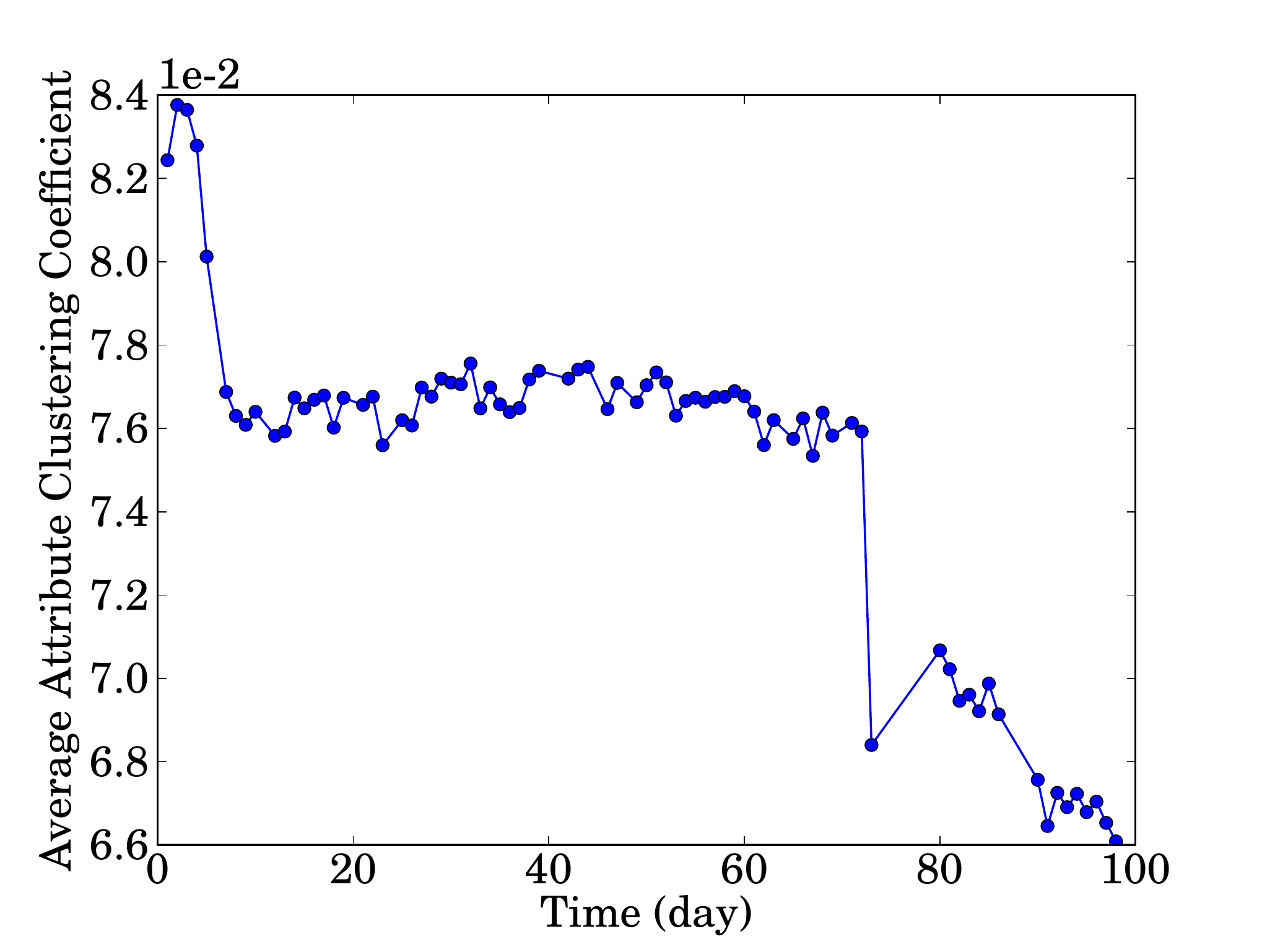}
\label{attri-clustering-coefficient}
}
\tightcaption{Evolution of the attribute density and average attribute clustering coefficient in the Google+ \san. }
\vspace{-2mm}
\end{figure}

\begin{figure}[htbf]
\vspace{-0.4cm}
\subfloat[\scriptsize{Distributions in the original \san}]{\includegraphics[width=0.25\textwidth, height=1.5in]{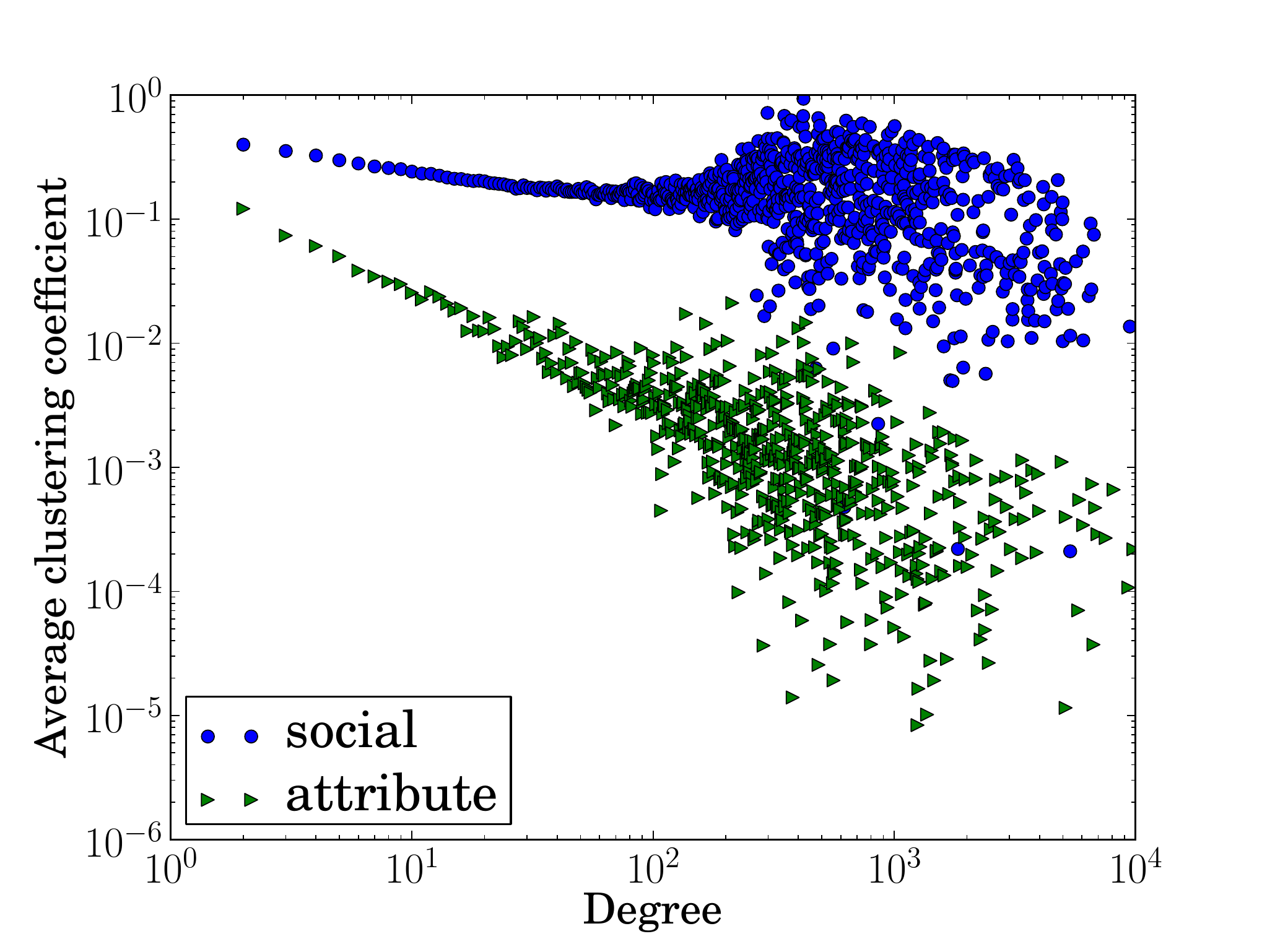}\label{fig:clustering-coefficient-dis-deg}}
\subfloat[\scriptsize{Comparison with a subsampled \san }]
{
\includegraphics[width=0.25\textwidth, height=1.5in]{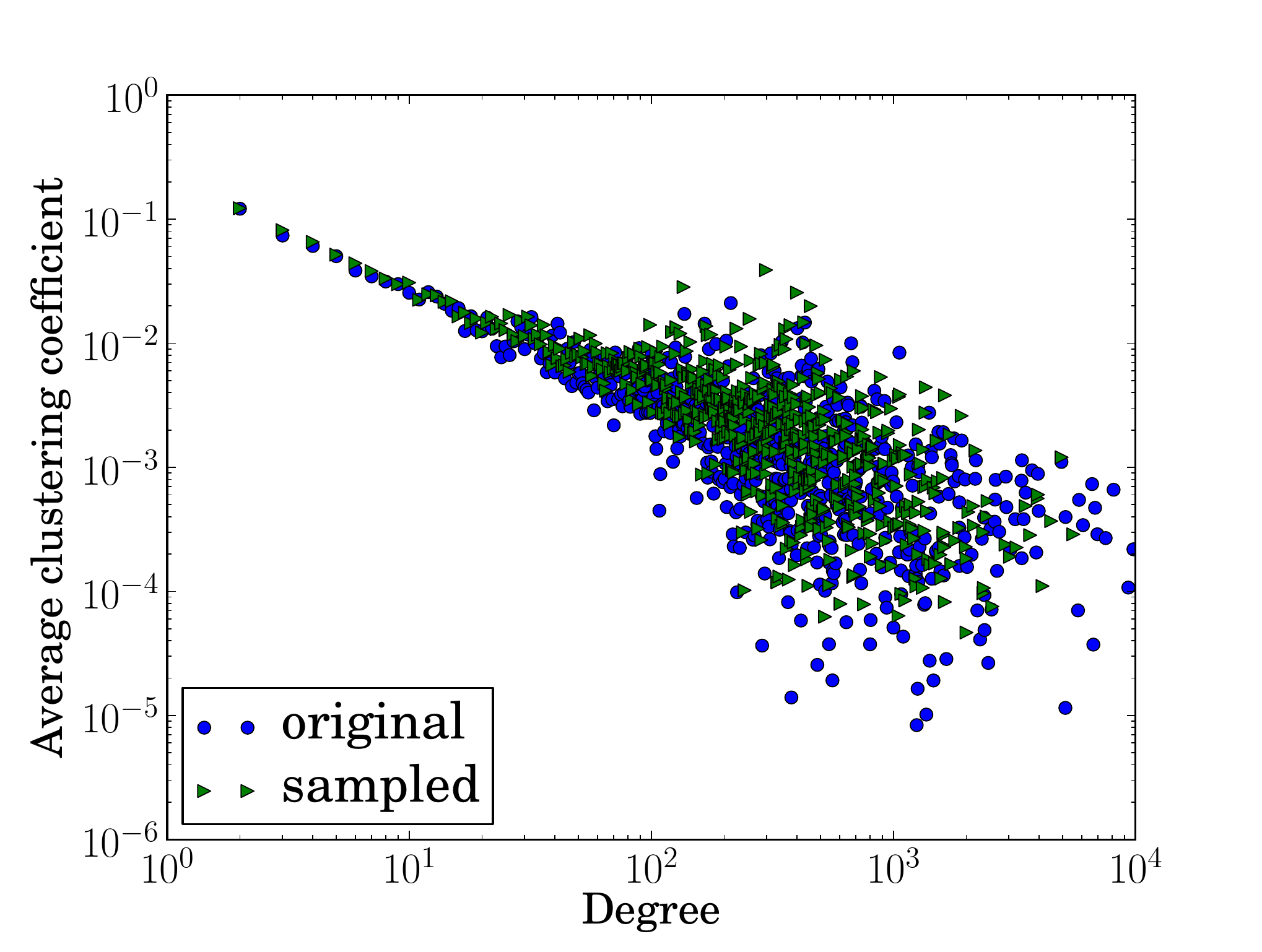}
\label{origin-sample}
}
\tightcaption{Distributions of clustering 
coefficient with respect to node degrees. (a) Comparison of social and attribute clustering coefficient distributions in the original \san. (b)Comparison of distributions of attribute clustering coefficients in the original \san and the subsampled \san.}
\vspace{-2mm}
\end{figure}

\subsection{Attribute Metrics} 

\myparatight{Density} 
We consider a natural extension of the social density metric from
\Section\ref{sec:socialdensity} and define \emph{attribute density} as
$\frac{|E_a|}{|V_a|}$. 
Different from our observations with
social density in Figure~\ref{soc-density}, 
in Figure~\ref{attri-density}, 
We observe the attribute density increases rapidly in \PI,
stays relatively flat  in \PII, and slightly decreases in 
\PIII.  
The reason for the decrease in \PIII  is  the large volume of new (i.e.,
non-invitation) users joining Google+ with many new attribute nodes whose
social degrees are small.

\myparatight{Diameter} We extend the distance metric from \Section\ref{sec:structure:distance}
 to define  the \emph{attribute distance} between two
attribute nodes $a$ and $b$ as $dist(a,b)=\text{min} \{dist(u,v) |
u\in\Gamma_s(a), v\in \Gamma_s(b)\} + 1$.\footnote{Other 
definitions are possible, e.g, using average instead of min. We choose min because of its computational efficiency.}
Intuitively, attribute
distance is the minimum number of social nodes that a attribute node has to
traverse before reaching to the other one; i.e., attribute distance is the
distance between two attribute communities. Similarly, we can 
 consider the effective diameter using this attribute distance. Figure~\ref{diameter-evo} 
 also shows the evolution of the attribute diameter and shows that it 
 very closely mirrors the social diameter.

\myparatight{Clustering coefficient} Similarly, we generalize the social clustering coefficient from
\Section\ref{sec:socialclustering} to define the \emph{attribute
clustering coefficient} $c(u)$ for the attribute node $u$, and the average attribute
clustering coefficient as $C_a=\frac{1}{|V_a|}\sum_{u\in V_a}c(u)$.  This attribute clustering coefficient $c(u)$ characterizes  the power
of attribute $u$ to form communities among users who have the attribute $u$.  
Compared to Figure~\ref{soc-clustering-coefficient}, we find in
Figure~\ref{attri-clustering-coefficient} that the average attribute clustering
coefficient evolves in a different pattern since it's relatively stable in \PII. 

We also show the distribution of average social and attribute clustering coefficients as a 
 function of node degree in Figure~\ref{fig:clustering-coefficient-dis-deg}. We observe that both social and attribute clustering coefficients follow a power-law distribution with respect to node degrees, but attribute  clustering coefficient distribution has a larger exponent. Moreover, we see that in general attribute clustering coefficients are lower  because many shared attributes (e.g., city or major) will not naturally translate into a social relationship. 

\begin{figure}[htbf]
\vspace{-0.5cm}
\centering
\subfloat[\scriptsize{Attribute degree of social nodes}]{\includegraphics[width=0.24\textwidth, height=1.5in]{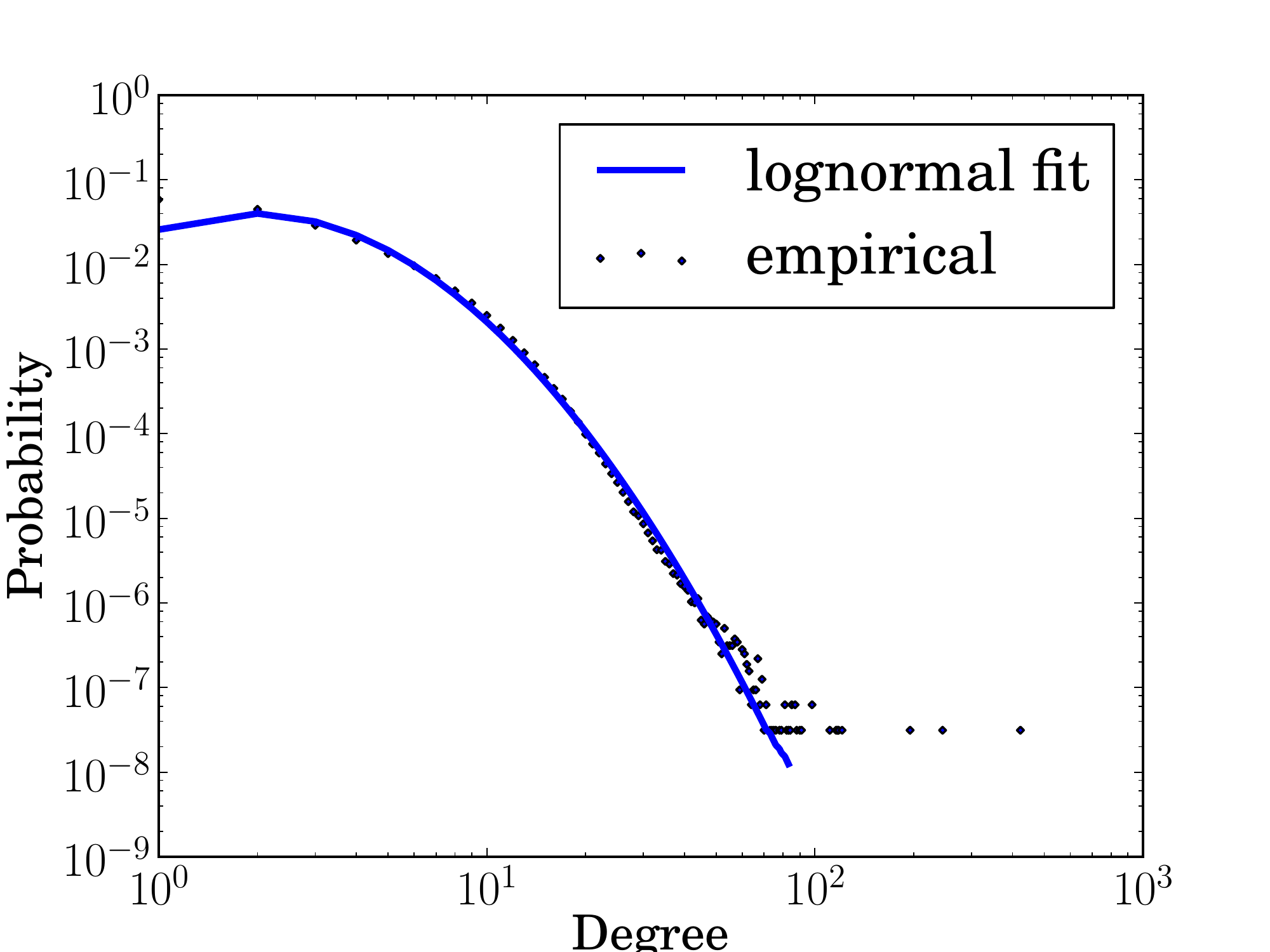} \label{fig:soc-attri}}
\subfloat[\scriptsize {Social degree of attribute nodes}]{\includegraphics[width=0.24\textwidth, height=1.5in]{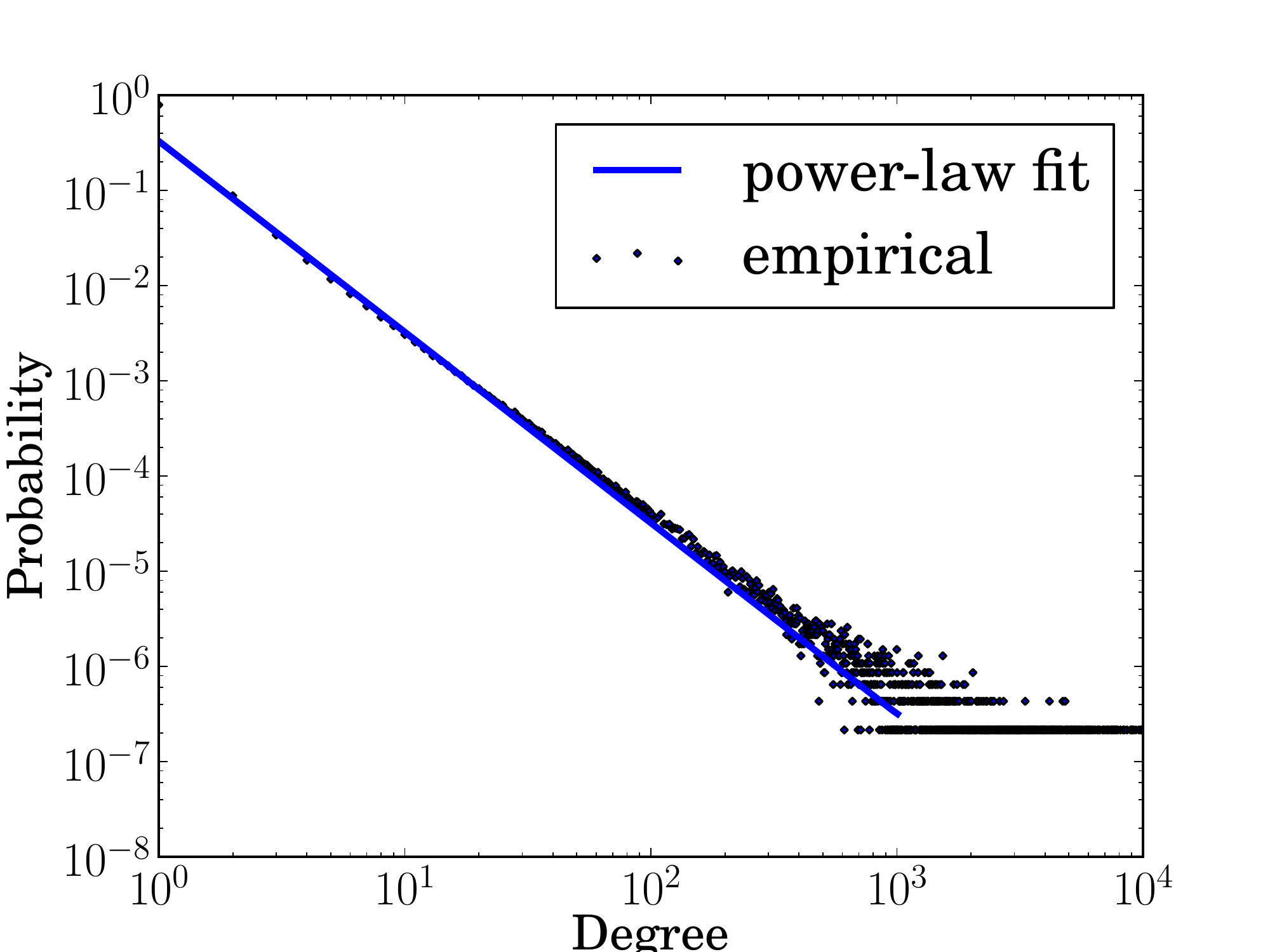} \label{fig:attri-soc}}
\tightcaption{Distributions of attribute-induced degrees in the Google+ \san along with 
 their best fits. The attribute degree of social nodes is best modeled by a 
lognormal whereas the social degree of attribute nodes is 
  best modeled by a power-law distribution. 
}
\label{fig:attrdegree-distribution}
\vspace{-2mm}
\end{figure}

\begin{figure}[htbf]
\vspace{-0.5cm}
\centering
\subfloat[ \scriptsize{Attribute degree of social  nodes}]{\includegraphics[width=0.24\textwidth, height=1.5in]{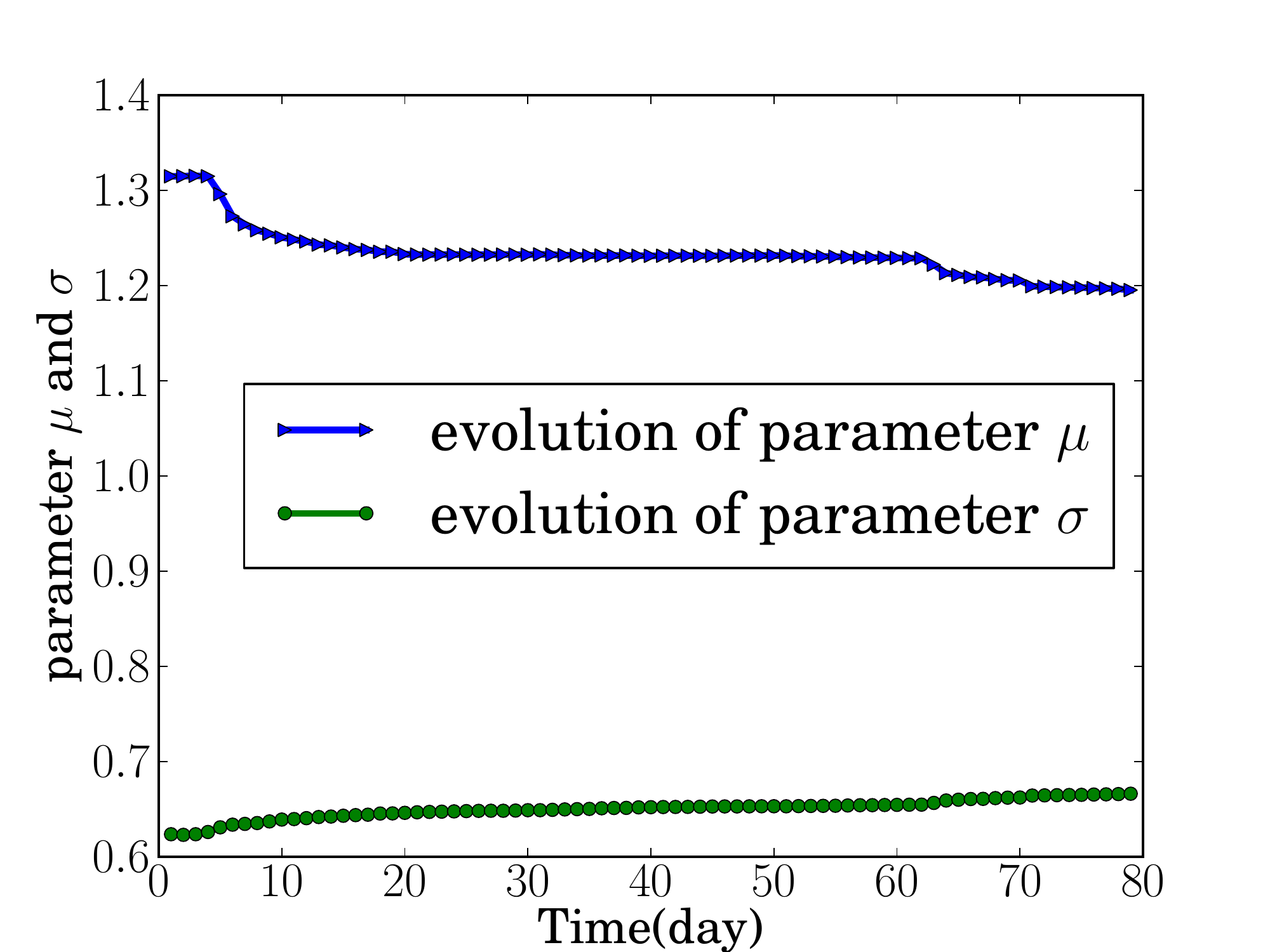} \label{fig:soc-attri-evo}}
\subfloat[ \scriptsize{Social degree of attribute nodes}]{\includegraphics[width=0.24\textwidth, height=1.5in]{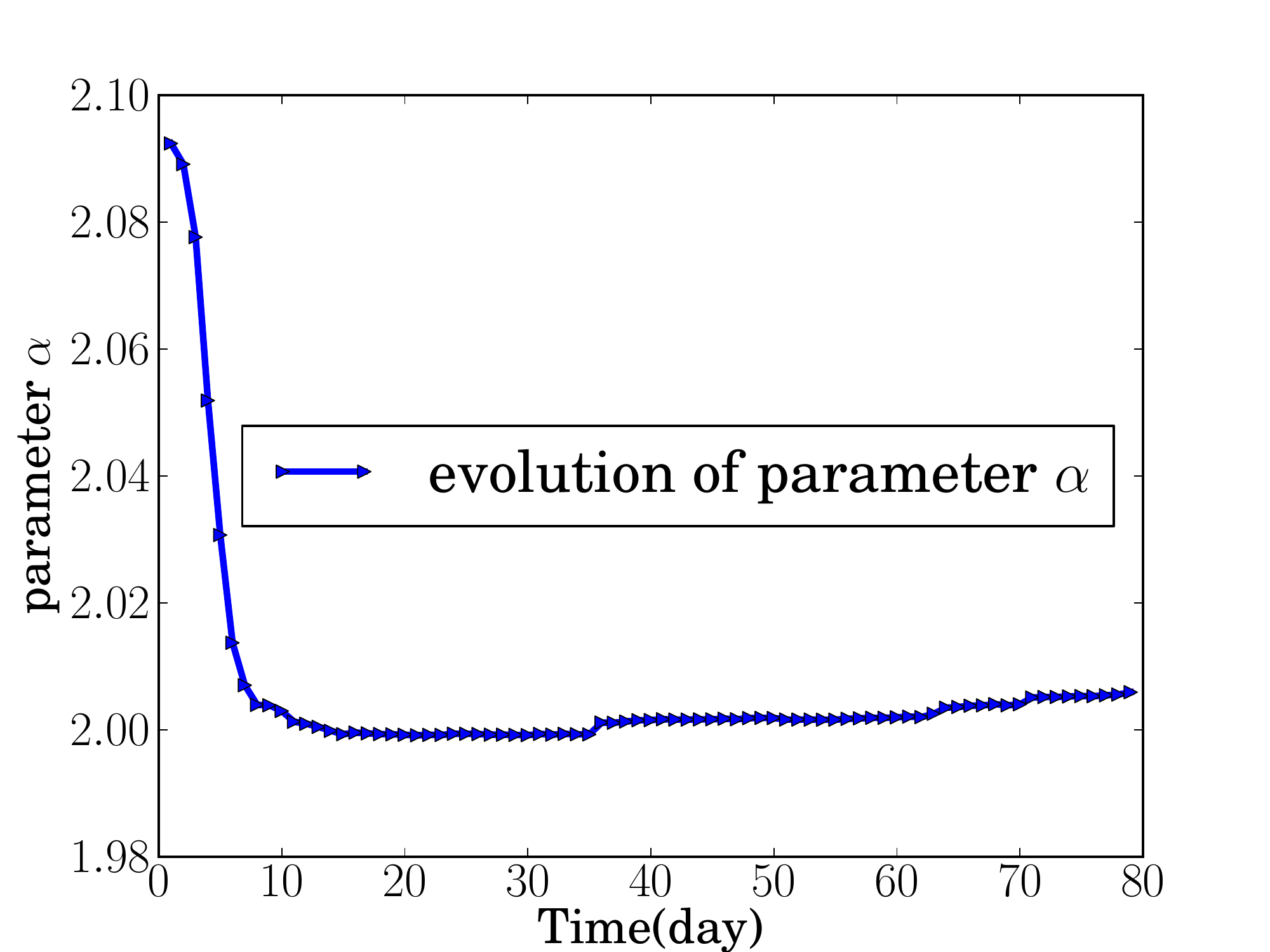} \label{fig:attri-soc-evo}}
\tightcaption{Evolution in lognormal and power-law parameters for the attribute and social degree distributions}
\label{fig:attrdegree-evolution}
\vspace{-1mm}
\end{figure}

\myparatight{Degree distributions} As discussed earlier, \sanplural 
  introduce  edges between social and attribute nodes. Thus, we  consider two new
notions of node degrees: (1)  \emph{social degree} of  attribute nodes (i.e.,
the number of users that have this attribute) and (2) \emph{attribute degree}
of  social nodes (i.e., the number of attributes each user has).
We find that the attribute degree of social nodes is best modeled by a 
lognormal distribution whereas the social degree of attribute nodes is best 
modeled by a power-law distribution. Figure~\ref{fig:attrdegree-distribution} and
Figure~\ref{fig:attrdegree-evolution} show the degree distributions and
evolution of their fitted parameters.

%
In terms of the evolution, we find the attribute degree
evolution seen in  Figure~\ref{fig:attrdegree-evolution} is significantly
different from the previous observation in  Figure~\ref{fig:degree-evolution}:
its mean decreases in \PI,  remains roughly constant in \PII, and decreases
again in \PIII. However, its standard deviation increases slightly in all
phases.  Finally, for the  social degree which follows a power-law
distribution, the exponent decreases fast in  \PI, and increases slightly in
\PII  and \PIII.

\begin{figure}[htbf]
\vspace{-0.4cm}
\centering
\subfloat[\scriptsize{Attribute $k_{nn}$}]{\includegraphics[width=0.25\textwidth, height=1.5in]{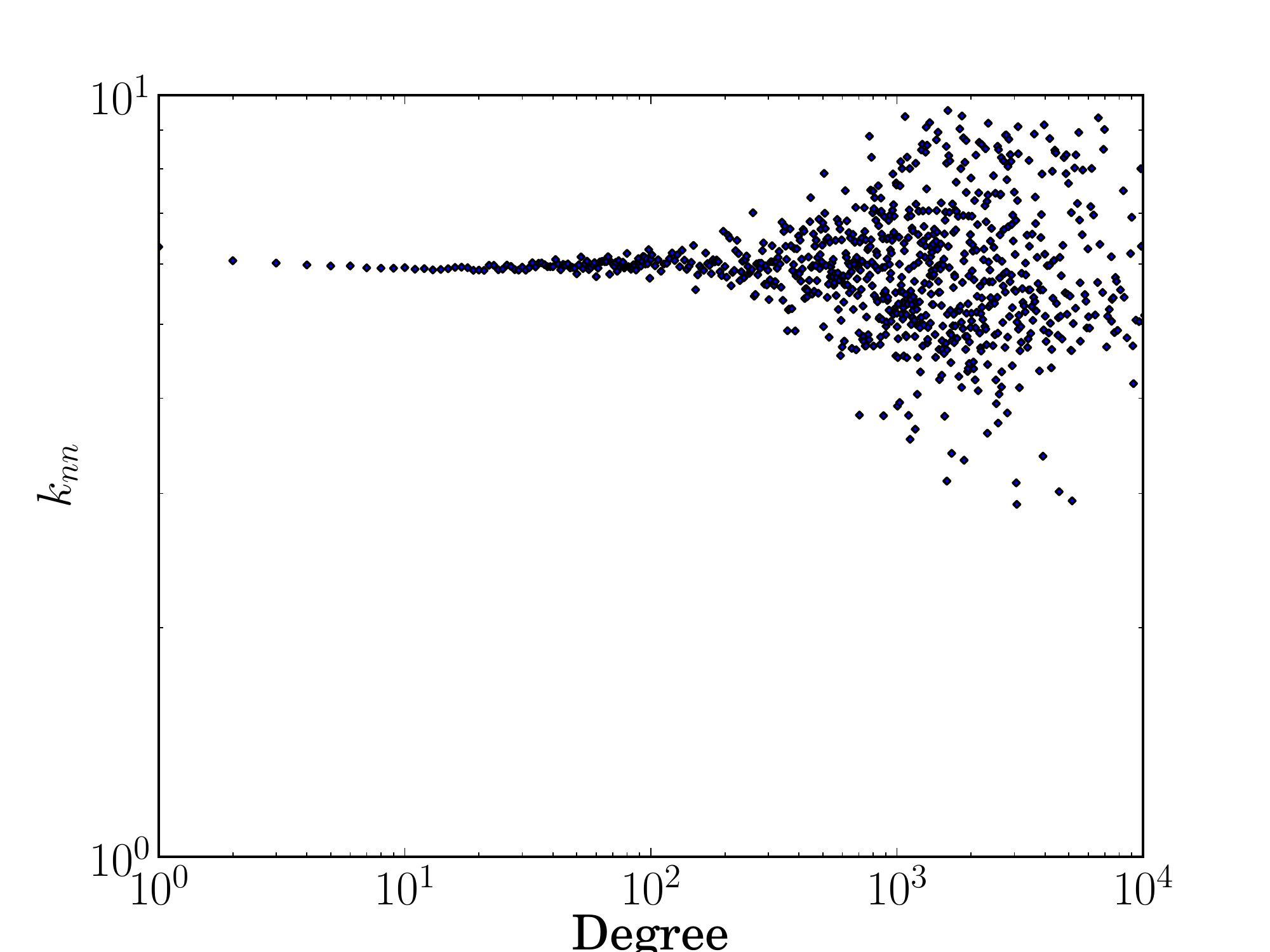}} 
\subfloat[\scriptsize{Evolution of assortativity}]{\includegraphics[width=0.25\textwidth, height=1.5in]{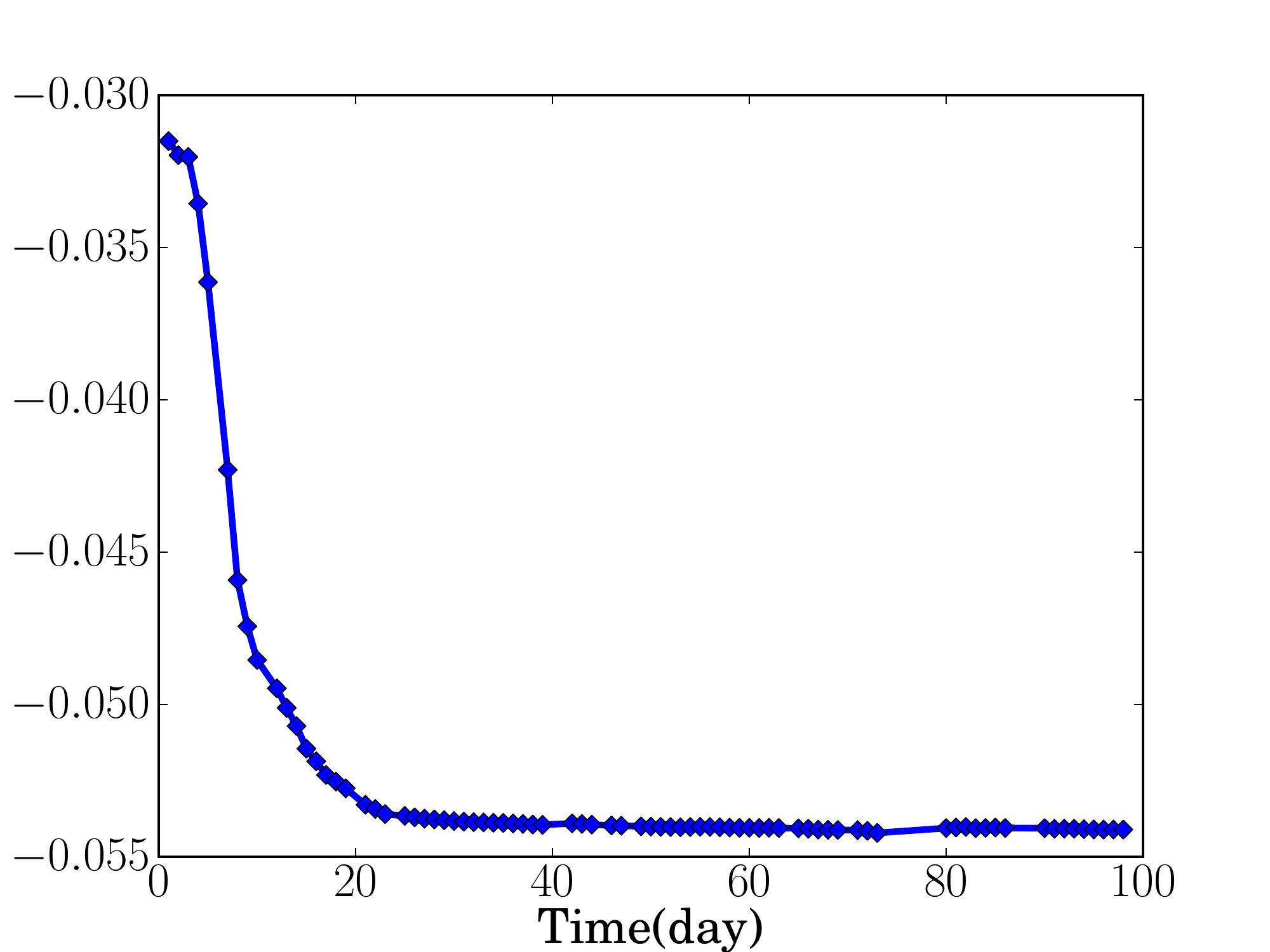}\label{fig:attrijoint-degree-distribution-evo}} 
\tightcaption{(a) Joint degree of attribute nodes: Log-log plot of the social degree  versus the average attribute degree of social neighbors of attribute nodes. (b) The evolution of the attribute assortativity coefficient.} 
\label{fig:attri-joint-degree-distribution}
\vspace{-2mm}
\end{figure}

\myparatight{Joint degree distribution} Next, we  extend the joint degree
distribution (JDD) analysis to attribute nodes.  For each social degree $k$, we
compute $k_{nn}$ as the average attribute degree of social neighbors of
attribute nodes that have social degree $k$.  Intuitively, it captures the
tendency of attribute nodes with high social degree to connect to social nodes
with high attribute degree; i.e., if many nodes share a particular attribute,
then are these nodes likely have many attributes?
Figure~\ref{fig:attri-joint-degree-distribution} shows the $k_{nn}$ function
for attribute JDD and the evolution of the attribute assortativity.
Intuitively, we expect this relationship to be neutral and the result confirms
this intuition; e.g., there are many Google+ users in New York but that does
not imply the people in New York have many attributes. One interesting
observation is that attribute assortativity coefficient evolves slightly differently
compared to social assortativity coefficient
(Figure~\ref{fig:joint-degree-distribution-evo});  it is stable in \PIII
whereas social assortativity decreases significantly.
 

\subsection{Influence on Social Network Structure}
Next, we look at how attributes influence the social structure of the 
 Google+ \san w.r.t the metrics  discussed in   \Section\ref{sec:structure}.

\begin{figure}[t]
\vspace{-0.4cm}
\centering
\subfloat[ \scriptsize{Reciprocity}]{\includegraphics[width=0.25\textwidth, height=1.5in]{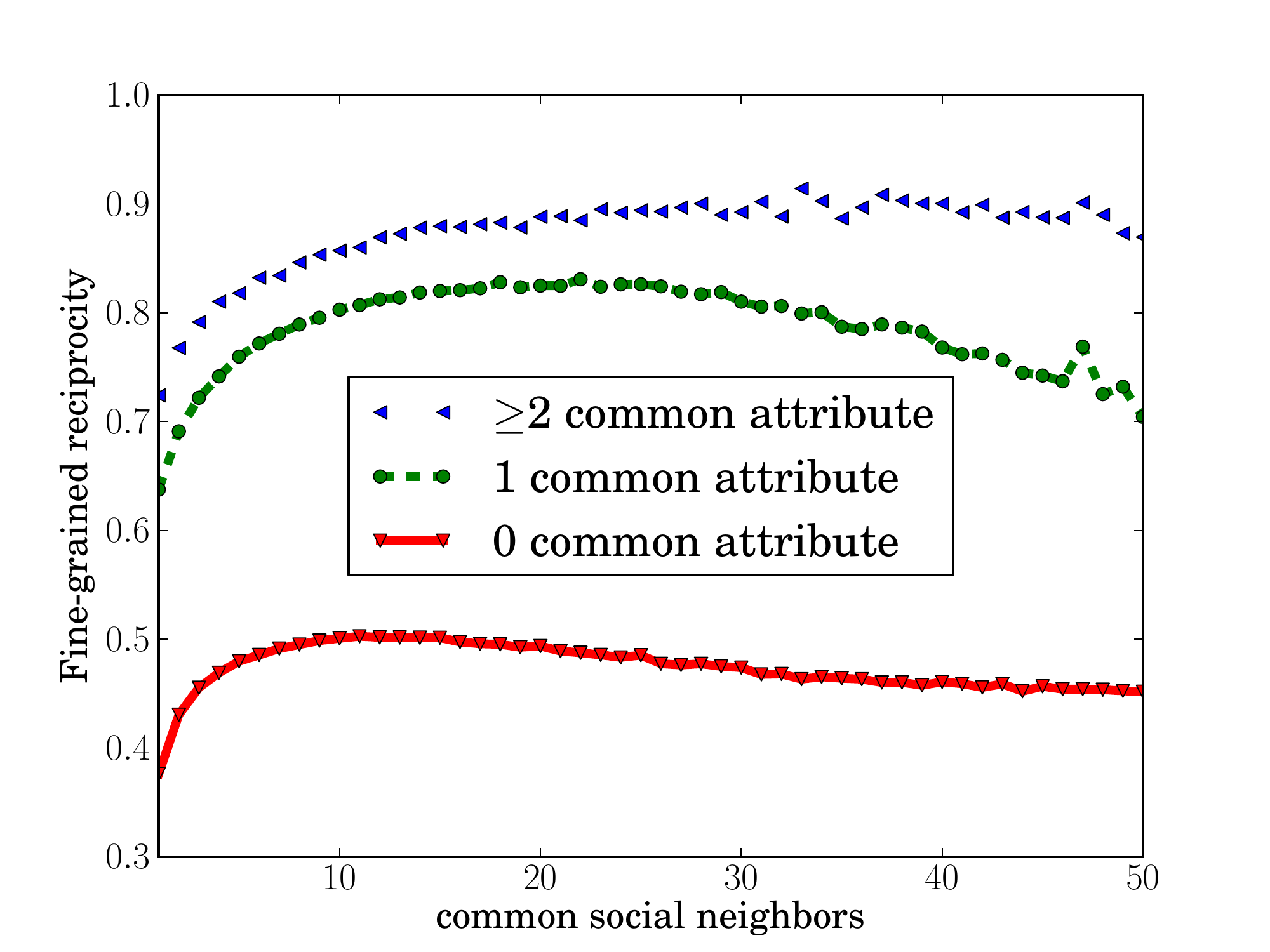}\label{fig:friend-accept}}
\subfloat[ \scriptsize{Clustering coefficients}]{\includegraphics[width=0.25\textwidth, height=1.5in]{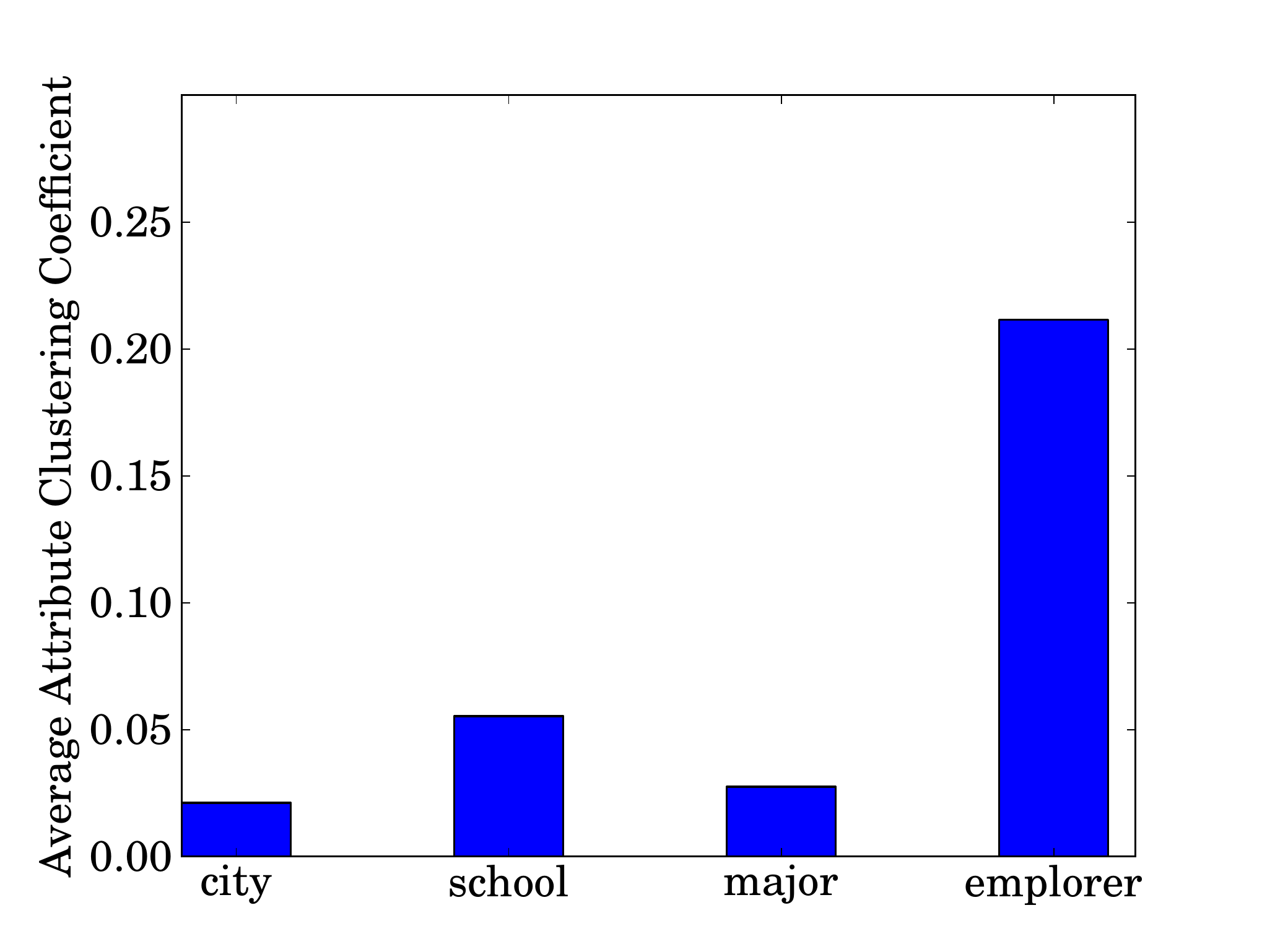}\label{attri-clustering-coefficient-dis}}
\tightcaption{Influence of attribute on reciprocity and clustering coefficients.}
\label{fig:reciprocity-attr}
\vspace{-2mm}
\end{figure}

\myparatight{Reciprocity} We study how the number of common attribute 
neighbors influences reciprocity in conjunction with the number of 
common social neighbors. Let $a$ and $s$ denote the number of attribute 
and social neighbors of a given node, respectively.  
For each  pair $(s,a)$, we compute 
$r_{s,a}$ as the percentage of links that are reciprocal among
all the links whose  endpoints have  $s$ social neighbors and $a$ attribute
neighbors.  

%


To compute this, we look at all \emph{one directional} links at the snapshot
collected halfway and then compute the number of such links that become
\emph{bidirectional} at the last snapshot.  We split these by the number of
common social and attribute neighbors between these nodes at the halfway stage
and  show the $r_{s,a}$ values in Figure~\ref{fig:reciprocity-attr}.  We see that the reciprocity is almost twice
as high for nodes that share common attribute neighbors compared to nodes
without common attributes, regardless of the number of common social neighbors.
While sharing common social neighbors improves link reciprocity,  there is
 a natural diminishing returns property beyond 10 common social neighbors, and
even decreasing for much larger values.  We speculate that nodes 
 sharing too many social neighbors are likely users with many 
 ``weak'' ties. For recent reciprocity prediction problem~\cite{Cheng11, Hopcroft11}, our findings imply that any reciprocity predictor should incorporate node attributes instead of pure social structure metrics.

\myparatight{Clustering coefficient} 
Next, we compute the average attribute clustering coefficient for the 4 attribute types: Employer, School, Major and City.
For example, we compute the attribute clustering coefficients for all attribute nodes belonging to the attribute type Employer, and then average them to obtain the average attribute clustering coefficient for Employer. Figure~\ref{attri-clustering-coefficient-dis} shows that  attribute types vary
in their influence on  forming communities and that users with the same Employer
attributes are much more likely to form communities than users
sharing  other attribute types. This has interesting implications for 
link prediction and attribute inference.

\begin{figure}[htbf]
\vspace{-0.5cm}
\centering
\subfloat[\scriptsize{Employer}]{\includegraphics[width=0.25\textwidth, height=1.8in]{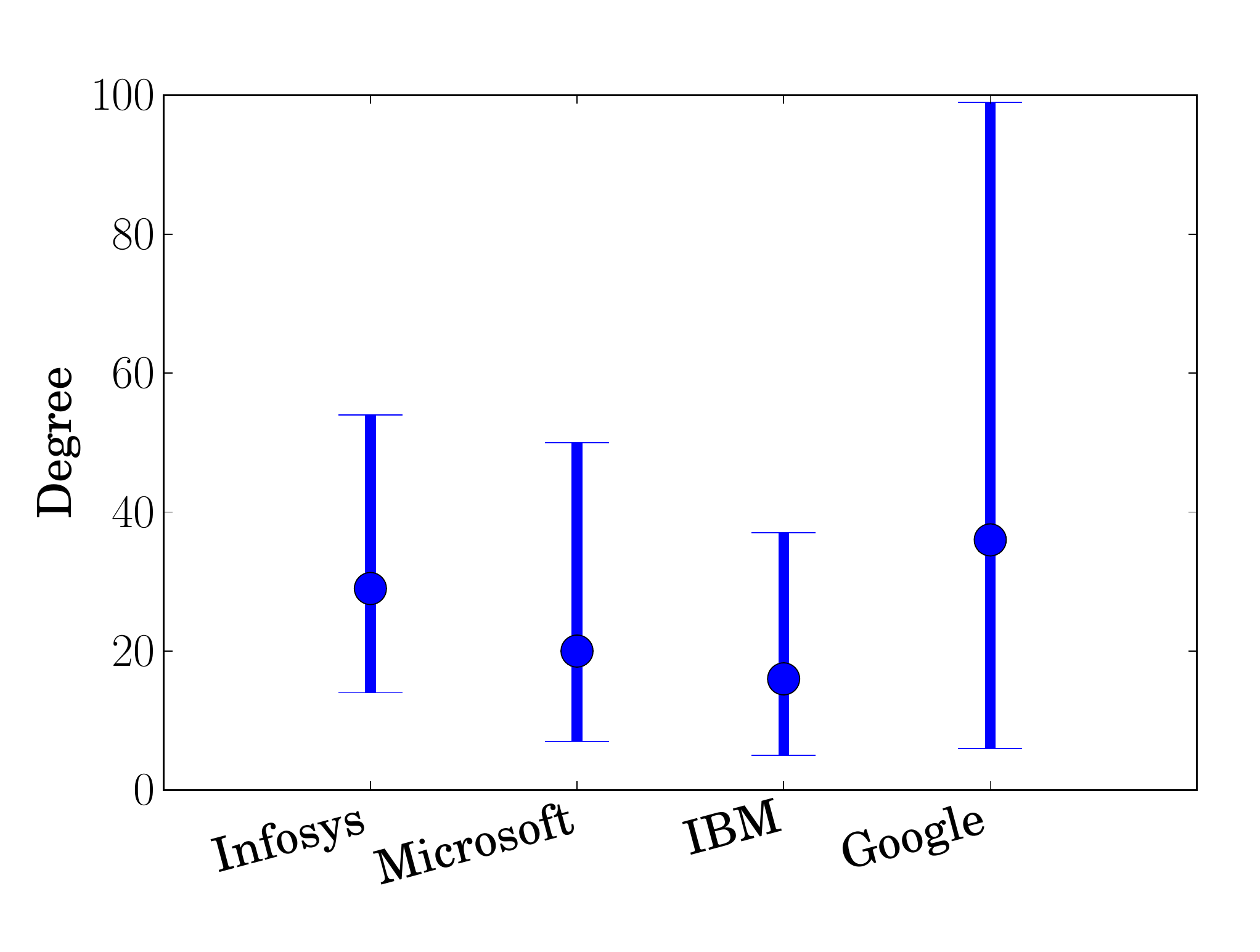}\label{soc-density-distribution-employer}}
\subfloat[\scriptsize{Major}]{\includegraphics[width=0.25\textwidth, height=1.8in]{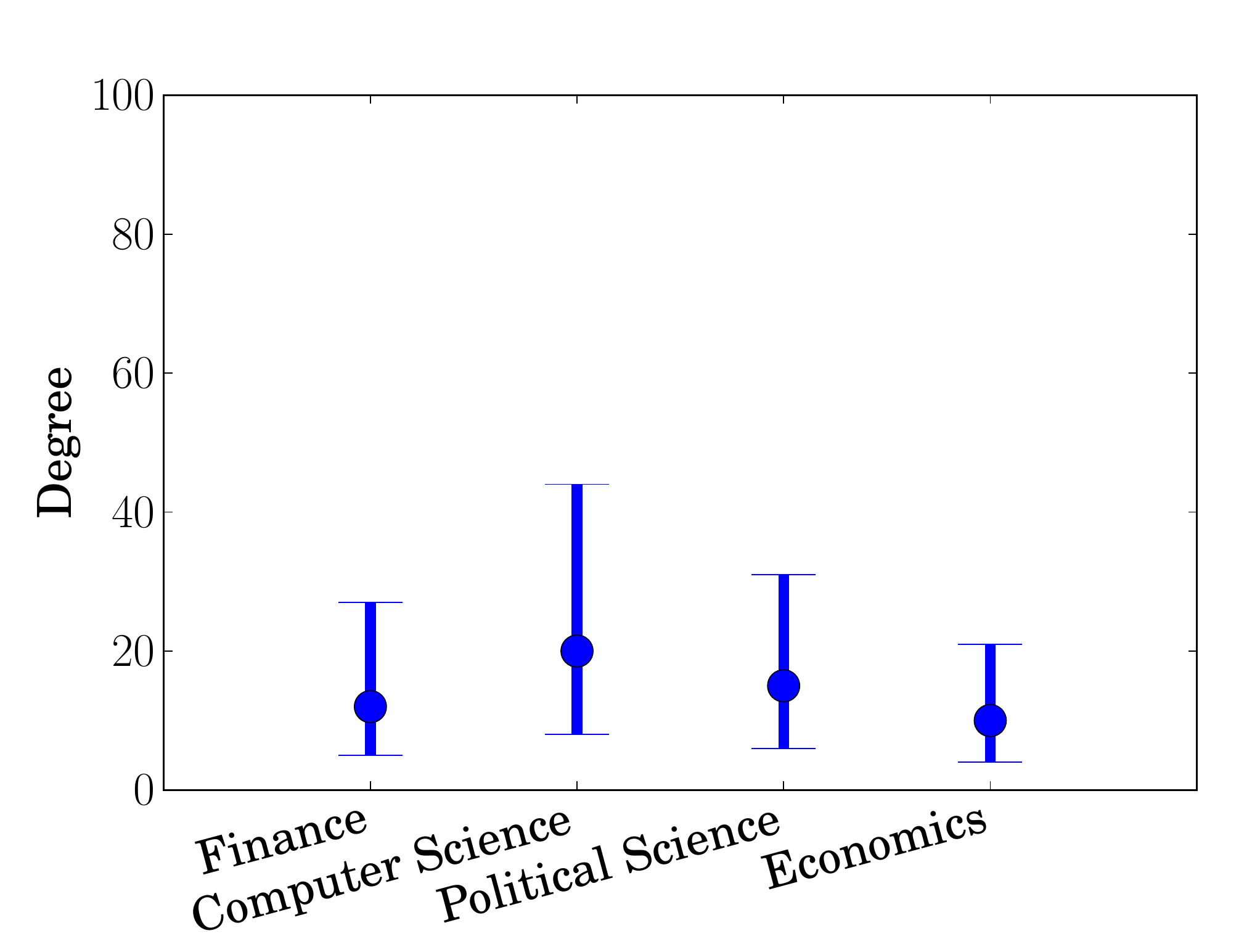}\label{soc-density-distribution-major}}
\tightcaption{Influence of attribute on social degree}
\label{fig:soc-density-distribution}
\vspace{-2mm}
\end{figure}

\eat{
\myparatight{Social Density} 
\ling{need edits.}
To study the influence of attributes on network
density, we compute the average social density for different attribute values
in Figure~\ref{fig:soc-density-distribution}.  For brevity, we only focus on
the employer and major attributes and show the top attribute values 
observed within each category. We see that the 
users in Google and users with computer science major are likely to have 
higher degrees showing a potential relationship between the attribute and 
social density. This could be an artifact of Google+ network, because 
the early adapters of Google+ network consist of large portion of 
Google employees and people in IT industry (both of which may have computer
science major).
}
%
 

%

\myparatight{Degree distribution}  For brevity, we only focus on the Employer
and Major attributes and show the result for the top attribute values
observed within each category. We plot the median, 25$^\mathit{th}$, and
75$^\mathit{th}$  percentile of the social outdegree of nodes that have these
attribute values in Figure~\ref{fig:soc-density-distribution}. We see that the
users with Employer=Google and Major=Computer Science are likely to have higher
degrees. We also computed the full degree distributions for these attribute
values and saw that they follow different lognormal distributions (not shown).
We speculate this could be a specific artifact of the Google+ network as  many
of the  early adopters  likely consist of Google employees and users in the
IT/CS industry.

\subsection{Validation via Subsampling}

One natural question is whether the attributes of 22\% of users we collected is a good representative of the entire attributes. To this end, we use subsampling method to validate our attribute-related results. We use attribute clustering coefficient distribution with respect to node degrees as an example, and observe similar results for other metrics. For each user with attributes, we remove her attributes with probability 0.5, from which we obtain a subsampled \san. Then we calculate the attribute clustering coefficient distributions for the original and this subsampled \sanplural. Figure~\ref{origin-sample} shows that the results of the original and subsampled \sanplural are almost identical. Given the assumption that whether a user fills in her attributes is a random and independent event, our results demonstrate that the attributes of 22\% of users is a representative sample of the attributes of all the users.    

\subsection{Summary of Key Observations and Implications}
In this section, we studied the attribute 
 structure of the Google+ \san and how such attribute structure impacts the social structure. Our key observations are: 
\begin{packeditemize}

\item While some attribute metrics mirror their social 
 counterparts (e.g., diameter), several show distributions and trends that are significantly different (e.g., clustering coefficient, attribute 
degree). These observations will guide us to design models for \san.

\item We confirm that attributes have interesting impact on the social
structure. e.g., nodes are likely to have higher reciprocity if they share common attributes. These findings have various implications. For instance, reciprocity predictor should incorporate node attributes.   

\item We also observe that some attribute types naturally have stronger
influence than others.  For example, users sharing the same employer have higher probability to be linked
compared to users sharing the same city. Data mining tasks such as link prediction and attribute inference should potentially benefit from these findings.

\end{packeditemize}

\section{A Generative Model For \san}
\label{sec:model}
  
From the previous sections, we have seen novel phenomena in the social and
attribute structure of the Google+ \san and that the attribute structure
impacts the social structure significantly.  A natural question is whether we
can create an accurate generative network model that can reproduce both
the social and attribute structures we observe. Such a generative model can help us understand the growing mechanism of \san, and allow other applications such as network extrapolation and sampling, network visualization and compression, and network anonymization~\cite{Leskovec10-JMLR}.

Prior work on generative models focus primarily on the social
structure~\cite{Barabasi99, Kleinberg99, VAZQUEZ03, Leskovec05, Leskovec08}.
Consequently, these approaches cannot model the attribute structure or their
impact on social structure.  To address this gap, we provide a new generative model taking into account the attribute structure from
first principles rather than overlaying it after-the-fact.  To this end, we
extend a prior generative model~\cite{Leskovec08}, using
attribute-augmented models for link generation and addition , which are key
building blocks for such generative models. As we will show, this provides more
realistic synthetic \san that closely matches the Google+ \san.  



\begin{figure}[t!]
\centering
\subfloat[\scriptsize{PAPA model}]{\includegraphics[width=0.25\textwidth, height=1.5in]{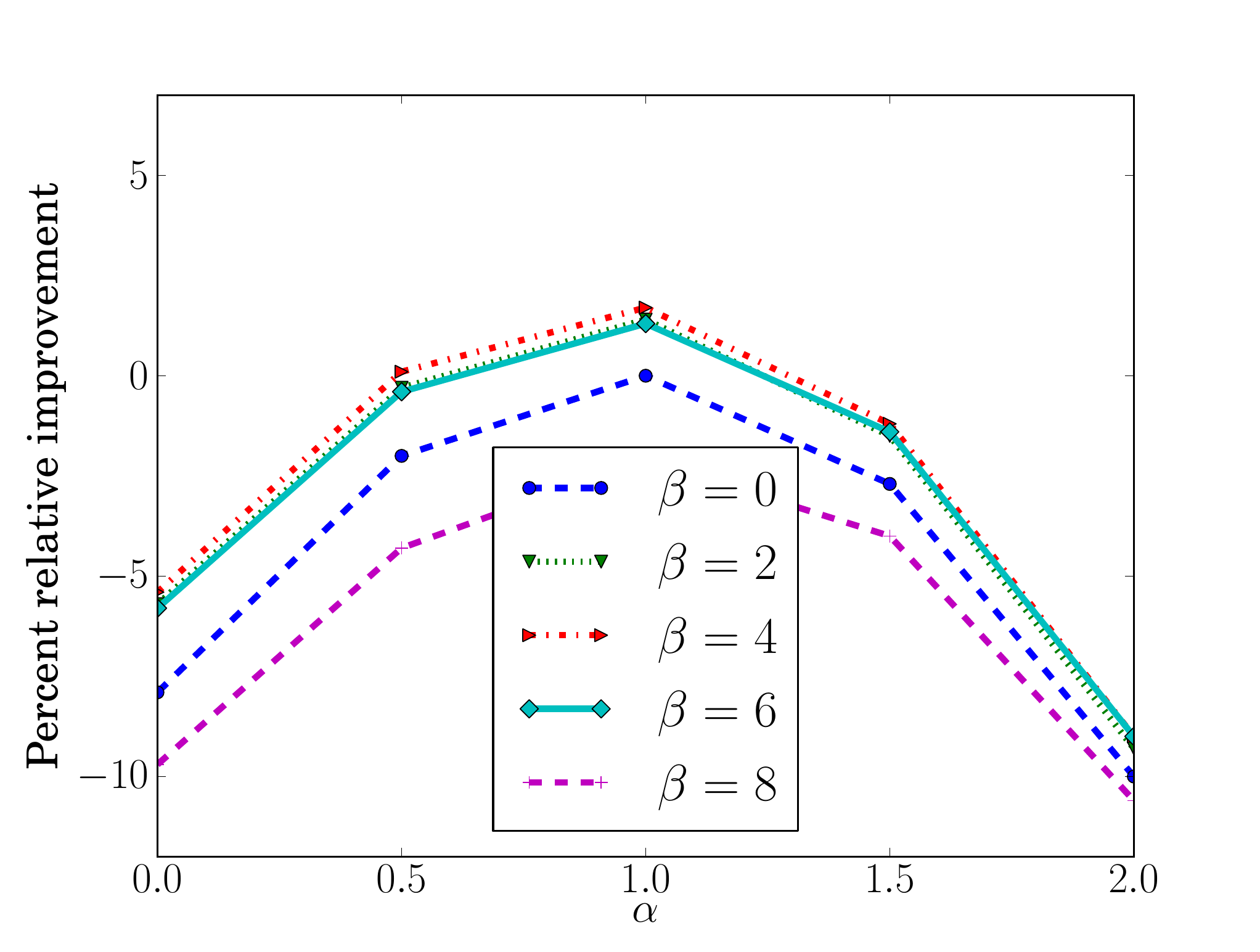}\label{fig:request-a}}
\subfloat[\scriptsize{LAPA model}]{\includegraphics[width=0.25\textwidth, height=1.5in]{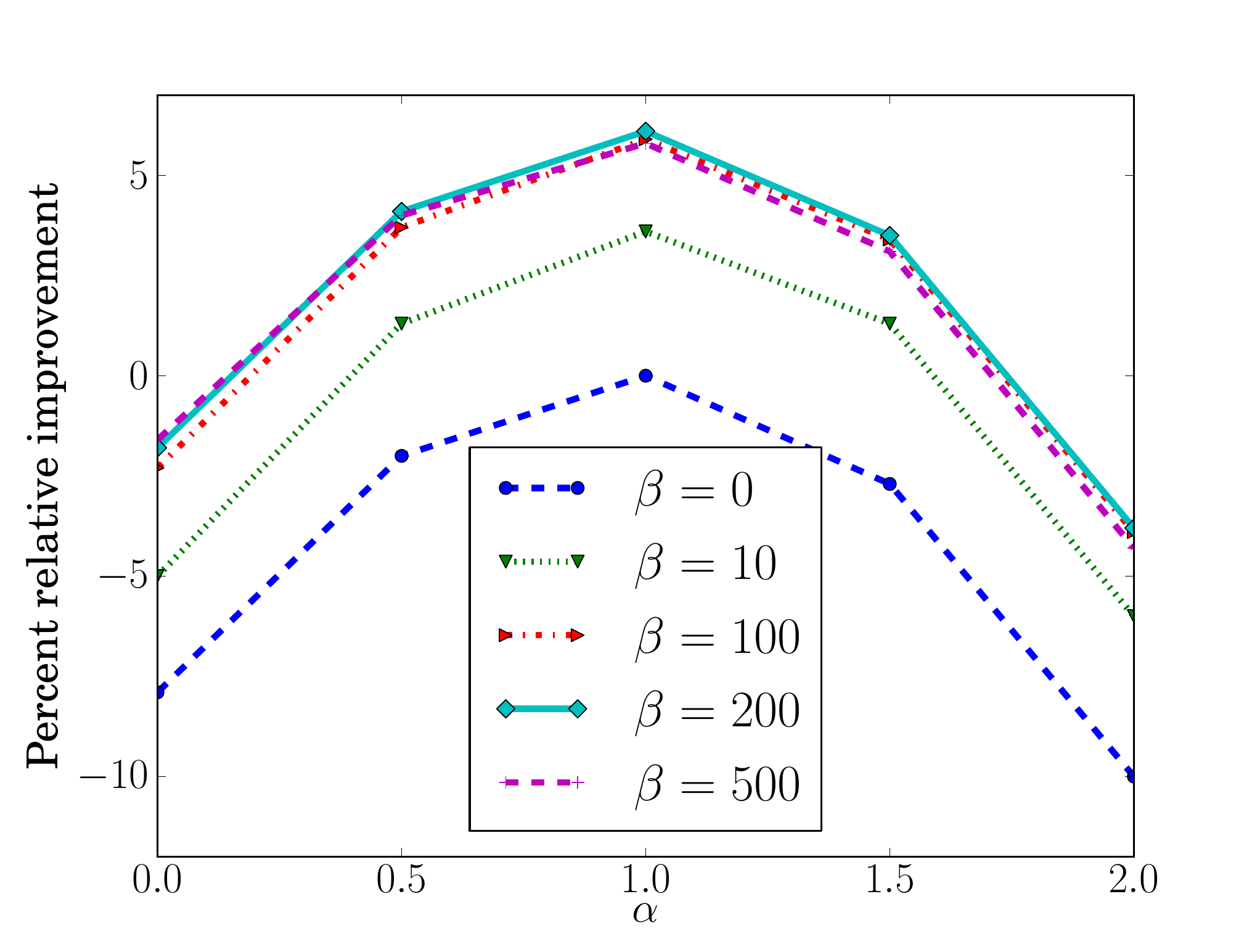}\label{fig:request-b}}
\tightcaption{Comparison between Power Attribute Preferential Attachment 
(PAPA) and Linear Attribute Preferential Attachment (LAPA) models.  
All result numbers are percentage of relative improvements
over the loglikelihood of the PA model, i.e., with $\alpha$ =1 and $\beta$ = 0.}
\label{fig:request-pa}
\end{figure}

\subsection{Building Block 1: \\Attribute-Augmented Preferential Attachment}
 Leskovec et al.  
 showed that the Preferential Attachment (PA)~\cite{Barabasi99} is 
 a suitable choice for creating edges~\cite{Leskovec08}. 
 The key idea in PA is that a new node $u$ is likely to connect 
to an existing node $v$ with a probability proportional to $v$'s degree.  
 As we saw earlier, users who share attributes are also 
 more likely to be connected. Thus, we consider two ways to augment the PA model: 
\begin{itemize}
\item \emph{Power Attribute Preferential Attachment (PAPA):} \\  $f(u,v)$ $\propto$ $d_i(v)^\alpha (1+a(u,v)^\beta)$
\item \emph{Linear Attribute Preferential Attachment (LAPA):}\\ $f(u,v)$ $\propto$ $d_i(v)^\alpha (1+\beta \cdot a(u,v))$
\end{itemize}

Here, $f(u,v)$ is the probability with which social node $u$ adds a link to
social node $v$, $d_i(v)$ is the indegree of $v$ and  $a(u,v)$ is the number of
common attributes that social nodes $u$ and $v$ share.\footnote{In a more
general setting, we can also weight attribute types differently; e.g., Employer
 is stronger than City.}  Notice
that when $\alpha$ = $\beta$ = 0, both reduce to a uniform distribution (i.e.,
$v$ is sampled uniformly at random) and when $\alpha$=1,$\beta$=0 both  reduce
to the PA model.

%

The relative improvement of a model with parameter $\alpha,\beta$ over 
the PA model is defined as $\frac{l_{PA} - l(\alpha, \beta)}{l_{PA}}$, 
where $l$ denote the log-likelihood of the model with respect to the  
empirically observed Google+ \san. 
Figure~\ref{fig:request-pa} shows the relative improvements  of these models
over the PA model for varying values of $\alpha,\beta$. First,
LAPA models perform better than PAPA models, which indicates that attribute
likely influence friend requests in a linear way.  Second, the PA model
($\alpha$ =1, $\beta$ = 0) is 7.9\% better than a uniform random model
($\alpha$ =0, $\beta$ = 0).  A LAPA model with $\alpha$ =1 and $\beta$ = 200
 achieves a further 6.1\% improvement over the PA model.  Third, $\alpha=1$
achieves the best loglikelihood for any given $\beta$, which indicates that
 social degree has a linear effect on friend requests.  In summary, we conclude
that there is a combined linear effect of both social degree and attributes. 

\eat{
\ling{We no longer talk about born degree.}
All previous dynamic models assume the born friends are chosen by either a PA model~\cite{Barabasi99, Kleinberg99, Kumar00, Leskovec08, Zheleva09-evo} or a uniformly random model~\cite{VAZQUEZ03, Sala10, Leskovec05}. Our observations inform us that these models should be replaced by LAPA model.
}


\eat{
\begin{figure}[t]
\centering
\includegraphics[width=0.25\textwidth, height=1.5in]{triangle-closing-model}
\captionsetup{font=small,labelfont=bf}         
\caption{\bf Percent relative improvements over the loglikelihood of the baseline model as a function of the attribute weights. \ling{The weight
has very small impact on the performance. I would suggest dropping
it and reduce the figure into two numbers.}}
\label{fig:triangle-closing-model}
\end{figure}
}
\subsection{Building Block 2: \\Attribute-Augmented Triangle-Closings}

 Triangle closing, where a node $u$  selects a node $v$ from its 2-hop
neighbors and adds an edge,  is an essential part of many generative network
models~\cite{Leskovec08, Zheleva09-evo, Leskovec05, VAZQUEZ03, Sala10,
Toivonen09}.  We explore if node attributes can improve triangle closing.

In the context of \san, we can consider two types of triangle-closing: one is
closing a triangle with no attribute node involved (e.g., $u_4 \rightarrow u_2$
in Figure~\ref{fig:san}), and the other is closing a triangle which includes an
attribute node (e.g., $u_1 \rightarrow u_2$ in Figure~\ref{fig:san}).
Following prior work, we refer them as \emph{triadic} and \emph{focal} closure
respectively~\cite{Kossinets06}.  
 In the friend requests we observe in Google+,  84\% percent are triadic
(common friend), 18\% percent are focal (common attribute), and 15\% percent
are cases where the nodes share both common friends and common attributes
(e.g., $u_6 \rightarrow u_5$ in Figure~\ref{fig:san}).

  
 This suggests the importance of incorporating attributes in the triangle 
 closure. To this end, we consider three models:
\begin{packeditemize}
\item \emph{Baseline}: Select a social neighbor $v$ within a 2-hop radius uniformly at random.
\item \emph{Random-Random (RR)}: Select a social neighbor $w\in \Gamma_s(u)$ uniformly at random, and then select a social neighbor $v\in \Gamma_s(w)$ uniformly at random which is shown to have very 
 good performance in previous work~\cite{Leskovec08}.
\item \emph{Random-Random-\san (RR-\san)}: select a neighbor $w\in \Gamma_s(u) \cup  \Gamma_a(u)$ uniformly at random, and then select a social neighbor $v \in \Gamma_s(w)$ uniformly at random.\footnote{We also tried 
 a weighted model where we select neighbors proportional to link weights. For brevity, 
 we do not show this because it performs similarly.}
\end{packeditemize}
We compare these models using friend requests that are triadic closures, 
focal closures, or both. Our experimental results confirm that
RR model performs 14\% better than the Baseline model~\cite{Leskovec08}, 
and our RR-\san model performs 36\% better than RR model. 
 This confirms that  attributes  play a significant  role in the 
triangle-closing phenomenon as well and has natural implications 
 for applications such as link prediction and friend recommendation.



\begin{algorithm}[t!]
\small
\SetAlgoLined
\DontPrintSemicolon
T, simulated time steps\;
\emph{Initialization}.\;
\For{$1\leq t \leq T$}{
\emph{Social node arrival}. Sample a set of new social nodes $V_{t, new}$.\;
\For{$v_{new}\in V_{t, new}$} {
\emph{Attribute degree sampling.} Sample the number of attributes $n_a(v_{new})$ for $v_{new}$ from a lognormal distribution. \;
\For{$1\leq i \leq n_a(v_{new})$}{
\emph{Attribute linking}. \;
}

\emph{First outgoing linking}.

\emph{lifetime sampling}. \;
\emph{sleep time sampling}.\;
}
Collect woken social nodes $V_{t, woken}$. \;
\For{$v_{woken}\in V_{t, woken}$} {
\emph{Outgoing linking}. \;
\emph{sleep time sampling}.\;
}
}
\caption{Social-Attribute Network Model}
\label{alg:model}
\vspace{-1mm}
\end{algorithm}

\subsection{Our Generative Model for \san}
Our stochastic process models several key aspects of \san evolution:
node joining, how nodes issue outgoing links and receive incoming links, and
how they link to attribute nodes.  The key differences from prior work~\cite{Leskovec08}  are the
two building blocks we described earlier:  Linear Attribute Preferential
Attachment (LAPA) and Random-Random-\san (RR-\san) triangle-closing. 

%
 Here, nodes arrive at some pre-determined rate. On arrival, each node picks an
initial set of attributes and  social neighbors (using the LAPA model). After
joining the network, each node subsequently ``sleeps'' for some time, wakes up,
and adds new links based on the RR-\san model.  
We describe the model formally in Algorithm~\ref{alg:model}  and discuss 
each step next. From the analysis below, we find that the key step for 
generating lognormal social \emph{outdegree} distribution is to make 
the lifetime of nodes follow a \emph{truncated normal distribution}.

\begin{trivlist}
\item \textbf{Initialization:} The \san is initialized with a few social and
attribute nodes and links. We observed that the starting point has no
detectable influence when the number of initialization nodes is small compared
to the overall network. We currently use a complete social-attribute network with 5 social nodes and 5 attribute nodes.

\item \textbf{Social node arrival:}  Social nodes arrive as predicted by a
node arrival function $N(t)$, which could be estimated from real social
networks. In our simulations, we simply let $N(t)=1$  modeling each node 
 arrival as a discrete time step.

\item \textbf{Attribute degree:} Each node picks some number of
attributes sampled from a lognormal distribution with mean $\mu_{a}$ and
variance $\sigma_{a}^2$.

\item \textbf{Attribute linking:} Each new social node $v_\mathit{new}$ with
$n_a(v_\mathit{new})$ attributes, we connect it to $n_a(v_\mathit{new})$ attribute nodes with
the stochastic process defined as follows: for each attribute, with probability
$p$, a new attribute node $a$ is generated; otherwise an existing attribute
node $a$ is chosen with probability proportional to its social degree.

\item \textbf{First outgoing links:} Each new node issues an outgoing link to a
social node according to the LAPA model.

\item \textbf{Lifetime sampling:} The lifetime $l$ of $v_\mathit{new}$ is
sampled from a truncated normal distribution, i.e., $p(l)
\propto \text{exp}(-\frac{(l-\mu_l)^2}{2\sigma_l^2})$ for $l \geq 0$. 
(Prior models use an
exponentially distributed lifetime value~\cite{Leskovec08, Zheleva09-evo}.)

\item \textbf{Sleep time sampling:} Sleep time $s$ of any node $v$ with
outdegree $d_o$ can be sampled from any distribution with mean $m_s/d_o$. Our
model only depends on mean sleep time. The intuition of making mean sleep time
reversely proportional to outdegree is that a node with larger outdegree has
higher tendency to issue outgoing links. (Prior models assume a power-law with cutoff distributed lifetime value~\cite{Leskovec08, Zheleva09-evo}.)


\item \textbf{Outgoing linking}. Each woken social node $v_\mathit{woken}$ issues a new outgoing link
according to our RR-SAN triangle-closing model.

\end{trivlist}

\eat
{
\textbf{Discussion.}
\ling{Neil, please add and revise any idea here.}
In a high level, our model differs from Zheleva's 
model~\cite{Zheleva09-evo} in following way. Our model concerns about 
static attributes that come with users when joining the network and 
changing slowly over time. This is the reason why our model lets nodes 
create and acquire attributes when joining the network, and no 
longer do so afterward. On the contrary, Zheleva's model concerns about 
dynamic attributes acquired after joining the network that may change 
quickly over time. This is the reason why their model lets nodes keep 
creating and acquiring new attributes after joining the network. 
While our model mainly explore the impact of 
static attributes on social structure, Zheleva's model mainly explore the 
impact of social structure on dynamic attributes. It is an interesting future
work to combine both models together to study the mutual impacts among 
static attributes, social structure and dynamic attributes.
}

\subsection{Theoretical Analysis}

By design, the  attribute degree distribution of social nodes follows a
lognormal distribution. Next, we show via analysis that the outdegree  of social nodes  and the social degree
of attribute nodes follow a lognormal and
power-law distribution respectively. 
For brevity, we provide a high-level sketch of the proofs.

Let $\phi(x)$ and $Φ \Phi(x)$ denote the probability density function and
cumulative density function of standard normal distribution. Let
$\gamma_l =-\frac{\mu_l}{\sigma_l}$,
$g(\gamma)=\frac{\phi(\gamma)}{1-\Phi(\gamma)}$ and $\delta(\gamma) =
g(\gamma)(g(\gamma)-\gamma)$.


\begin{theorem}
\label{theorem:cutoff}
If the sleep time is sampled from some distribution with mean $m_s/d_o$, then the social out degrees of \sanplural generated by our model follow a lognormal distribution with mean $(\mu_l + \sigma_l g(\gamma_l))/m_s$ and variance $\sigma_l^2 (1-\delta (\gamma_l))/m_s^2$.
\end{theorem}

\begin{proof}
For any social node $v$, assume its final outdegree is $D_o$, then we have
$$\sum_{d_o=1}^D s(d_o) \leq l,$$
where $s(d_o)$ is the random sleep time whose mean is $m_s/d_o$. 
Thus, with mean-field approximation, we obtain
$$m_s\sum_{d_o=1}^D \frac{1}{d_o} \leq l.$$
Moreover, according to Euler's asymptotic analysis on harmonic series, we have 
$$\sum_{d_o=1}^D \frac{1}{ d_o} \approx \text{ln}D_o.$$
That is, $\text{ln} D_o \approx l /m_s$.
Lifetime $l$ is also a normal distribution truncated for $l\geq 0$, thus having mean $\mu_l + \sigma_l g(\gamma_l)$ and variance $\sigma_l^2(1-\delta (\gamma_l))$.
Thus, $\text{ln} D_o$ follows a truncated normal distribution with mean $\mu_o = (\mu_l + \sigma_l g(\gamma_l))/m_s$ and variance $\sigma_o^2 = \sigma_l^2 (1-\delta (\gamma_l))/m_s^2$. So $D_o$ follows a lognormal distribution with mean $\mu_o$ and variance $\sigma_o^2$.
\end{proof}

Next, we derive the distribution of social degree of attribute nodes using
mean-field rate equations~\cite{Barabasi99Phy}.

\begin{theorem}
\label{theorem:attri}
The social degrees of attribute nodes in the \sanplural
 generated by our model follow a power-law 
distribution with exponent $\frac{2-p}{1-p}$.
\end{theorem}

\begin{proof}

Without loss of generality, we assume one attribute link joins the \san at each
discrete time step. Let $D_i$ denote the social degree of the attribute node
$i$ that joins the network at time $t_i$. According to the stochastic process
in our algorithm, we have
\begin{displaymath}
\frac{dD_i}{dt}=\frac{(1-p)D_i}{\sum_i D_i}=\frac{(1-p)D_i}{t + m_0}
\end{displaymath} 
,where $m_0$ is the initial number of attribute links. Solving this ordinary differential equation with initial condition $D_i = 1$ at $t=t_i$ gives us
$$D_i = (\frac{t+m_0}{t_i+m_0})^{(1-p)}.$$
So the probability of $D_i < D$ is
$$Pr(D_i < D) = 1-Pr( t_i+m_0 \leq (t+m_0) D^{-\frac{1}{1-p}}).$$
According to our model, $Pr(t_i)$ has a uniform distribution over 
the set $\{1,2,\cdots, t\}$. Thus we obtain
$$Pr(D_i < D) = 1-\frac{(t+m_0) D^{-\frac{1}{1-p}} - m_0}{t}.$$
Then the distribution of $D_i$ can be calculated as
$$Pr(D) = \frac{dPr(D_i \leq D)}{d D} =\frac{t+m_0}{t(1-p)}D^{-\frac{2-p}{1-p}}.$$
As $t \rightarrow \infty$, we obtain $Pr(D) \propto D^{-\frac{2-p}{1-p}}$. So the social degrees of attribute nodes follow a power-law distribution with exponent $\frac{2-p}{1-p}$.
\end{proof}

Mitzenmacher~\cite{Mitzenmacher03} did a comprehensive study on generative models (e.g., PA, multiplicative models, random monkey) for power-law and lognormal distributions. In this work, we have proposed two new generative models.

\begin{figure*}[t]
\vspace{-0.5cm}
\centering
\subfloat[\scriptsize{Social outdegree}]{\includegraphics[width=0.25\textwidth, height=1.5in]{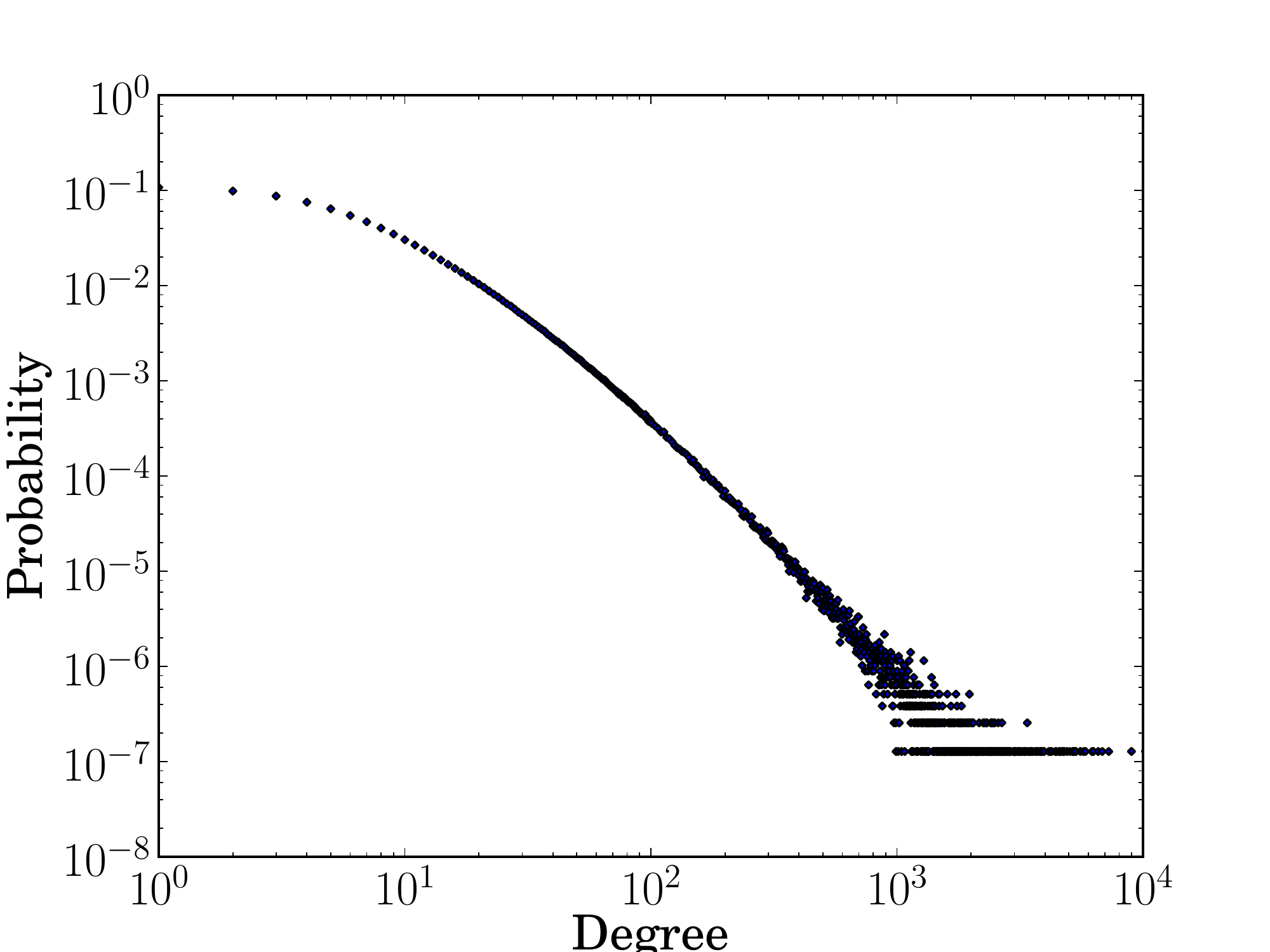}\label{gen-soc-out}}
\subfloat[\scriptsize{Social indegree}]{\includegraphics[width=0.25\textwidth, height=1.5in]{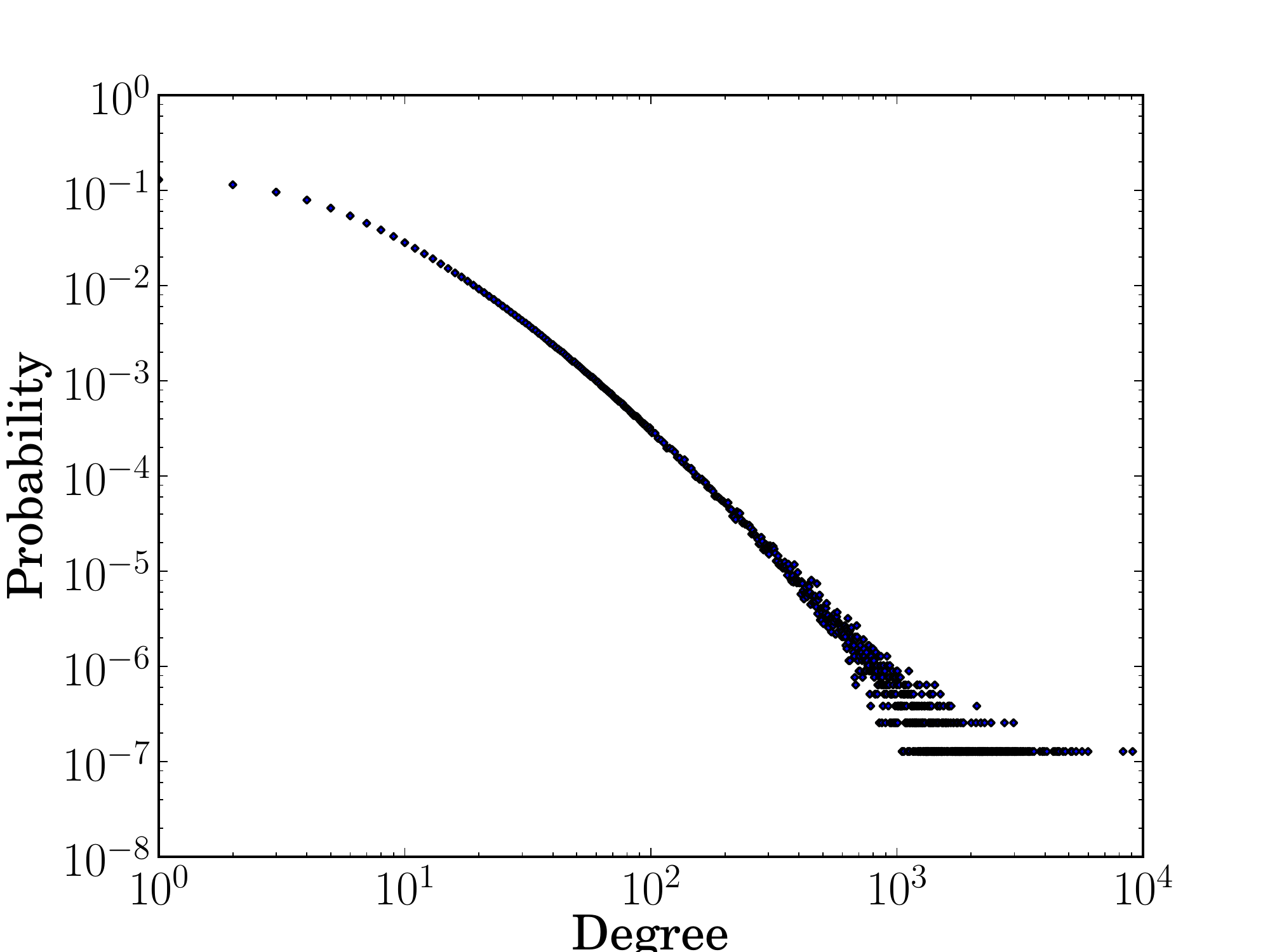}\label{gen-soc-in}}
\subfloat[\scriptsize{Attribute degree of social nodes}]{\includegraphics[width=0.25\textwidth, height=1.5in]{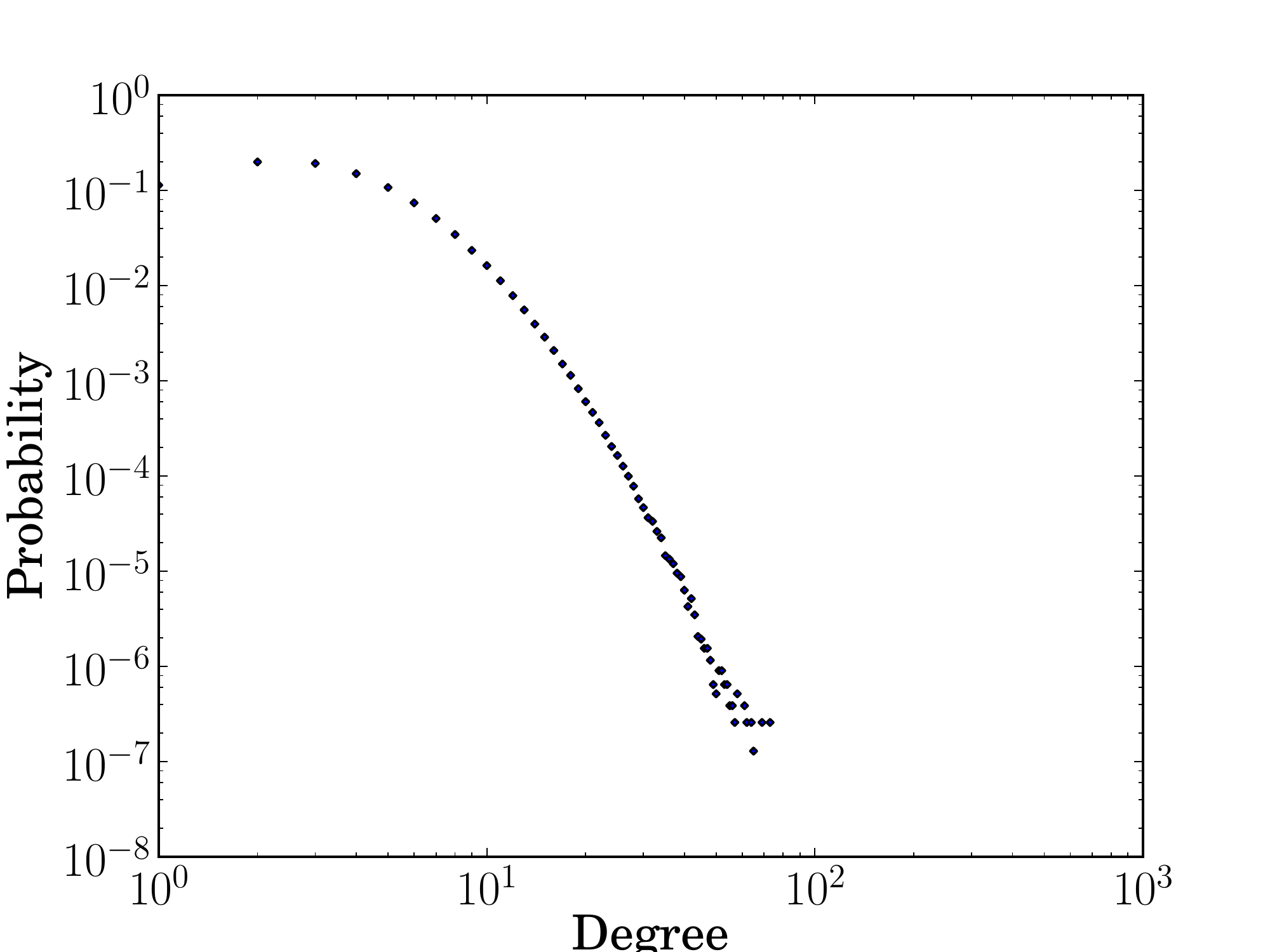}\label{gen-soc-attri}}
\subfloat[\scriptsize{Social degree of attribute nodes}]{\includegraphics[width=0.25\textwidth, height=1.5in]{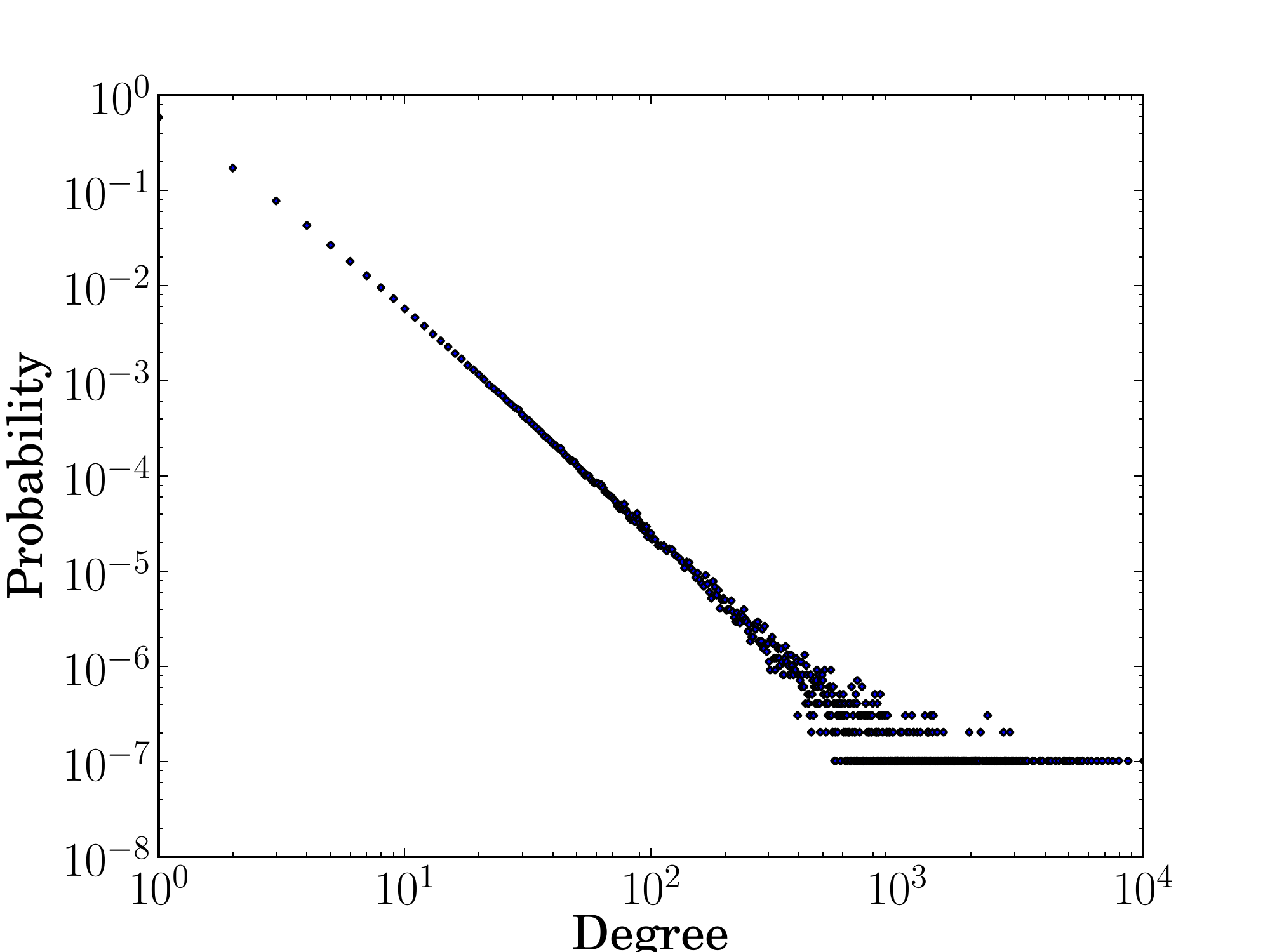}\label{gen-attri-soc}} 

\eat{
\subfloat[\scriptsize{Social outdegree}]{\includegraphics[width=0.25\textwidth, height=1.5in]{gen-baseline-soc-soc-out}\label{gen-baseline-soc-out}}
\subfloat[\scriptsize{Social indegree}]{\includegraphics[width=0.25\textwidth, height=1.5in]{gen-baseline-soc-soc-in}\label{gen-baseline-soc-in}}
\subfloat[\scriptsize{Attribute degree of social nodes}]{\includegraphics[width=0.25\textwidth, height=1.5in]{gen-baseline-attri-soc}\label{gen-baseline-attri-soc}}
\subfloat[\scriptsize{Social degree of attribute nodes}]{\includegraphics[width=0.25\textwidth, height=1.5in]{gen-baseline-soc-attri}\label{gen-baseline-soc-attri}}
}
\subfloat[\scriptsize{Social outdegree}]{\includegraphics[width=0.25\textwidth, height=1.5in]{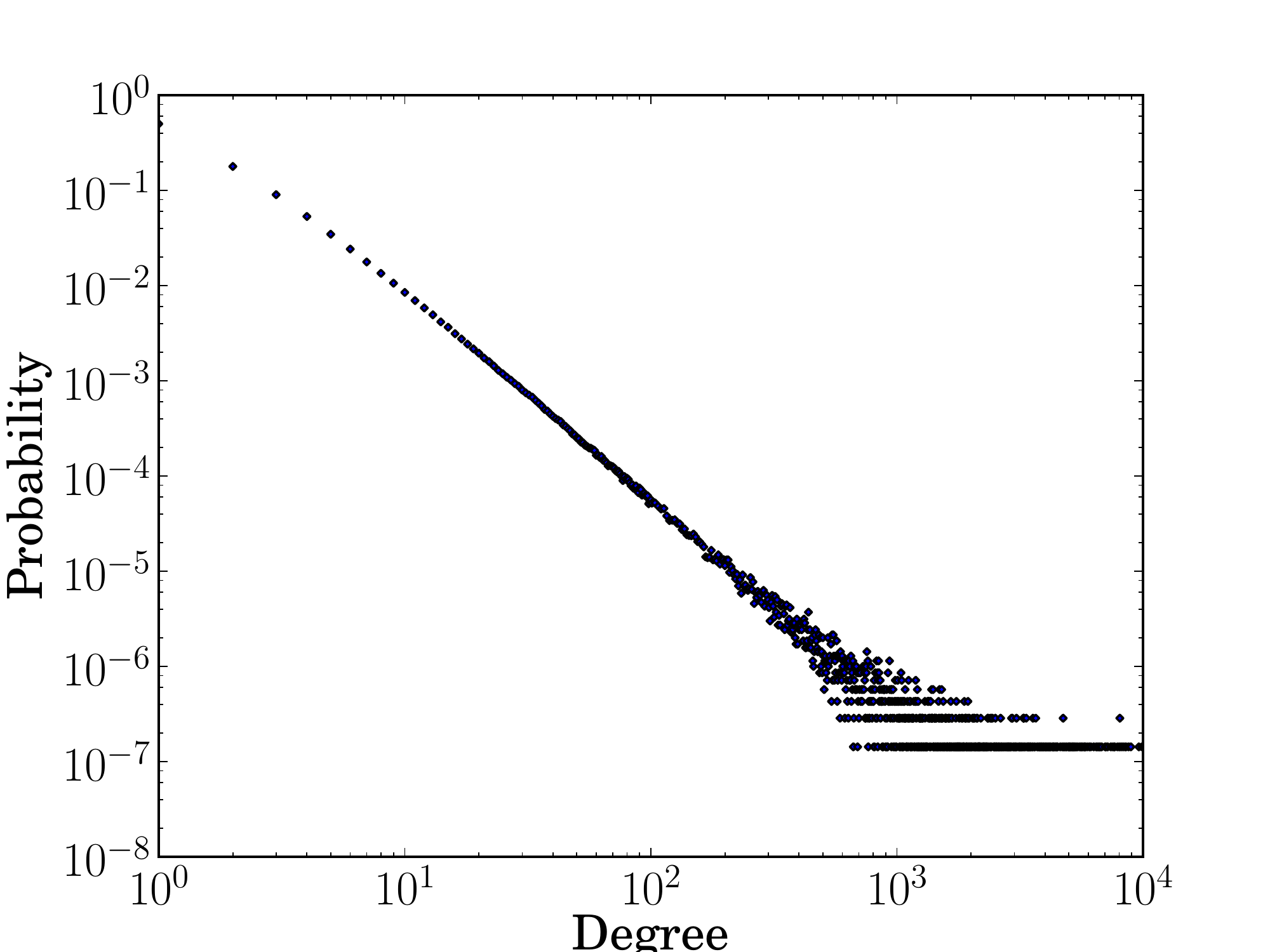}\label{gen-zheleva-soc-out}}
\subfloat[\scriptsize{Social indegree}]{\includegraphics[width=0.25\textwidth, height=1.5in]{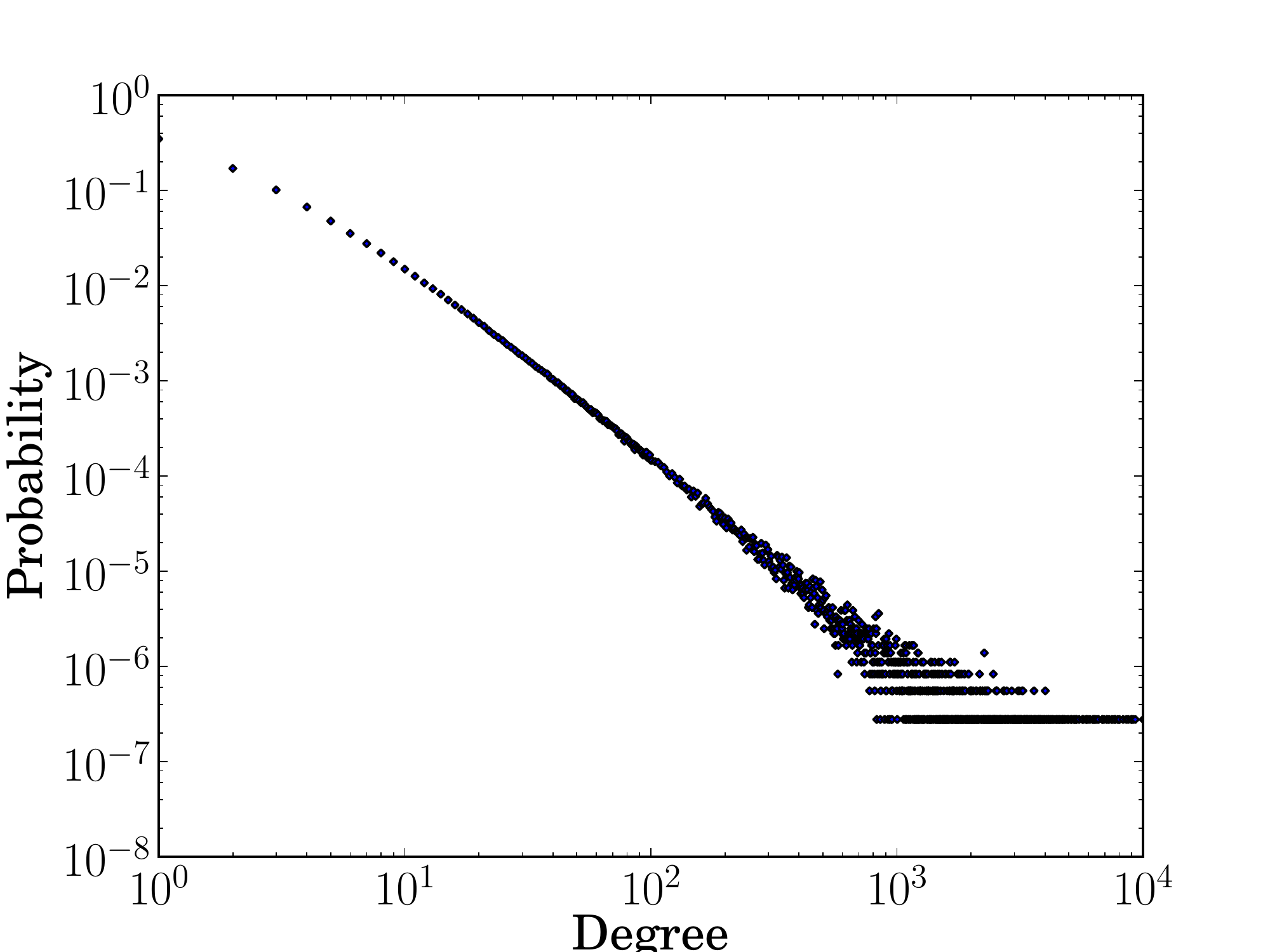}\label{gen-zheleva-soc-in}}
\subfloat[\scriptsize{Attribute degree of social nodes}]{\includegraphics[width=0.25\textwidth, height=1.5in]{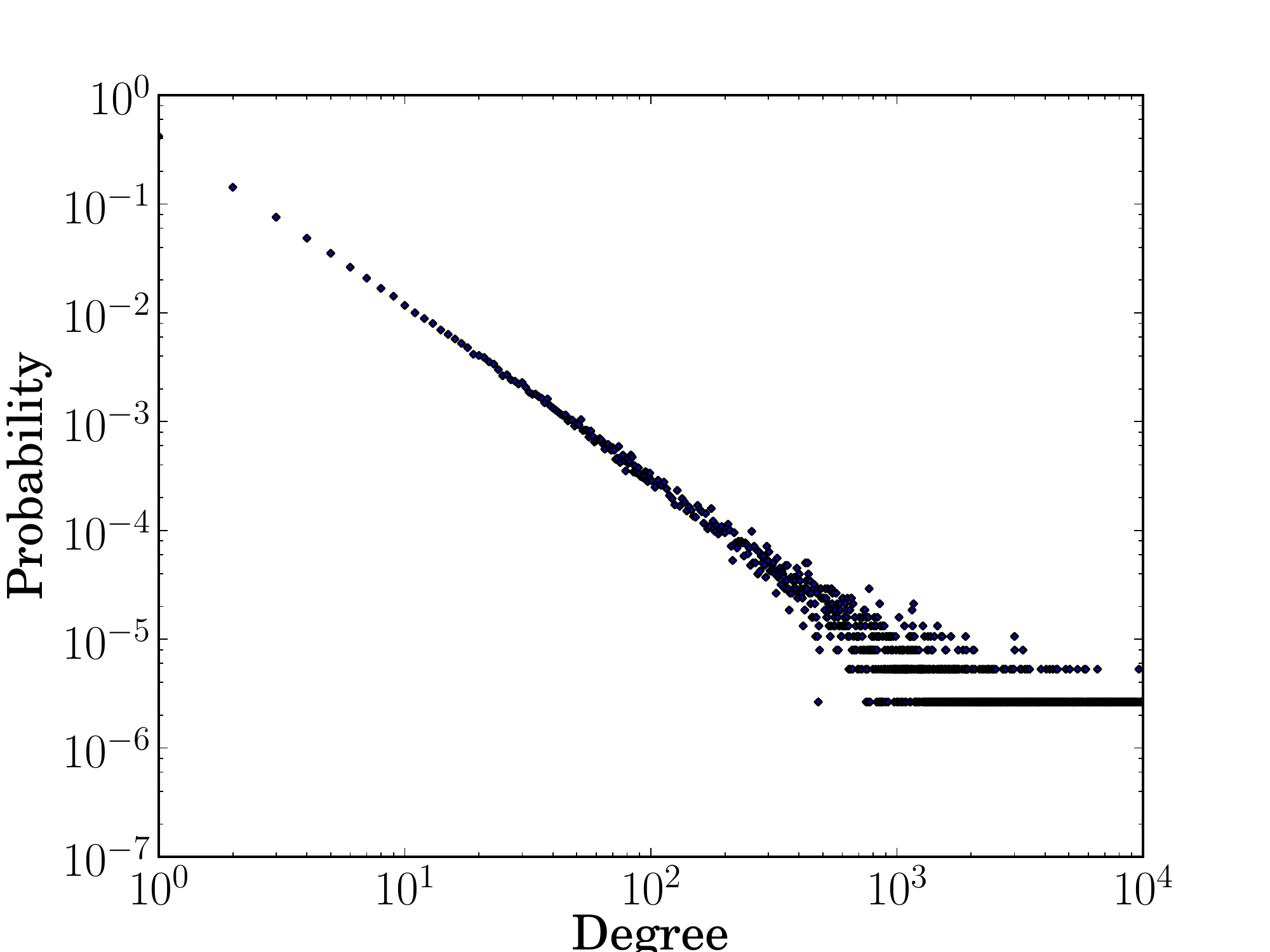}\label{gen-zheleva-soc-attri}}
\subfloat[\scriptsize{Social degree of attribute nodes}]{\includegraphics[width=0.25\textwidth, height=1.5in]{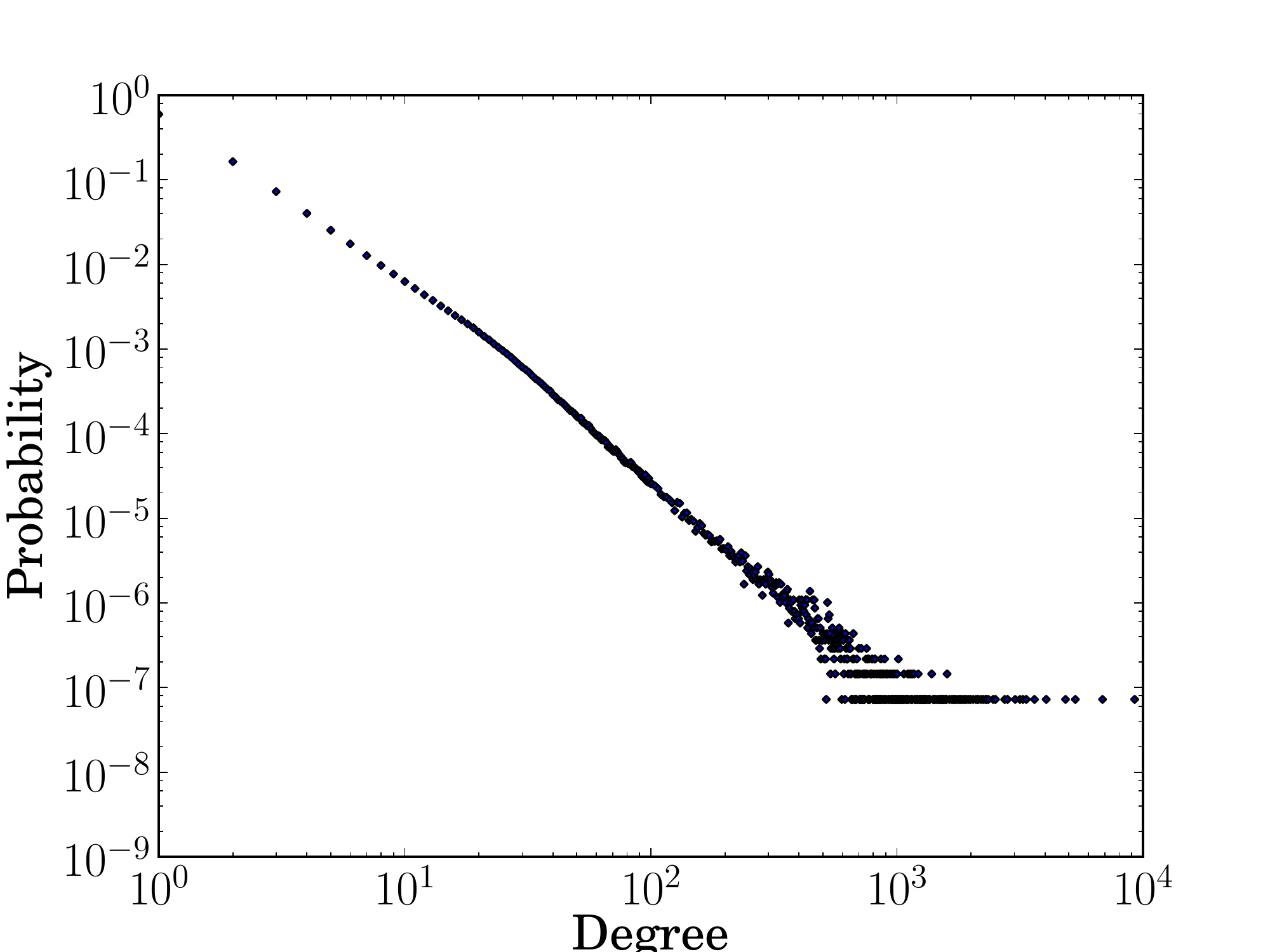}\label{gen-zheleva-attri-soc}}
\tightcaption{Degree distributions of synthetically generated \san using 
 our model in (a)-(d) vs.\ \Zhel shown in (e)-(h).} 
\label{fig:gen-degree-dis}
\vspace{-2mm}
\end{figure*}

\section{Evaluation}
\label{sec:evaluation}
In this section, we validate our \san generative model.
Because the \san area is still very nascent there are few standard
models of comparison. We pick the closest generative model by Zheleva
et~al~\cite{Zheleva09-evo}. Note that their model is actually orthogonal to ours since it's modeling dynamic node attributes while ours is modeling static node attributes. Furthermore, their original model generates undirected social networks. In order to compare with our model and directed Google+ \sanplural, we extend their model to generate directed social networks\footnote{Extending their model is straightfoward. For instance, when the original model issues an undirected link, we change it to be a directed outgoing link.}. We refer to the extended model as the \Zhel model throughout
this section.  We start with network metrics, including single-node
degree distribution, joint degree distribution and clustering coefficient.
Then, following  the spirit of~\cite{Sala10}, we also evaluate our model using
real application contexts. 

For comparison, we use the  Google+  snapshot crawled on July 15, 2011, which
has roughly 10 million nodes and we believe it is representative of Google+
\san.  Using this Google+ snapshot, we run a guided greedy search to
estimate appropriate parameters for our model and \Zhel to generate synthetic
\san that best match the Google+.

\begin{figure}[ht]
\vspace{-0.5cm}
\centering
\subfloat[\scriptsize{JDD of attribute nodes}]{\includegraphics[width=0.25\textwidth, height=1.5in]{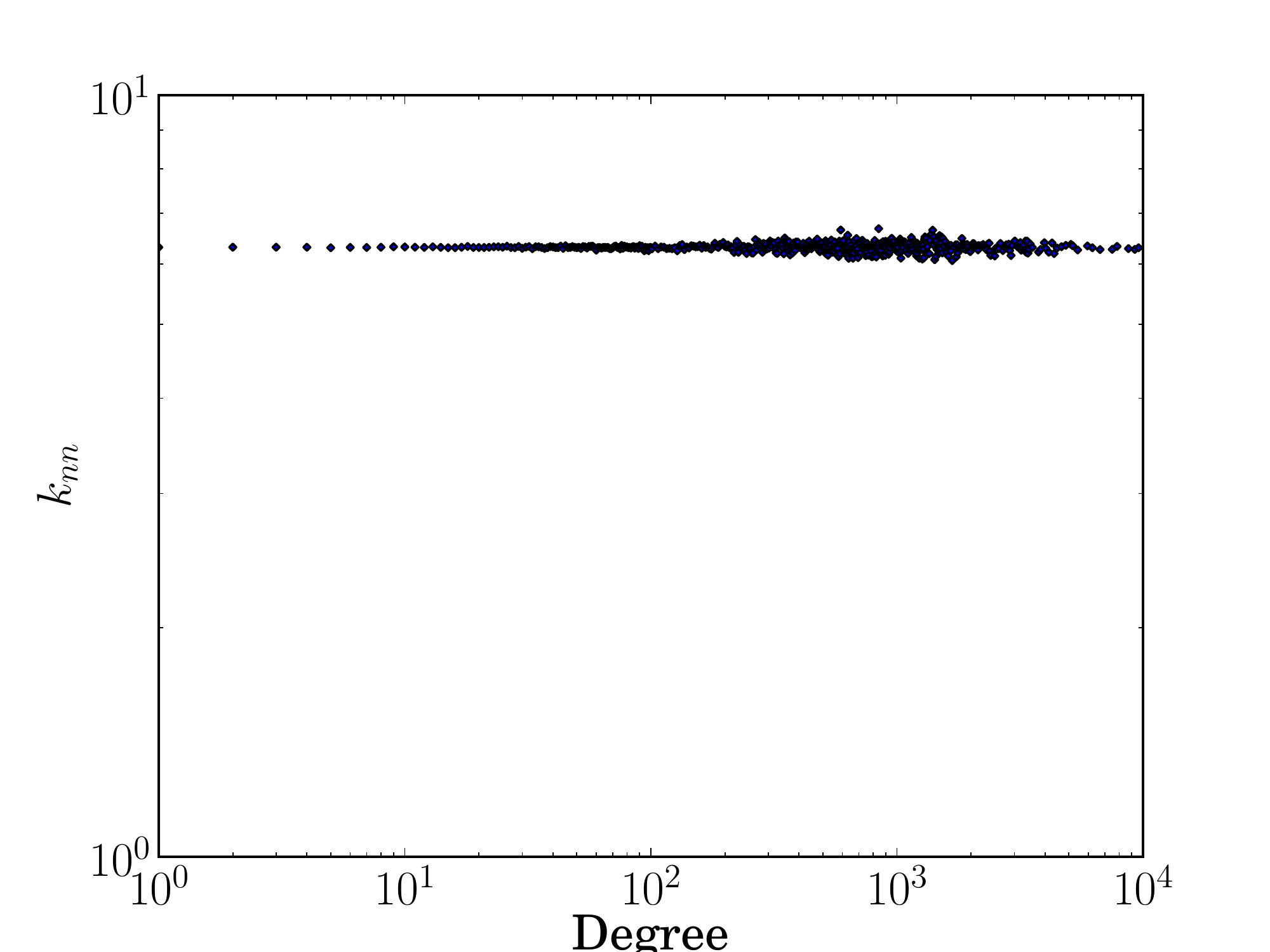}\label{gen-assort-attri-attri}}
\subfloat[\scriptsize{Clustering coefficient}]{\includegraphics[width=0.25\textwidth, height=1.5in]{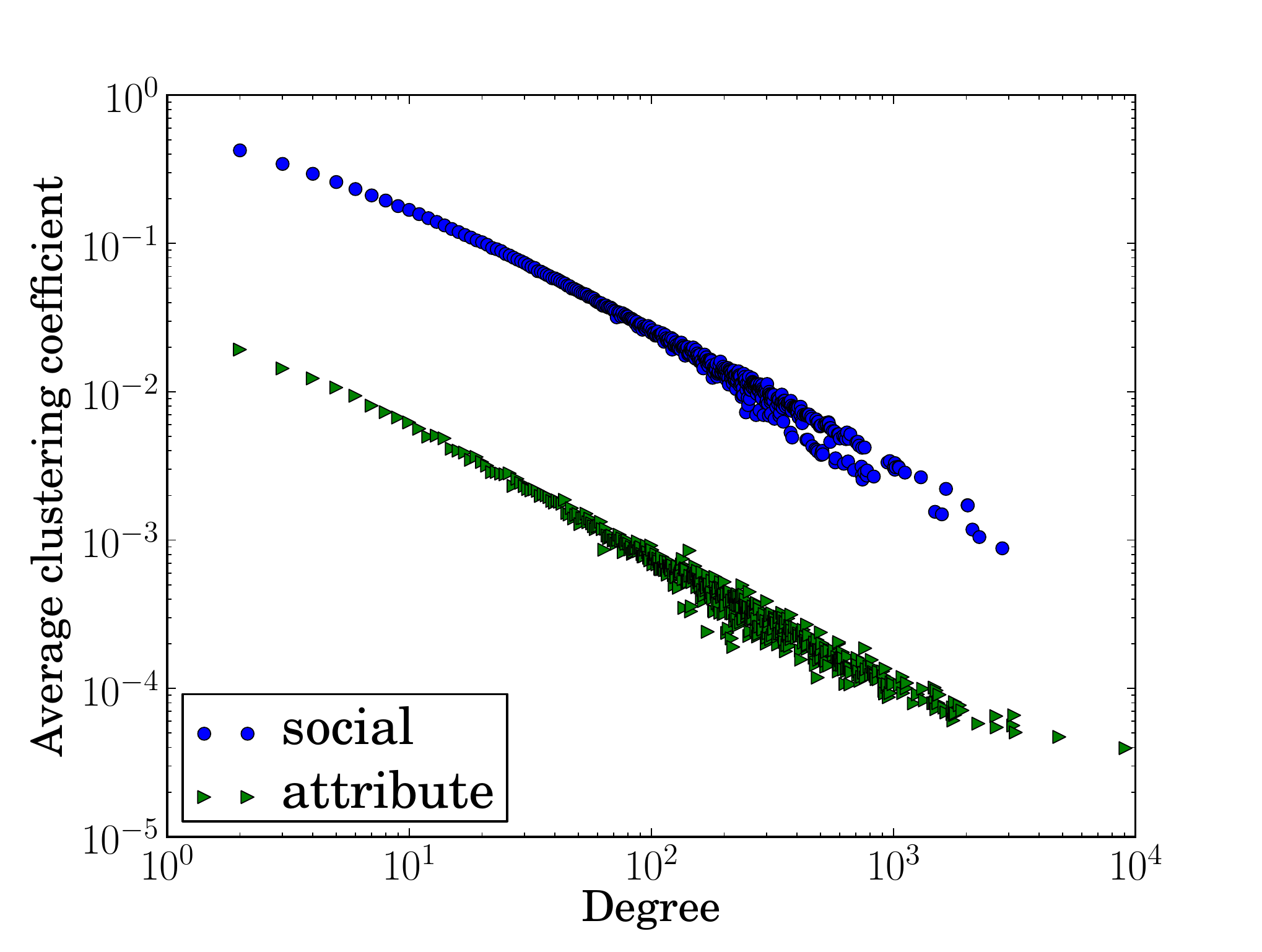}\label{gen-clustering-coefficient-distribution}}

\eat{
\subfloat[attri-attri]{\includegraphics[width=0.25\textwidth, height=1.5in]{gen-baseline-assort-attri-attri}\label{gen-baseline-assort-attri-attri}}
\subfloat[out-in]{\includegraphics[width=0.25\textwidth, height=1.5in]{gen-baseline-assort-out-in}\label{gen-baseline-assort-out-in}}
\subfloat[clustering coefficient]{\includegraphics[width=0.25\textwidth, height=1.5in]{gen-baseline-clustering-coefficient-distribution-degree_loglog}\label{gen-baseline-clustering-coefficient-distribution}}
}
\vspace{-0.3cm}
\subfloat[\scriptsize{JDD of attribute nodes}]{\includegraphics[width=0.25\textwidth, height=1.5in]{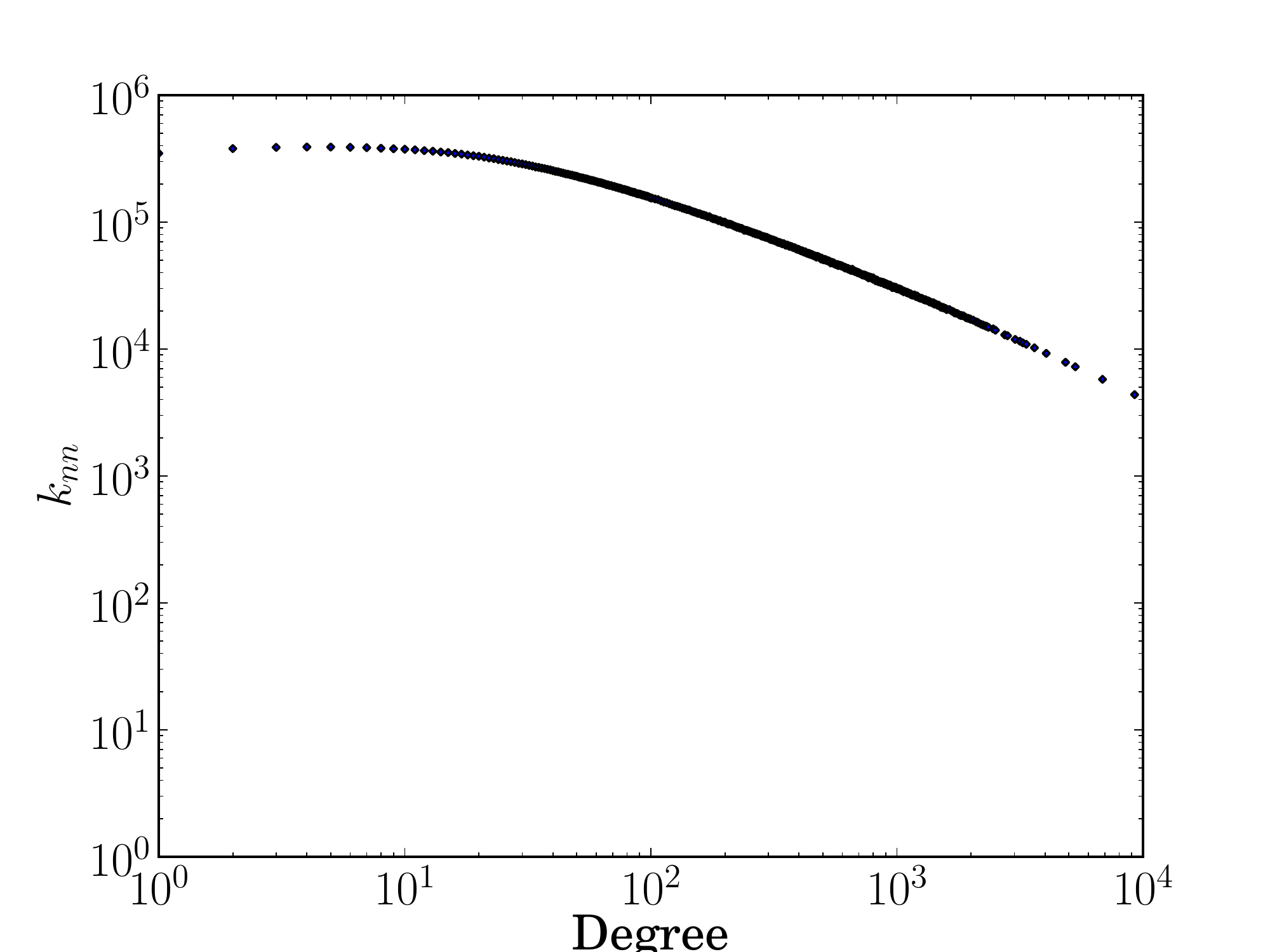}\label{gen-zheleva-assort-attri-attri}}
\subfloat[\scriptsize{Clustering coefficient}]{\includegraphics[width=0.25\textwidth, height=1.5in]{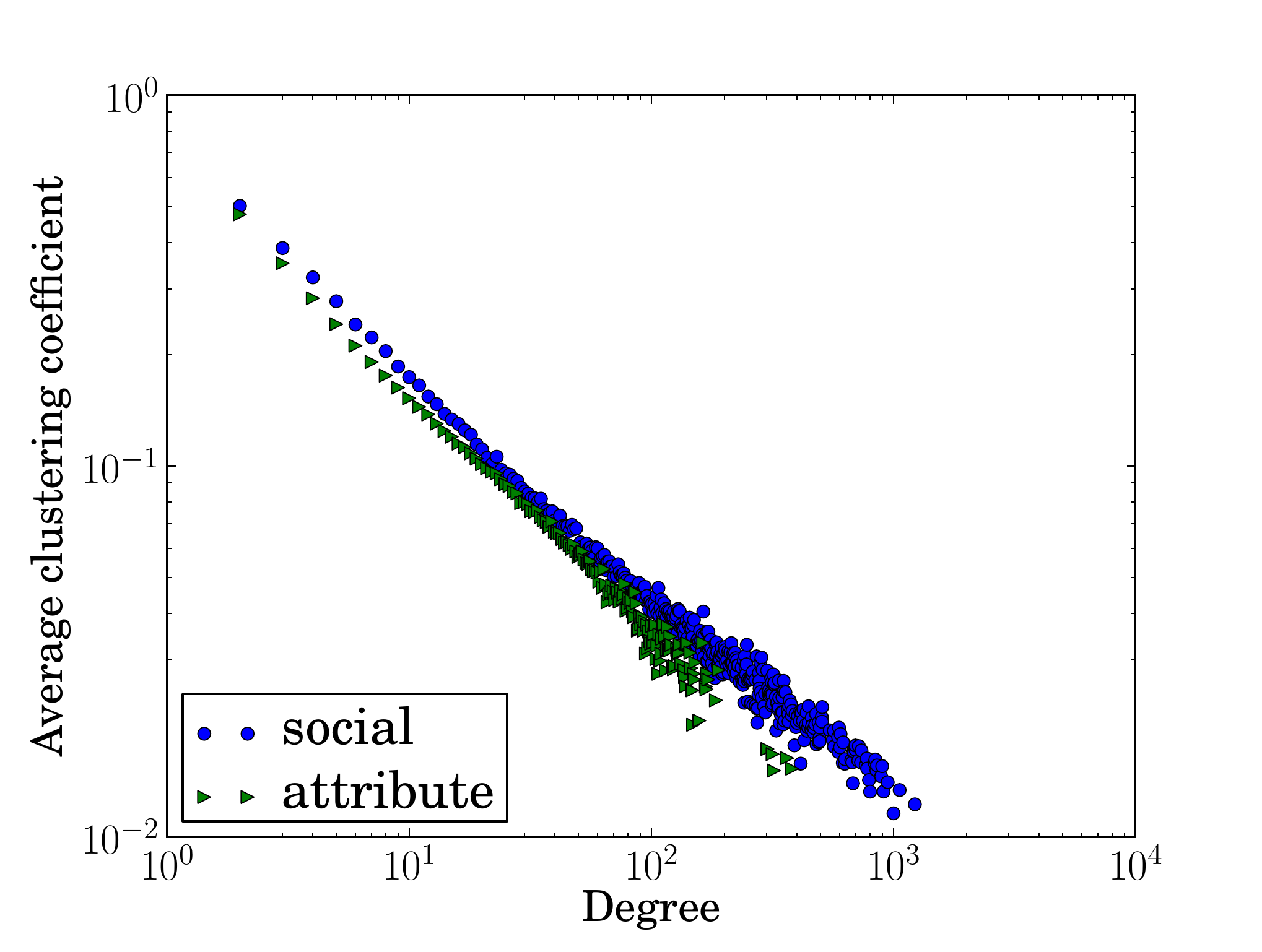}\label{gen-zheleva-clustering-coefficient-distribution}}
\tightcaption{Joint degree and clustering coefficient distributions of our model (a)--(b) vs.\ \Zhel in (c)--(d).}
\label{fig:gen-jd}
\vspace{-2mm}
\end{figure}

\subsection{Network Metrics}
In this section, we qualitatively compare our model to the \Zhel model,
and demonstrate that our model can generate synthetic \san that better 
 reproduces various network metrics closer to Google+ \san.


\mypara{Degree distributions} 
We first examine the degree distributions of the synthetic \san generated
by our model and the \Zhel model in Figure~\ref{fig:gen-degree-dis}.  The most
visually evident result looking at  Figure~\ref{gen-soc-out} and
Figure~\ref{gen-soc-in} is that our model can generate  synthetic
networks with social indegree and outdegree following lognormal distributions
similar to the Google+ \san that we saw in
Figure~\ref{fig:degree-distribution}.  In contrast,
Figure~\ref{gen-zheleva-soc-in} and Figure~\ref{gen-zheleva-soc-out} confirm
that the \Zhel model generates  indegree and outdegree following power-law
distributions.  Similarly, comparing
Figure~\ref{gen-soc-attri} and~\ref{gen-zheleva-soc-attri} to Figure~\ref{fig:soc-attri}, the attribute degree
of social nodes in our model follows the lognormal distribution that matches
that of the Google+ \san, whereas the \Zhel  model generates  attribute degrees
that follow a power-law distribution.  Finally, Figure~\ref{gen-attri-soc}
and~\ref{gen-zheleva-attri-soc} confirm that both our model and \Zhel  generate
social degrees of attribute nodes that follow power-law distribution, which is
 again consistent with Google+ \san from Figure~\ref{fig:attri-soc}.

\mypara{Joint degree distributions} The ability to mirror more fine-grained  properties
beyond the degree distributions has been shown to be a key metric for
evaluating generative models~\cite{Mahadevan06}.  Thus,  we look at the joint
degree distribution approximated by degree correlation function $k_{nn}$ in
Figure~\ref{gen-assort-attri-attri} and~\ref{gen-zheleva-assort-attri-attri}
for our model and \Zhel.  Compared to
Figure~\ref{fig:attri-joint-degree-distribution}, we see that the JDD of
attribute nodes in our model generated \san matches  Google+ \san much better
than \Zhel.  We observe similar pattern for JDD of social nodes. 

\mypara{Clustering coefficient} Fig.~\ref{gen-clustering-coefficient-distribution} and 
Fig.~\ref{gen-zheleva-clustering-coefficient-distribution}
shows the clustering coefficient distributions of synthetic \sanplural 
generated by our model and \Zhel, respectively.
When comparing them to Fig.~\ref{fig:clustering-coefficient-dis-deg} 
we see that our model generates synthetic \san with
both social and attribute clustering coefficient distributions 
matching well to those of Google+ \san, which is not the case
for \Zhel.

\mypara{Significance of building blocks} 
Recall that our model has two key building blocks that extend preferential 
attachment via LAPA and also extending triangle closing via focal closure. 
 A natural question is what each of these components contribute toward the 
overall generative model.

First, we investigate how LAPA impacts the structure of the generated \san
in our model. To this end, we consider an intermediate model with the classical PA
(but with the RR-\san enabled)  and compute the previous metrics for \sanplural generated by this intermediate model.  We find that all metrics except the distribution of social \emph{indegree} are qualitatively the same.  
Figure~\ref{gen-baseline-soc-in} shows that the distribution of social indegree
of the synthetic \san generated by our intermediate model is very close to
a power-law distribution, different from the lognormal distribution generated
by our full model shown in Figure~\ref{gen-soc-in} and derived from the real
Google+ \san shown in Figure~\ref{fig:degree-distribution}. This 
 suggests that the LAPA component is necessary for modeling a key aspect 
 of the Google+ \san.

\begin{figure}[t]
\vspace{-0.5cm}
\centering
\subfloat[\scriptsize{Social indegree w/o LAPA}]{\includegraphics[width=0.25\textwidth, height=1.5in]{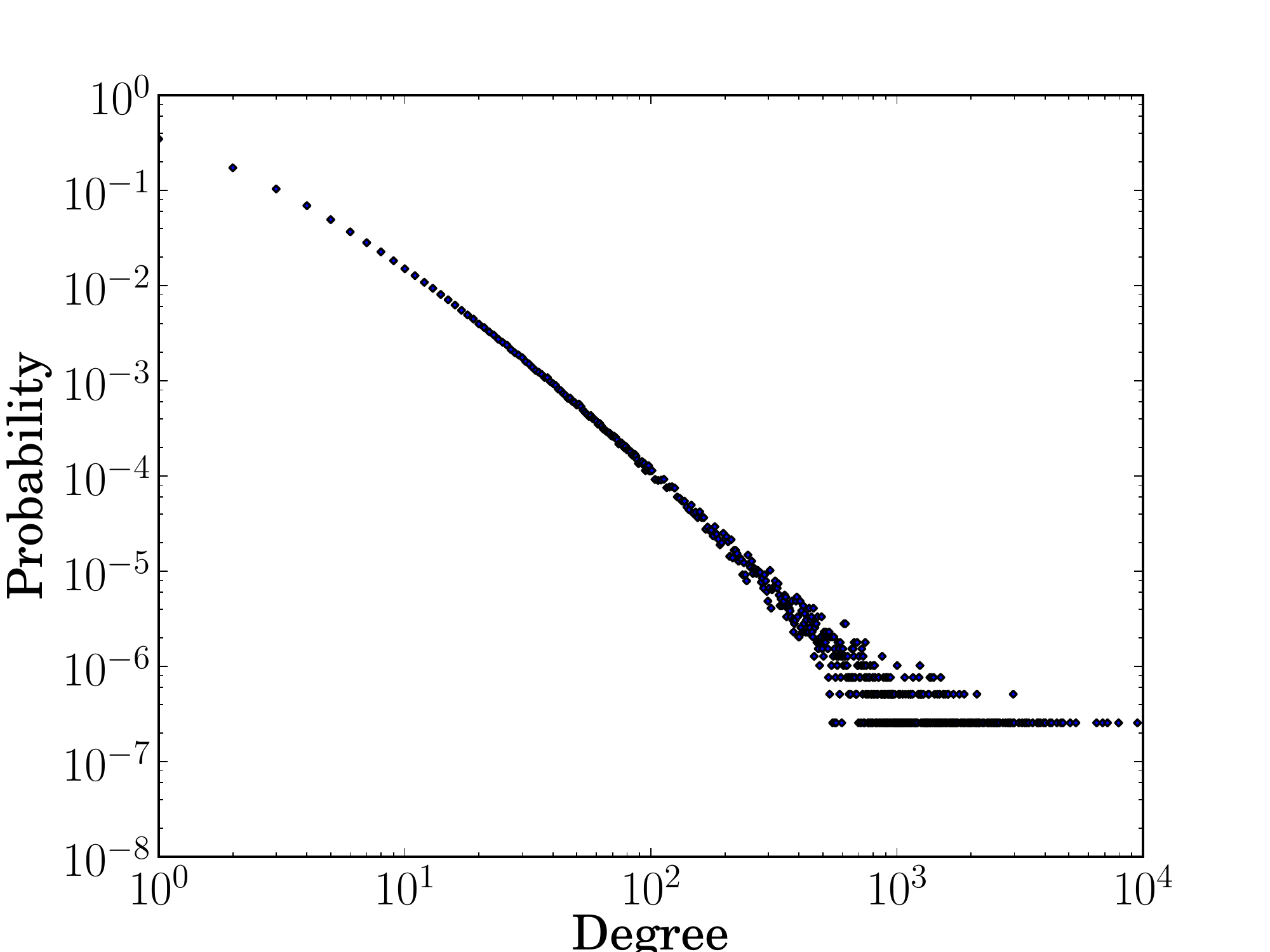}\label{gen-baseline-soc-in}}
\subfloat[\scriptsize{Clustering coefficient w/o focal closure}]{\includegraphics[width=0.25\textwidth, height=1.5in]{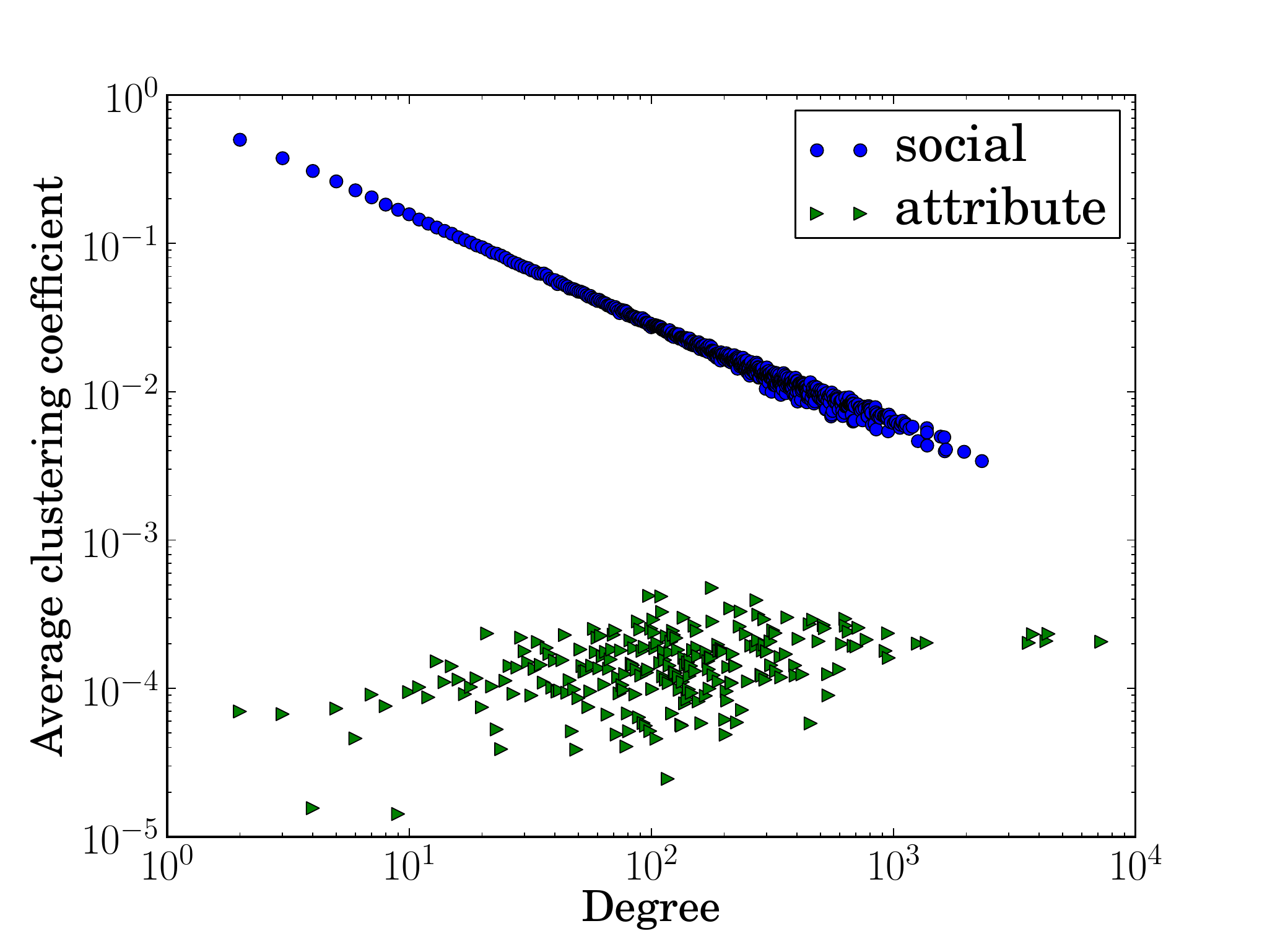}\label{gen-baseline-clustering-coefficient-distribution}}
\tightcaption{The effect of LAPA and focal closure.}
\label{fig:no-pa-tc}
\vspace{-2mm}
\end{figure}

\eat{
\begin{figure}[t]
\centering
\subfloat[Our full model]{\includegraphics[width=0.25\textwidth, height=1.5in]{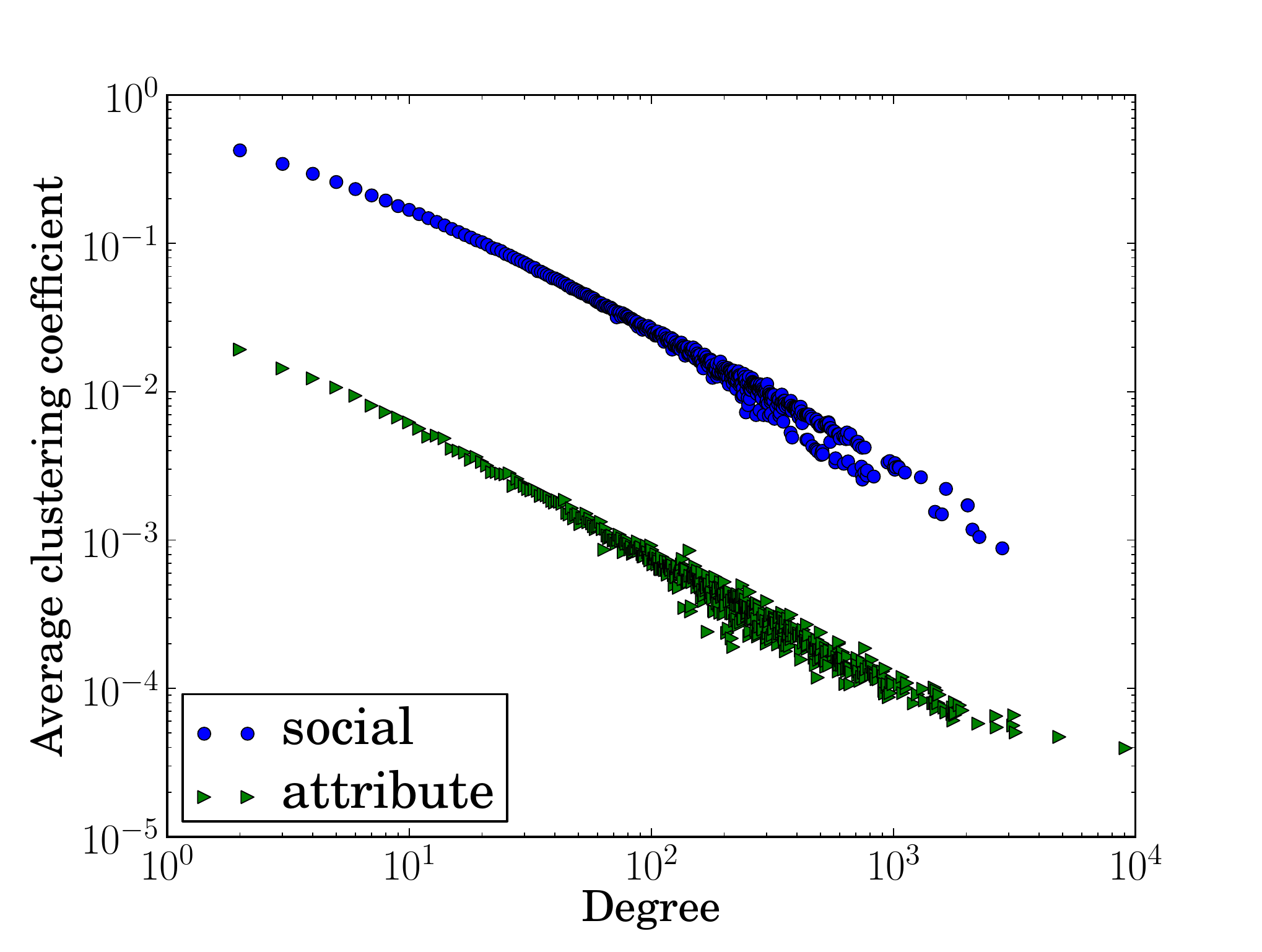}\label{gen-clustering-coefficient-distribution}}
\subfloat[Without focal closure]{\includegraphics[width=0.25\textwidth, height=1.5in]{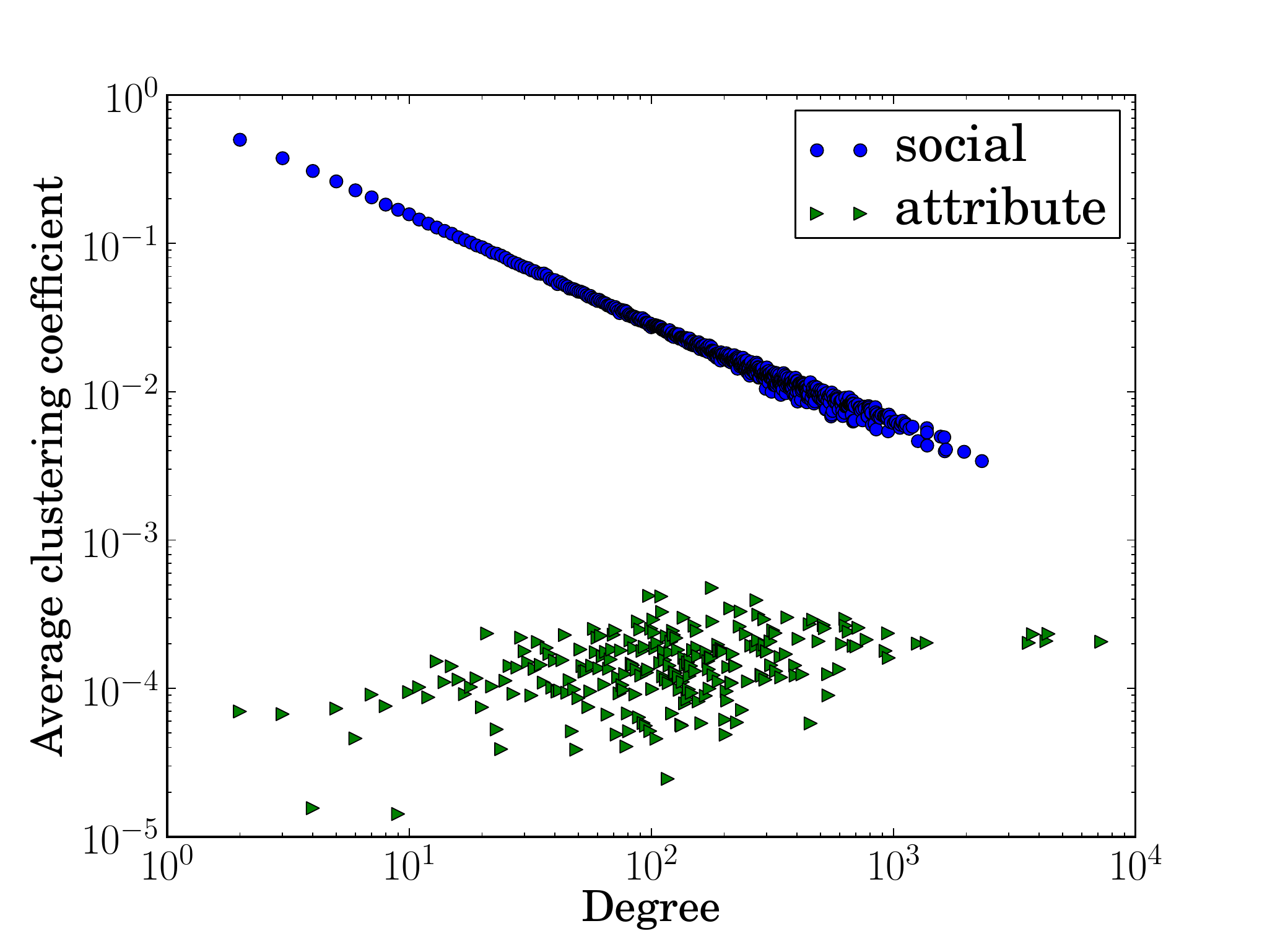}\label{gen-baseline-clustering-coefficient-distribution}}
\captionsetup{font=small}
\caption{The distributions of clustering coefficients of synthetic 
networks generated by our models.}
\label{fig:gen-cc}
\end{figure}
}

Second, we investigate the impact of RR-\san. 
 The key metric impacted by  the focal closure component of  RR-\san 
 is the attribute clustering coefficient.
 Figure~\ref{gen-baseline-clustering-coefficient-distribution} shows 
 the social and attribute 
clustering coefficients of synthetic \sanplural generated by our model without RR-\san (with classical RR enabled). 
Looking at 
Figure~\ref{gen-clustering-coefficient-distribution} and 
Figure~\ref{gen-baseline-clustering-coefficient-distribution} together,
we see that RR-\san has a significant impact on the attribute clustering
coefficient.


These results confirm  both attribute-augmented building blocks, LAPA and RR-\san, play important but complementary roles in our model in generating synthetic \san that closely mirrors the real Google+ \san.


\subsection{Application Fidelity}

 Next, we use two real-world application contexts to evaluate the fidelity of
our generative model and the  \Zhel model with respect to a real Google+
snapshot. In each case, we use the metric of interest relevant to each
application. Note that all these applications only rely on the social structure.

\begin{figure}[t]
\vspace{-0.5cm}
\centering
\subfloat[\scriptsize{Sybil defense}]{\includegraphics[width=0.22\textwidth, height=1.5in]{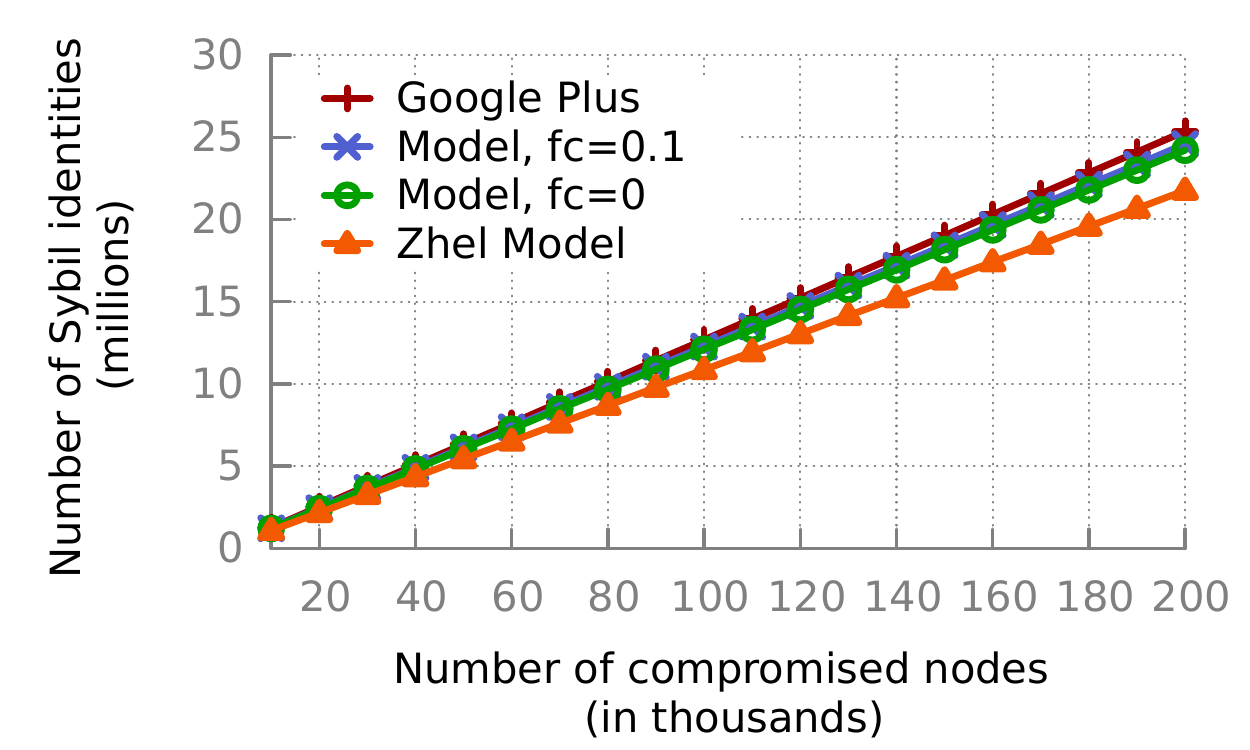}\label{fig:sybil}}
\subfloat[\scriptsize{Anonymous communication}]{\includegraphics[width=0.25\textwidth, height=1.5in]{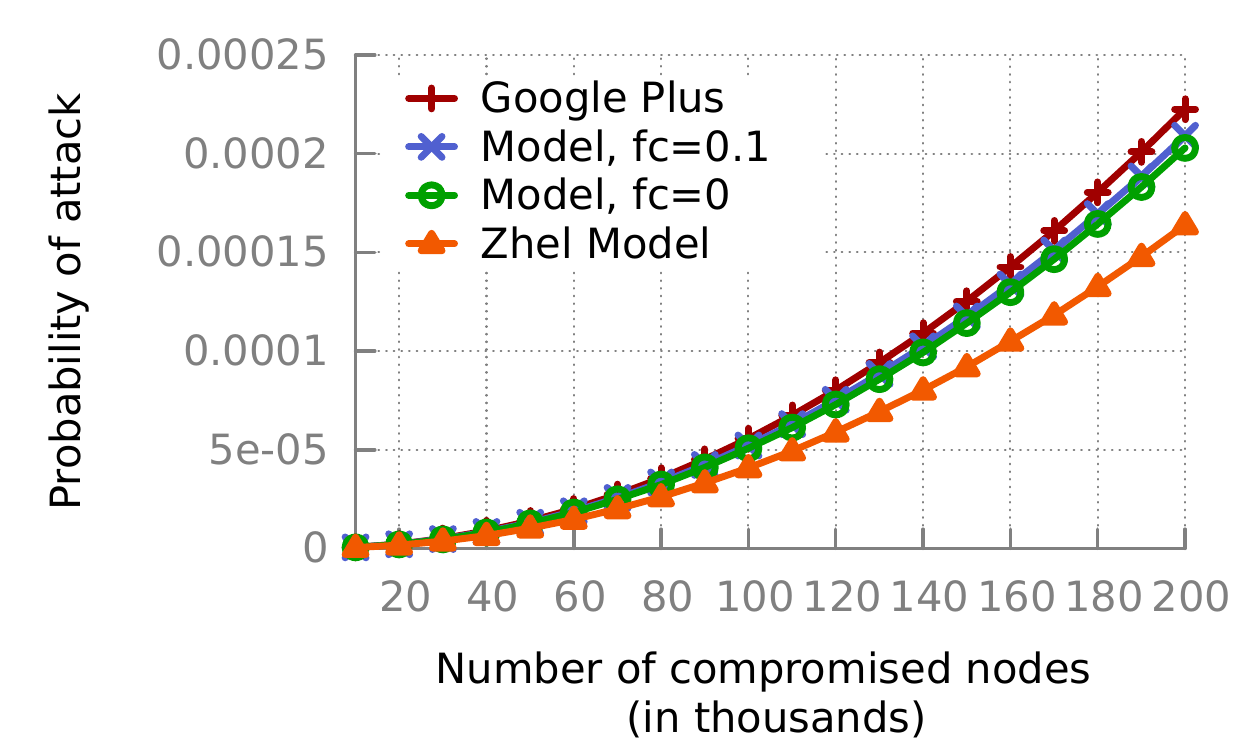}\label{fig:anonymity}}
\tightcaption{Application fidelity of our model. (a) Sybil defense: SybilLimit false negatives as a function of number of compromised nodes. (b) Social network based anonymity: Probability of end-to-end timing analysis as a function of 
number of compromised nodes.}
\label{fig:no-pa-tc}
\vspace{-2mm}
\end{figure}

\eat{
\begin{figure}[t]
\centering
{\includegraphics[width=0.45\textwidth, height=1.8in]{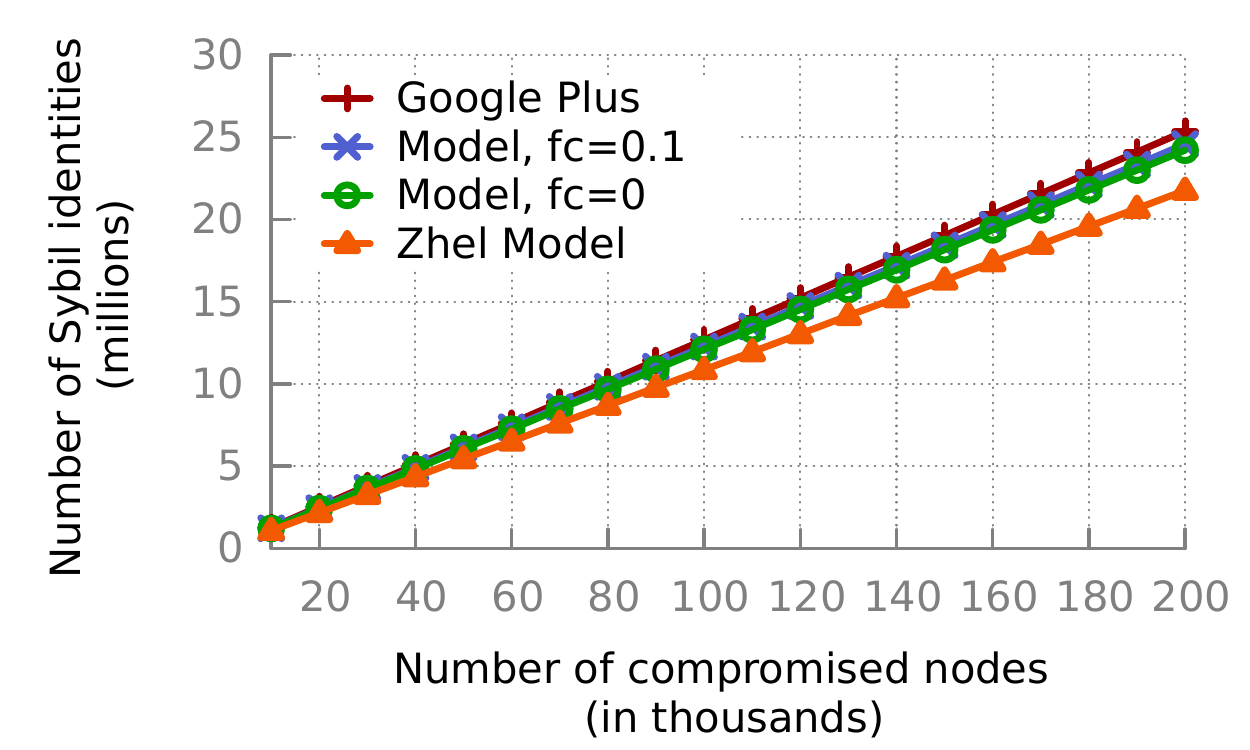}}
\tightcaption{Sybil defense: SybilLimit false negatives as a function of number of compromised nodes using Google+ data 
and our proposed model. We can see the accuracy of our model, and also the improvement over prior \Zhel model. 
}
\label{fig:sybil}
\end{figure}
}

\myparatight{Sybil defense}
In a Sybil attack~\cite{sybil}, a single entity emulates the behavior of 
a large number of identities to compromise the security and privacy properties 
of a system. Sybil attacks are of particular concern in decentralized systems, 
which lack mechanisms to vet identities and perform admission control. 
 Several recent works have proposed the use of social trust relationships to mitigate 
Sybil attacks~\cite{sybil,sybillimit}. Next, 
we show the fidelity of our model using a representative social network based 
Sybil defense mechanism called SybilLimit~\cite{sybillimit}. 

 In order to prevent an adversary from obtaining a large number of attack 
edges (edges between compromised and honest users), 
 SybilLimit  bounds the effective node degree in the social 
network topology. Following their guidelines, we also imposed a node degree bound 
of 100 in evaluating their proposal on the different 
 \sanplural. Figure~\ref{fig:sybil} depicts the number of Sybil identities that 
an adversary can insert, as a function of number of compromised nodes in the 
network. We compromised the nodes uniformly at random, and set the 
SybilLimit parameter $w=10$. The parameter $fc$ governs 
the attribute link weight in our RR-\san component; $fc=0$ means no focal closure.

We can see that (a) SybilLimit results using the synthetic topology generated
by our model are a close match to the real Google+ data, and (b) our model
outperforms the baseline approach (\Zhel model). For example, when the number
of compromised nodes is 200,000  the average number of Sybil identities in the
Google+ topology is about 25.3 million, while our model predicts 24.5 million
(error of 3.1\% using $fc=0.1$).  In contrast, the baseline approach 
has almost 4$\times$ worse error with a prediction error of 12.5\%.  This shows 
the importance of using attribute information to influence the structure of 
the social structure (the \Zhel model only uses the social structure to influence the 
attribute structure.)

\eat{
\mypara{RE: Reliable Email} Next, we evaluate the fidelity of the models using
the Reliable Email (RE) application~\cite{reliableemail}. The high-level idea
in RE is to reduce the false positives introduced by spam filters by creating a
spam ``whitelist'' using the social trust relationships.  Specifically, it
marks any incoming message from a friend or a friend-of-friend as legitimate
mail. A natural consequence of such whitelisting of course is that if some of
the friend or friend-of-friend machines are compromised, then users are likely
to receive more spam.  We simulate such a scenario and  vary the number of
compromised nodes  from 2--10\%  (selected at random) and in each case measure
the number of legitimate nodes who are likely to receive spam as a consequence. 
 We find that our model predicts {\bf XXX \%}  of users will receive spam 
which is closer to the real Google+ estimate of {\bf YYY\%}, whereas the 
 \Zhel model predicts {\bf XXX \%} of users will receive spam. Again, we see 
 that our model better predicts the  application-level metric.
}

\eat{
\begin{figure}[t]
\centering
{\includegraphics[width=0.45\textwidth, height=1.8in]{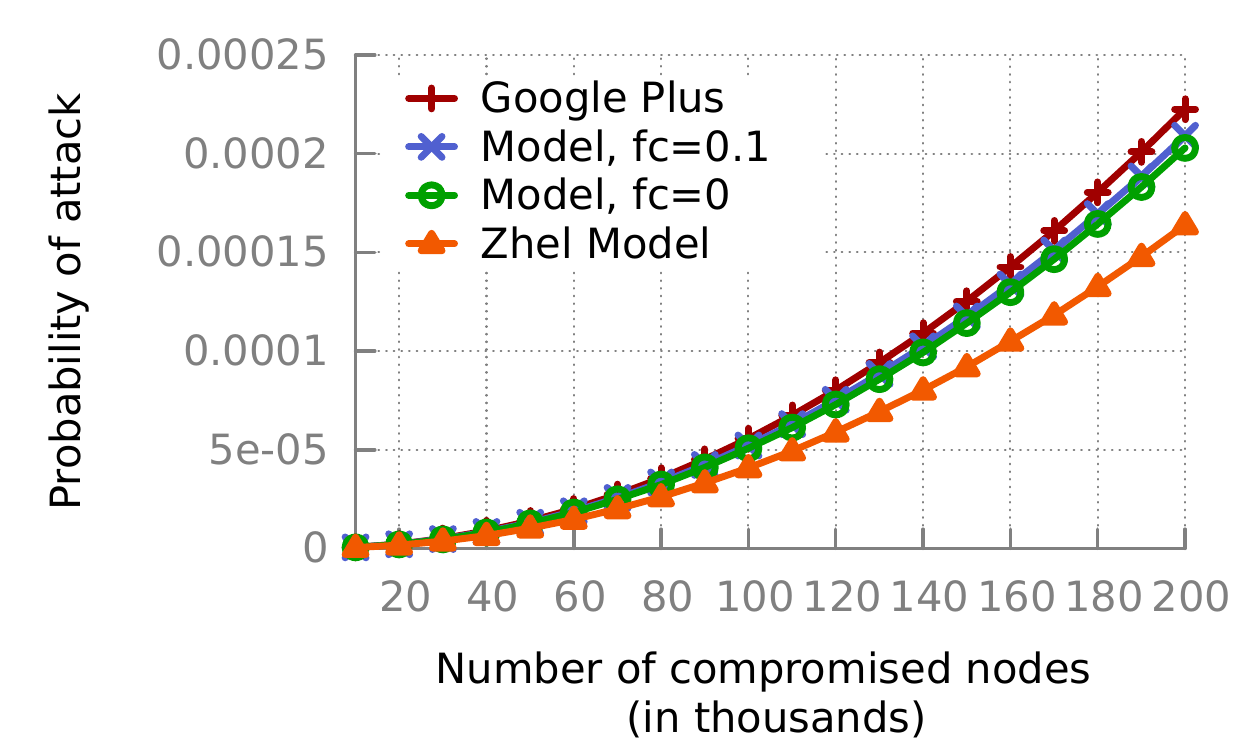}}
\tightcaption{Social network based anonymity: Probability of end-to-end timing analysis as a function of 
number of compromised nodes using Google+ data and our  model.  
}
\label{fig:anonymity}
\end{figure}
}

\myparatight{Anonymous communication}
Anonymous communication aims to hide user identity (IP address) 
from the recipient (destination) or from third parties on the Internet 
such as autonomous systems. The Tor network~\cite{tor} is a deployed 
system for anonymous communication that serves hundreds of thousands of 
users a day. It is widely used by political dissidents, journalists, 
whistle-blowers,  and even law enforcement/military. 
Recent work~\cite{johnson:ccs11,drac} has proposed leveraging social 
links in building anonymous paths for improving resistance to 
attackers. For example, the Drac~\cite{drac} system selects proxies (onion routers) 
by performing a random walk on the social network. For low-latency communications, 
if the first and the last hops of the 
forwarding path (onion routing circuit) are compromised, then the adversary can perform 
end-to-end timing analysis and break user anonymity. 
Figure~\ref{fig:anonymity} depicts the probability of end-to-end timing analysis 
when random walks on social networks are performed for anonymous communication, 
using the Google+ social network and our synthetic network.
Similar to our SybilLimit experiments, we compromise nodes uniformly at 
random in the network, and impose an upper bound of 100 on the node degree.  
Again, we can see the 
accuracy of our model, as well as the improvement over prior work. 

\subsection{Summary}
Via evaluating our model with respect to network metrics and real-world applications, we find that:

\begin{packeditemize}

\item Our model can reproduce \sanplural that well match Google+ \san with respect to various network metrics (e.g., degree distributions, joint degree distributions and clustering coefficients.), but the \Zhel model cannot match several metrics (e.g., social degree distributions, joint degree distributions and clustering coefficient.).

\item Our model also performs better than the \Zhel model for real-world applications such as Sybil defense and anonymous communication.

\item The two attribute-augmented building blocks, i.e., LAPA and RR-\san, play important but complementary roles in our model. 

\end{packeditemize}

\section{Discussion}
\label{sec:discussion}

\myparatight{Using attributes to strengthen defenses}  
Our evaluation largely focuses
on how our model better matches the real-world \san. We hypothesize that several
attack defenses (e.g., Sybil proofing) can also be enhanced  by taking into
account the attribute structure. For example, we could check if the
attribute structure of the nodes matches normal nodes, or even if an attacker
manages to obtain a ``compromised'' edge to one node we can limit the influence
of this compromised edge by checking the attribute structure.

\myparatight{LAPA Computation} The LAPA model as described  
requires a costly linear time (in number of nodes) step 
  when a new node arrives. This is because we have to consider the number of common
 attributes between the new node and each current node,  unlike PA 
 which only needs the global degree distribution.
 Fortunately, we can approximate LAPA using a practical heuristic. The high-level idea is 
to pick one of the new node's attributes at random and use PA within the nodes having this 
 attribute. This approximates LAPA as nodes sharing 
 more attributes are more likely to get selected. 

\myparatight{Dynamic attributes}  Our model currently focuses on static
attributes that nodes pick when they join the \san. In our future work,
we plan to incorporate dynamic attributes, and investigate whether the static
attribute structure also influences the selection of dynamic attributes. Note that static attributes influence the social structure in our model while the dynamic attributes are influenced by the social structure in the model from  Zheleva et~al~\cite{Zheleva09-evo}.

\myparatight{Parameter inference}  We currently use a guided greedy search 
to empirically estimate model parameters. While this works quite 
well, we plan to develop a more rigorous parameter inference algorithm
based on maximum-likelihood principle~\cite{scalable_modeling,likelihood}. 

\myparatight{Parsimoniousness of our model}  In \Section~\ref{sec:evaluation}, we have shown that each component of our model is necessary. However, it's an interesting future work to design a more parsimonious
model. 

\myparatight{Implications for social network designs} Our results that users sharing common employer attributes are more likely to be linked than users sharing other attributes can help design a better friend recommendation system, which is a very fundamental component of online social networks.

\myparatight{Relationship to heterogeneous networks} Our \san can be viewed as a heterogeneous network since it consists of multiple types of nodes and links. Heterogeneous networks are shown recently to work better than traditional homogeneous networks for various data mining tasks such as link prediction~\cite{Gong11, Yang11, Sun11, Sun12}, attribute inference~\cite{Gong11, Yang11} and community detection~\cite{Sun09, Sun12-1, Zhou09}. 
It is an interesting future work to generalize our new attribute-related metrics and generative model to other heterogeneous networks.

\section{Related Work}
\label{sec:related}
 Given the growing role of social networks in users' lives and the potential
for using such insights for building better systems and applications, there is
a rich literature on measuring and modeling social networks. Next, we discuss
our work in the context of this related work. At a high-level, our specific new
contributions are: (1) we characterize the evolution of a new large-scale
network (Google+), and (2) we provide measurement-driven insights and models on
the impact of attributes on social network evolution. 

\myparatight{Measuring social networks}  Many 
 prior efforts characterize social networks using the network metrics 
 we also describe in \Section\ref{sec:structure}~\cite{Kumar06, Kwak10,
Leskovec08, Mislove07, Kossinets06}.   Most of these focus 
 on static snapshots; a few notable work also focus on evolutionary aspects 
 similar to our work~\cite{Ahn07, Wilson09, Backstrom12}. With multiple Google+ snapshots crawled around its public release, Schioberg et al.~\cite{Schioberg12} studied a few network metrics, geographic distribution of the users and links, and correlation of users' public information of Google+.
%

Concurrently, Gonzalez et al.~\cite{Gonzalez12} characterize several key 
features of Google+ during its first 10 months, and compare them to those 
of Facebook and Twitter.
Using a static Google+ snapshot crawled after its public release, 
Magno et al.~\cite{Magno12} identify the key differences between Google+ 
and Facebook and Twitter, study the adoption patterns of Google+ in different 
countries, and characterize the variation of privacy concerns across 
different cultures.  
Zhao et al.~\cite{MultiDyna} study the early evolution 
of the Renren social network, and analyze its network dynamics at 
different granularities to determine their influence on individual users. 
While we follow the spirit of these works, our work is unique 
  in terms of the specific dataset (i.e., three phases of Google+), the scale of the network, 
 and the fact that we had a singular opportunity to study the evolution across 
different phases. 
 
There has been recent realization of the importance of user
attributes in characterizing social networks~\cite{Mislove07, Zheleva09-evo}.
These focus on the influence of social structure on dynamic node
attributes (e.g., interest groups).  Our work focuses on the orthogonal
dimension of  analyzing and modeling the influence of static node attributes on
social structure formation using Google+.

\myparatight{Modeling social networks} There are two broad classes of models for
generating social networks: \emph{static} and \emph{dynamic}. Static
models try to reproduce a single static network snapshot~\cite{Erdos59, Watts98, Mahadevan06, STANTON12}.  Dynamic models
can provide insights on how nodes arrive and create links; these include
models such as  preferential attachment~\cite{Barabasi99},
copying~\cite{Kleinberg99}, nearest neighbor~\cite{VAZQUEZ03}, forest
fire~\cite{Leskovec05}.  Sala et al.~\cite{Sala10} evaluated such models using
both network metrics and application benchmarks and showed that the nearest
neighbor model outperforms others.  The dynamic/generative model by Leskovec et~al.
mimics the nearest neighbor model in a dynamic setting~\cite{Leskovec08}, and
thus we use it as our starting point in \Section\ref{sec:model}, However, these
models are known to generate networks with power-law degree distributions.  Many
social networks including Google+,
however, exhibit lognormal degree distributions~\cite{Gomez08, Leskovec08-WWW,
Liben-Nowell05}. Our dynamic model extends these prior work to provably
generate a lognormal distribution for social outdegree. Our model also provides
a more general framework by capturing both social and attribute structure.

\myparatight{Modeling social-attribute networks}  There has been relatively
little work on generating \sanplural, though a few recent work jointly generating both social structure and node attributes can be viewed as \san models; the most relevant  work is from Zheleva
et al.~\cite{Zheleva09-evo} and Kim and Leskovec~\cite{Kim12}.  Zheleva et
al.~\cite{Zheleva09-evo} focus on dynamic attributes; their model generates
undirected networks with power-law distribution for social degree and
non-lognormal distribution for attribute degree (see
Figure~\ref{fig:gen-degree-dis}).  Kim and Leskovec model  the social and
attribute structure simultaneously~\cite{Kim12}. Here, both the social degree
of attribute nodes and attribute degrees of social nodes follow binomial
distribution, which differs from empirically observed \sanplural. Our model can
generate \sanplural that we confirm through both analysis and simulations to
be consistent with real \sanplural.

\section{Conclusion}
\label{sec:conclude}

Using a unique dataset collected by crawling
Google+ since its launch in June 2011, we provide a first-principled
understanding of the attribute structure and its impact on the social structure and
their evolutions with the \san model.  We observe several interesting phenomena in the structure and evolution of Google+. For example, the social degree
distributions are lognormal, the assortativity is neural while many other social networks have positive assortativities, and the distinct phases in the evolution manifest
themselves in the network structure. We also provide new metrics for
characterizing the attribute structure and demonstrate that attributes
 can significantly impact the social structure.  Building on these empirical
insights, we provide a new generative model for \sanplural and validate
that it is close to the real Google+ \san  using both network metrics
and real application contexts.  We believe that our work is one of the first
steps in this regard and that there are several interesting directions for future
work  to  harness the power of using the attribute structure for designing better
social network based systems and applications.

\section{Acknowledgments}
We would like to thank Mario Frank, our shepherd Ben Zhao and the anonymous reviewers for their insightful feedback. This work is supported by the NSF TRUST under Grant No. CCF-0424422, NSF Detection under Grant no. 0842695, by the AFOSR under MURI Award No. FA9550-09-1-0539, by the Office of Naval Research under MURI Grant No. N000140911081, the NSF Graduate Research Fellowship under Grant No. DGE-0946797, the DoD National Defense Science and Engineering Graduate Fellowship, by Intel through the ISTC for Secure Computing, and by a grant from the Amazon Web Services in Education program. Any opinions, findings, and conclusions or recommendations expressed in this material are those of the author(s) and do not necessarily reflect the views of the funding agencies.

\appendix

\begin{algorithm}[t!]
\small 
\DontPrintSemicolon 
\KwIn{$(SAN, \Omega, K)$, where $SAN=(V_s, V_a, E_s, E_a)$, $\Omega$ is the set of nodes whose average clustering coefficient $C_\Omega$ is approximated and $K$ is the number of samples needed.} 
\KwOut{Approximate average clustering coefficient $\tilde{C}_\Omega$.} 
\Begin{
$L \longleftarrow 0$\; 
$k \longleftarrow 0$\; 
\While{$k < K$}{
$k \longleftarrow k + 1$\;
Sample a node $u$ uniformly at random from $\Omega$ \;
Sample a pair of nodes $v$ and $w$ uniformly at random from $u's$ social neighbors $\Gamma_s(u)$ \;
$L \longleftarrow L + F(v, u, w)$\;  }
$\tilde{C}_\Omega \longleftarrow L/(2^{\idg}K)$ \;
} \caption{Constant Time Approximate Algorithm for Computing the Average Clustering Coefficient} 
\label{alg:alg-clustering-coefficient}
\end{algorithm}

\section{A Constant Time Algorithm for Approximating Clustering Coefficients}
\label{sec:clustering}
Before going to details of the algorithm and analysis, we introduce a few notations. In both directed or undirected \sanplural, a triple $t$ consists of three nodes $(v, u, w)$ satisfying $v, w \in \Gamma_s(u)$, where $u$ is called the center and $v, w$ are called the endpoints of $t$. 
Moreover, $\alpha_t$ and $\beta_t$ denote respectively the center 
node and the two endpoints of $t$. 
 
For a directed \san and a set of triples $T$, we define a mapping $F : T \longrightarrow \{0,1, 2\}$, where $F(t=(v,u,w)) = 0$ if $v$ and $w$ are not connected, $F(t=(v,u,w)) = 1$ if they are connected by one directed link and $F(t=(v,u,w)) = 2$ if they are reciprocally linked. For an undirected \san, the mapping is defined as $F : T \longrightarrow \{0, 1\}$, where $F(t=(v,u,w)) = 0$ if $v$ and $w$ are not connected, otherwise $F(t=(v,u,w)) = 1$. Let $\idg$ be an indicator variable of the directedness of a \san, where $\idg=0$ when the \san is undirected, otherwise $\idg=1$. With the indicator variable $\idg$, we have $0 \leq F(t) \leq 2^{\idg}$, which is useful for deriving the approximation bounds in the follows.

For any set of nodes $\Omega$, their average clustering coefficient can be represented as $C_\Omega = \frac{1}{|\Omega|}\sum_{u\in\Omega} c(u) = 2^{-\idg} \sum_{t\in T_\Omega} \frac{1}{|\Omega| \tau (\alpha_t)}F(t)$, where $T_\Omega = \{t| \alpha_t \in \Omega\}$ and $\tau(\alpha_t)=\frac{1}{2}|\Gamma_s(\alpha_t)|(|\Gamma_s(\alpha_t) | - 1)$ is the number of triples whose center node is $\alpha_t$. If $t$ is a uniformly distributed random variable over $\Omega$, then we have  $C_\Omega =2^{-\idg}E[F(t)]$. This observation informs us the design of our approximate algorithm, which is shown in Algorithm~\ref{alg:alg-clustering-coefficient}. Our algorithm computes the average social clustering coefficient when setting $\Omega = V_s$, and the average attribute clustering coefficient when setting $\Omega = V_a$. Note that our algorithm can also be used to compute average clustering coefficient distribution with respect to node degrees. The following theorem bounds the error of our algorithm.

\begin{theorem}
With the number of samples $K=\lceil \frac{\mathrm{ln}2\nu}{2\epsilon^2} \rceil$, the approximated average clustering coefficient $\tilde{C}_\Omega $ output by our algorithm satisfies $|\tilde{C}_\Omega - C_\Omega| \leq  \epsilon$ with probability at least $1-\frac{1}{\nu}$.
\end{theorem}

\begin{proof}
Assume $t_1, t_2, \cdots, t_K$  are $K$ independently and uniformly distributed random variables over the triple set $T_\Omega$. Then we have $C_\Omega=E[\frac{1}{2^{\idg} K}\sum_{i=1}^K F(t_i)]$. According to Hoeffding's bound~\cite{Hoeffding63}, we obtain 
$$Pr(|\frac{1} {K}\sum_{i=1}^K F(t_i) - E[\frac{1}{K}\sum_{i=1}^K F(t_i)]| \geq 2^{\idg}\epsilon) \leq 2e^{-2K\epsilon^2}.$$
Thus, $$Pr(|\tilde{C}_\Omega - C_\Omega| \leq \epsilon) \geq 1 - 2e^{-2K\epsilon^2}.$$
So we get $K=\lceil \frac{\mathrm{ln}2\nu}{2\epsilon^2} \rceil$ 
by setting $\nu=2e^{-2k\epsilon^2}$.
\end{proof}

\end{document}